{\let~\catcode~13 13\def^^M{^^J}~` 12~`@0~`\\12@xdef@asciiart{
                                                                                         
           ________  ___  ___  _______   ________  ___  __                               
          |\   ____\|\  \|\  \|\  ___ \ |\   ____\|\  \|\  \                             
          \ \  \___|\ \  \\\  \ \   __/|\ \  \___|\ \  \/  /|_                           
           \ \  \    \ \   __  \ \  \_|/_\ \  \    \ \   ___  \                          
            \ \  \____\ \  \ \  \ \  \_|\ \ \  \____\ \  \\ \  \                         
             \ \_______\ \__\ \__\ \_______\ \_______\ \__\\ \__\                        
              \|_______|\|__|\|__|\|_______|\|_______|\|__| \|__|                        
                                                                                         
  ________  _________  ________  ________  ________   ________  _______   ________       
 |\   ____\|\___   ___\\   __  \|\   __  \|\   ____\ |\   ____\|\  ___ \ |\   ___  \     
 \ \  \___|\|___ \  \_\ \  \|\  \ \  \|\  \ \  \___|_\ \  \___|\ \   __/|\ \  \\ \  \    
  \ \_____  \   \ \  \ \ \   _  _\ \   __  \ \_____  \\ \_____  \ \  \_|/_\ \  \\ \  \   
   \|____|\  \   \ \  \ \ \  \\  \\ \  \ \  \|____|\  \\|____|\  \ \  \_|\ \ \  \\ \  \  
     ____\_\  \   \ \__\ \ \__\\ _\\ \__\ \__\____\_\  \ ____\_\  \ \_______\ \__\\ \__\ 
    |\_________\   \|__|  \|__|\|__|\|__|\|__|\_________\\_________\|_______|\|__| \|__| 
    \|_________|                             \|_________\|_________|                     
                                                                                         
}}\makeatletter\newif\iflabor%
\labortrue

\documentclass{amsart}

\usepackage{mathtools,amssymb}\allowdisplaybreaks

\usepackage[english]{babel}
	\def\U@DECODE#1{\bgroup\def\UTFviii@defined##1{\expandafter#1\string##1+}}
	\def\U@LET#1:#2+{\egroup\U@DUC{\UTFviii@hexnumber{\decode@UTFviii#2\relax}}}
	\def\U@DEF#1:#2+{\U@LET:#2+{\@nameuse{意#2義}}\@namedef{意#2義}}
	\let\U@DUC\DeclareUnicodeCharacter\U@DUC{4EE4}{\U@DECODE\U@LET}令定{\U@DECODE\U@DEF}
	
	令Γ\Gamma	令Δ\Delta	令Θ\Theta	令Λ\Lambda	令Ξ\Xi		令Π\Pi		令Σ\Sigma	
	令Υ\Upsilon	令Φ\Phi		令Ψ\Psi		令Ω\Omega	令α\alpha	令β\beta		令γ\gamma	
	令δ\delta	令ε\varepsilon			令ζ\zeta		令η\eta		令θ\theta	令ι\iota		
	令κ\kappa	令λ\lambda	令μ\mu		令ν\nu		令ξ\xi		令π\pi		令ρ\rho		
	令ς\varsigma	令σ\sigma	令τ\tau		令υ\upsilon	令φ\varphi	令χ\chi		令ψ\psi		
	令ω\omega	令ϑ\vartheta	令ϕ\phi		令ϖ\varpi	令Ϝ\Digamma	令ϝ\digamma	令ϰ\varkappa	
	令ϱ\varrho	令ϴ\varTheta	令ϵ\epsilon	
	令𝔸{\mathbb A}	令𝔽{\mathbb F}	令ℚ{\mathbb Q}	
	令ℐ{\mathcal I}	
	令ℓ\ell		令∂\partial	令∇\nabla	
	令√\sqrt		
	\def\bigl@C#1{\bigl#1\iflabor\else彩\fi}	\def\bigr@C#1{\iflabor\else彩\fi\bigr#1}
	\def\({\bigl@C(}	\def\){\bigr@C)}	令（{\Bigl(}			令）{\Bigr)}			
	令［{\bigl@C[}		令］{\bigr@C]}		令「{\Bigl[}			令」{\Bigr]}			
	令｛{\bigl@C\{}		令｝{\bigr@C\}}		令『{\Bigl\{}		令』{\Bigr\}}		
	令⌈\lceil	令⌉\rceil	令⌊\lfloor	令⌋\rfloor	令⟨\langle	令⟩\rangle	
	\def\|{\mathrel\Vert}	令‖{\mathrel\Big\Vert}	令｜{\mid\nobreak}		
	令；{\mathrel;\nobreak}	令、{\setminus\nobreak}	令：{\colon}	
	令ˆ\hat		令˜\tilde	令¯\bar		令˘\breve	令˙\dot		令¨\ddot		令°\mathring	
	令ˇ{^\vee}	
	令∏\prod		令∑\sum		令∫\int		令⋀\bigwedge	令⋁\bigvee	令⋂\bigcap	令⋃\bigcup	
	令⨁\bigoplus			令⨂\bigotimes			
	令±\pm		令·\cdot		令×\times	令÷\frac		令•\bullet	令∘\circ		令∧\wedge	
	令∨\vee		令∩\cap		令∪\cup		令⊕\oplus	令⊗\otimes	令⋆\star	
	令¬\neg		令∀\forall	令∁\complement			令∃\exists	令∞\infty	令⊤\top		
	令⊥\bot		令⋯\cdots	令♠\spadesuit			令♡\heartsuit			
	令♢\diamondsuit			令♣\clubsuit	令♭\flat		令♮\natural	令♯\sharp	
	令←\gets		令→\to		令↞\twoheadleftarrow		令↠\twoheadrightarrow	令↤\mapsfrom	
	令↦\mapsto	令↩\hookleftarrow		令↪\hookrightarrow		
	令↾{\mathbin\upharpoonright}			令∈\in		令∉\notin	令∋\ni		令∼\sim		
	令≅\cong		令≈\approx	令≔\coloneqq	令≠\neq		令≡\equiv	令≤\leqslant	令≥\geqslant	
	令⊆\subseteq	令⋮{\smash\vdots\vphantom{|^{|^|}}}	令⋰\adots	令⋱\ddots	令⟂\perp		
	令⟵\longleftarrow		令⟶\longrightarrow		令⟼\longmapsto			
	令ß{\textit\ss\series}
	令…{,\penalty\binoppenalty\dotsc,\penalty\binoppenalty}	
	令，{,\penalty\binoppenalty}
	定†#1†{\text{#1}}			
	定©#1©{{\color{2728 C}#1}}	
	定®#1®{{\color{UCO}#1}}		
	\def\[{\@ifstar{\begin{equation*}\bgroup}{\begin{equation}\bgroup}}
	\def\]{\@ifstar{\egroup\end{equation*}}{\egroup\end{equation}}}
	\DeclarePairedDelimiter\abs\lvert\rvert	
	\DeclareMathOperator\ran{rank}			\DeclareMathOperator\spa{span}				
	\DeclareMathOperator\Mat{Mat}			\DeclareMathOperator\tr{tr}					
	\DeclareMathOperator\Cyc{Cyc}			
	
	\def\bma#1{\begin{bmatrix}#1\end{bmatrix}}
	\def\bsm#1{\left[\begin{smallmatrix}#1\end{smallmatrix}\right]}
	
	\let\AB\allowbreak
	\let\EA\expandafter
	\let\NE\noexpand
	\def\LHS/{left-hand side}
	\def\RHS/{right-hand side}
	定彈#1{\advance#1\glueexpr0ptplus1ptminus1pt}
	\AtBeginDocument{\hbadness99\overfullrule1em\advance\marginparwidth2.5em
		彈\baselineskip/8 彈\lineskip/8 彈\parskip/8}
	定墊{\iflabor\else{\color{427}\PMS\random@lineno{9*rnd}
		\foreach\\in{-9,...,\random@lineno}{\vskip1emplus2pt\hrule}}\fi}
	定邊#1{\def\@latex@warning@no@line##1{}\marginpar[彩\flushright#1]{彩\flushleft#1}}

\usepackage{tikz,tikz-cd,tikz-3dplot}
	定色#1!#2 {\definecolor{#1}{HTML}{#2}}% www.uillinois.edu/OUR/brand/color_palettes
	色UIB!13294b				色2728 C!0455A4			色2738 C!1F4096			% blues
	色427!E8E9EA				色Cool Gray 6!A5A8AA		色Cool Gray 1!5E6669		% neutrals
	色UCO!E84A27				色UIC red!D50032			色UIS blue!003366		% accents
	色Teal!0d605e			色Gray-blue!6fafc7		色Citron!bfd46d			% accents
	色Dark yellow!ffd125		色Salmon!ee5e5e			色Periwinkle!4f6898		% accents
	令彩{\PMS\R@R{  rnd/2}\PMS\G@G{  rnd/2}\PMS\B@B{  rnd/2}\color[rgb]{\R@R,\G@G,\B@B}}
	令虹{\PMS\R@R{1-rnd/2}\PMS\G@G{1-rnd/2}\PMS\B@B{1-rnd/2}\color[rgb]{\R@R,\G@G,\B@B}}
	\let\PMP\pgfmathparse		\def\PMR{\pgfmathresult}		\let\PMS\pgfmathsetmacro
	\let\PMT\pgfmathtruncatemacro		
	\def\MD@three#1#2#3#4.{\PMP{0x#1#2#3}}\expandafter\MD@three\pdfmdfivesum{\jobname}.
	\pgfmathsetseed{\PMR-\day-31*(\month-12*(\year-2000))}
	\usetikzlibrary{calc,decorations.markings}
	\tikzset{
		every picture/.style={cap=round,join=round},
		Tanner/.style={
			cn/.style={rectangle,minimum size=1.8em,draw},
			vn/.style={circle,minimum size=2em,inner sep=0,draw},
			hn/.style={vn,fill=UIB},
		}
	}
	\tikzcdset{}

\usepackage{pgfplotstable,booktabs,colortbl,multirow}
	\pgfplotsset{
		compat/show suggested version=false,compat=1.13, % arxiv = 1.13
		cycle multiindex* list={
			Dark yellow,Periwinkle,Citron,Gray-blue,Salmon,Teal,Cool Gray 6\nextlist% 7
			mark=*,mark=triangle,mark=diamond*,mark=square,mark=pentagon*,
			mark=o,mark=triangle*,mark=diamond,mark=square*,mark=pentagon\nextlist% 5x2
			,dash pattern=on16off.8,{dash pattern=on16off.8on1off.8}\nextlist
		},
		legend image code/.code={
			\draw[mark repeat=2,mark phase=2,#1]plot coordinates{
				(-5pt,-5pt)(0pt,0pt)(5pt,5pt)
			};
		},
		every axis/.style={
			width=13cm,height=8cm,const plot,
			cycle list shift=\accumulatenumplots,
		},
		before end axis/.code={
			\xdef\accumulatenumplots{\the\numexpr\accumulatenumplots+\numplots}
		}
	}
	\def\accumulatenumplots{0}
\usepackage{xurl}

\usepackage[unicode,pdfusetitle,pdfsubject=Information Theory (cs.IT)]{hyperref}
	\hypersetup{raiselinks,colorlinks,
		linkcolor=Salmon,citecolor=Teal,urlcolor=Periwinkle}
	\def\PM#1$#2${\texorpdfstring{$#2$}{#1}}
	\def\PT#1†#2†{\texorpdfstring{#2}{#1}}
	\def\U#1+{\unichar{"#1}}

\usepackage[noabbrev]{cleveref}
	
	定名#1:#2~#3?#4 {\crefname{#1}{#2#3}{#2#4}}				% AMS style guide
	名page:page~?s				名section:section~?s			名enumi:item~?s				
	定理#1:#2~#3?#4 {\newtheorem{#1}[amsthmall]{#2#3}名#1:#2~#3?#4 }
	令風\theoremstyle				理thm:Theorem~?s				
	理cor:Corollar~y?ies			理lem:Lemma~?s				理pro:Proposition~?s			
	風{definition}				理dfn:Definition~?s			理exa:Example~?s				
	風{remark}					理cla:Claim~?s				理rem:Remark~?s				
	定式#1:#2~#3?#4 {名#1:#2~#3?#4 \creflabelformat{#1}{\textup{(##2##1##3)}}}
	式equ:equation~?s			式ine:inequalit~y?ies		式for:formula~?s				
	式mat:matri~x?ces			式sta:star~?s				式ten:tensor~?s				
	\def\@ReadTypeLabel#1:#2?{\xdef\RTL@TYPE{#1}\xdef\RTL@LABEL{#2}}
	\def\eqlabel#1{\@ReadTypeLabel#1?\label[\RTL@TYPE]{\RTL@TYPE:\RTL@LABEL}}
	\def\steplabel{\incr@eqnum\tag\theequation\eqlabel}

	\def\block#1#2{_{#1#2}}
	\def\blockblock#1#2#3{{#1\block{#2}{#3}\if@display\else}{\fi}\block}
	\def\series#1{^{(#1)}}
	\def\seriesseries#1#2{{#1\series{#2}\mkern-2mu\if@display\else}{\fi}\series}
	\def\seriesseriesseries#1#2#3{
		{\seriesseries{#1}{#2}{#3}\mkern-2mu\if@display\else}{\fi}\series}
	\def\A{A\block}\def\AA{\blockblock A}
	\def\B{B\block}\def\BB{\blockblock B}
	\def\C{C\block}
	\def\J{J\block}
	\def\S{S\series}\def\SS{\seriesseries S}
	\def\L{L\series}
	\def\al#1{α\series{#1}\block}
	\def\be#1{β\series{#1}\block}

定行{\catcode13 }行13定讀#1{行13\def^^M{\\}\def\READ@ATTR{行5#1}\@ifnextchar[{選}{參}}
定選[^^M#1 ^^M]^^M#2 ^^M^^M{\READ@ATTR[#1]{#2}}定參^^M#1 ^^M^^M{\READ@ATTR{#1}}行5

讀
\title
				            Parity-Checked Strassen Algorithm

讀
\author
				              Hsin-Po Wang\and Iwan Duursma

讀{\def
\@pdfsubject}
				      94B05; 68W10, 05B35; math.IT; cs.DS; math.CO

讀
\subjclass
				          Primary 94B05, 68W10; Secondary 05B35

\address{
				               Department of Mathematics,               				
				      University of Illinois at Urbana--Champaign,      				
				                 Urbana, Illinois 61801                 				
}

\email{%
				            hpwang2 and duursma @illinois.edu           				
}

\begin{document}\makeatother\message{\asciiart}

\begin{abstract}
	To multiply astronomic matrices using parallel workers subject to straggling,
	we recommend interleaving checksums with some fast matrix multiplication algorithms.
	Nesting the parity-checked algorithms, we weave a product code flavor protection.
	
	Two demonstrative configurations are as follows:
	(A)	$9$ workers multiply two $2×2$ matrices;
		each worker multiplies two linear combinations of entries therein.
		Then the entry products sent from any $8$ workers
		suffice to assemble the matrix product.
	(B)	$754$ workers multiply two $9×9$ matrices.
		With empirical frequency $99.8\%$, $729$ workers suffice,
		wherein $729$ is the complexity of the schoolbook algorithm.
	
	In general, we propose probability-wisely favorable configurations
	whose numbers of workers are close to, if not less than,
	the thresholds of other codes (e.g., entangled polynomial code and PolyDot code).
	Our proposed scheme applies recursively, respects worker locality,
	incurs moderate pre- and post-processes, and extends over small finite fields.
\end{abstract}

\maketitle

\section{Introduction}

	\emph{Matrix multiplication} is a pivotal operation
	in all of engineering, science, and mathematics.
	As problems and challenges in these areas grow in magnitude,
	so do the underlying matrices.
	The demand is to multiply matrices speedily, securely, and secretly.
	
	With today's technology, it is possible to assign a difficult computational job%
	---be it matrix multiplication or other---to numerous interconnected devices.
	The head device, called the \emph{manager}, pre-processes the job and allocates
	the workloads to some other devices in the network, called the \emph{workers}.
	Each worker works out the assigned task in parallel,
	and submits the outcome back to the manager as soon as it finishes.
	The manager gathers those answers, post-processes, and declares completion.
	
	Workers encounter turbulence in their workflow.
	It is not uncommon to see that a minority of workers, termed \emph{stragglers},
	return their answers significantly later than the majority.
	Idling, the manager as well as the punctual workers
	have no choice but to wait for the stragglers to catch up.
	
	To counter the turbulence, a thoughtful manager designs and distributes
	redundant tasks to workers such that it is able to
	assemble the final output from but a portion of worker answers.
	A plain tactic is to assign repetitive tasks to multiple workers,
	and hope that they do not straggle all at once.
	As a motivating example, say there are $100$ workers and $10$ tasks.
	For the $i$th task, for each $1≤i≤10$,
	assign the worker labeled $(i,j)$, for all $1≤j≤10$, to do the task.
	let $T_{(i,j)}$ be the time the worker spends.
	Then tasks are all done at time
	\[*\max_{1≤i≤10}\min_{1≤j≤10}T_{(i,j)}.\]*
	For simple, unstructured jobs,
	repetition seems to be the only remedy, and is found good enough in practice.
	
	For linear operations---for instance matrix-vector multiplication---%
	the methodology of coding comes into play.
	Unlike repetition, a linear combination like $A_1+A_2$
	serves as a backup of $A_1$ when $A_2$ is known,
	and serves as a backup of $A_2$ when $A_1$ is known.
	As a result of linearity, any two out of $A_1x$, $A_2x$, and $(A_1+A_2)x$
	lead to the final matrix-vector product
	\[*\bma{A_1\\A_2}x=\bma{A_1x\\A_2x}
		=\bma{1\\-1&1}\bma{A_1x\\(A_1+A_2)x}
		=\bma{1&-1\\&1}\bma{(A_1+A_2)x\\A_2x}.\]*
	In formal notation, this is the $[3,2,2]$-Hamming code
	whose generator matrix is $[\mkern1.5mu^1_0{}^0_1{}^1_1]^⊤$.
	
	In general, to obtain $Ax$ when $A$ is about a square matrix and
	$x$ is a thin matrix, the manager adds redundancy to $A$ as follows:
	Prepare a tall generator matrix $G$ to form $˜A≔GA$.
	Then, send the $i$th row of $˜A$ to the $i$th worker
	and ask the latter to evaluate the dot product with $x$.
	The manager collects the dot products from the workers
	and declares completion if it finds the corresponding rows of $G$ full-rank.
	To be precise, let $P_ℐ$ be the restriction to the rows indexed by $ℐ$.
	Then $P_ℐ(˜Ax)=(P_ℐG)Ax$ is the collection of the dot products
	the manager receives, where $ℐ$ is the indices of the workers that return.
	Consequently, $(P_ℐG)^{-1}P_ℐ(˜Ax)=Ax$
	whenever a left inverse of $P_ℐG$ is defined.
	Within this framework, the overall time spent at the workers is
	\[*\min_ℐ\max_{i∈ℐ}T_i.\]*
	$T_i$ is the time the $i$th worker spends.
	The minimum is over all $ℐ$ that make $P_ℐG$ easy to invert.
	
	In reality, preparing a ``good'' $G$
	is subject to various contradictory constraints.
	For one, we want $P_ℐG$ invertible for as many $ℐ$'s as possible.
	The optimality is reached when $G$ is a generator matrix of an MDS code,
	and is approximated when $G$ is a random matrix.
	In \cite{LLPPR18}, the MDS approach was considered.
	The authors estimated the (theoretically) best worker time
	\[*\min_{\abs{ℐ}=m}\max_{i∈ℐ}T_i
		=†the $m$th least quantity among †T_1,T_2…T_{†\#workers†},\]*
	where $m$ is the width of $A$, the height of $x$.
	And they compared it with the naïve worker time $\max_{1≤i≤m}T_i$.
	
	Whereas the worker time is minimized, other costs emerge.
	In \cite{DCG16}, it was addressed that
	executing long dot products is prone to errors.
	Whence, they further divide the rows of $˜A≔GA$ into smaller segments.
	Workers now execute short dot products, hence the name \emph{Short-Dot code}.
	
	Other costs include the burden to invert $P_ℐG$ when $G$ is MDS or random.
	To that end, polar coding was applied in \cite{BP19}.
	Since polar code achieves capacity (especially over erasure channels),
	the manager waits for ``a little bit more than $m$ workers'' to respond.
	In exchange for the ``little bit more'',
	the decoding complexity at the manager side is vastly reduced.
	This speeds up the overall job.
	
	As if capacity-achieving is not good enough, \cite{MCSPJ20} proposed a $G$
	based on LT (Luby transform) code, a rateless code with low decoding complexity.
	With the rateless property, the manager need not worry about
	the capacity of the workers and the rate of the code to be used.
	Rather, it simply issues a new dot product whenever a worker reports availability
	(after the latter finishes the last task or reboots).
	Like in the polar code case, the manager declares completion when
	the number of responses reaches a pre-set bar that is slightly greater than $n$.
	
	Other concerns in this corner include
	assigning heavier tasks to more powerful workers.
	This is done by sending unequal numbers of rows of $˜A≔GA$
	to heterogeneous workers \cite{RPPA19}.
	Yet another concern is the communication cost \cite{SGR19}.
	Beyond this point, the marginal gain in performance seems to come from
	the understanding of the context (e.g., the MapReduce framework or the hardware).
	In this paper, however, we do not touch these aspects of computation.

\subsection{Coded matrix multiplication}

	In calculating matrix-matrix multiplication $AB=C$,
	the bilinear nature is exploited to help defeat stragglers.
	On top of adding redundancies to $A$, one may add redundancies to $B$.
	That is, choose generator matrices $G,H$ to form $˜A≔GA$ and $˜B≔BH$.
	Broadcast rows of $˜A$ and columns of $˜B$ to workers,
	and let them evaluate the dot products.
	The manager collects the dot products, and
	``fills in the blanks'' (i.e., to infer the unknowns) whenever possible.
	
	Two possible rules to derive a blank dot product are available:
	Vertically, $G$ introduces relations among the entries in the same column of $˜A˜B$.
	The relation reads $G^⟂(˜A˜B)=0$ for any $G^⟂$ that
	annihilates $G$ from the left (i.e., $G^⟂G=0$).
	Horizontally, $H$ introduces relations among the entries in the same row of $˜A˜B$.
	The relation reads $(˜A˜B)H^⟂=0$ for all $H^⟂$ such that $HH^⟂=0$.
	The manager applies the two Sudoku-like rules
	alternatively to maximize the recovery.
	
	This concept that a data matrix $M∈𝔽^{ℓ×n}$ has annihilators
	(or \emph{parity checks}) from both sides is not new;
	the debut dates back to Elias \cite{Elias54}.
	In fact, Elias's construction starts with a multi-dimensional array
	and injects redundancies along every cardinal direction.
	This construction is named \emph{product codes}
	or \emph{turbo product codes} (TPC).
	For more on this topic,
	see a review paper \cite{MAS16} with a three-figure list of references.
	Note that its dual construction is called, equivocally,
	product codes or tensor product codes \cite{Wolf06}.
	In this paper, we only refer to the former, primary product codes,
	not to the dual constructions.
	
	Back to distributed computation.
	Usage of TPC in matrix-matrix multiplication is first seen
	in the context of algorithm-based fault tolerance \cite{HA84}
	and later in machine-learning--driven works.
	We overview the latter below.
	
	Same as in the matrix-vector case, the hardness in the matrix-matrix case
	is how to design $G$ and $H$ to optimize over numerous considerations.
	In \cite{LSR17}, for instance, both $G$ and $H$ are MDS;
	and the authors estimated the asymptote of the worker time
	under a probabilistic model.
	
	Knowing that MDS codes are hard to decode, \cite{BLOR18} suggested
	tensor powers of the $[3,2,2]$-Hamming code---$[\mkern1.5mu^1_0{}^0_1{}^1_1]$%
	---playing as $G^⊤$ and as $H$.
	This gives the redundant product $˜A˜B$ a layout of
	multi-dimensional TPC, in lieu of being two-dimensional.
	The pros of this proposal is that all parity checks are utterly short, hence swift.
	The cons, on the contrary, is the low code rate.
	Our work shares common elements with this one and faces a similar trade-off.
	Incidentally, the impact on the code rate is mild in our case and
	we are able to spot a niche where our rate outruns others'.
	
	In \cite{WLS18}, the authors copped with the complaint that
	$G,H$ originating with good codes (especially MDS and random)
	tend to destroy the sparsity of a matrix.
	The authors recommended making $G,H$ sparse, in a way that
	performs acceptably well while not creating too many nonzero entries.
	
	\cite{GWCR18} commented that a sparse matrix usually undergoes
	a process so-called \emph{sketching} to reduce its dimension;
	think of this as compressing (source-coding) a matrix
	while keeping its ``meat'' as faithfully as possible.
	The authors then suggested \emph{over}-sketching a matrix to retain some of
	its redundancies, and turning the last bit of redundancies against the stragglers.
	This tactic is thus named \hbox{\emph{OverSketch}}.
	
	In the next subsection, we dive into other codes
	that fall outside the $(GA)(BH)$ vein.
	They are the codes we want to compete with.

\subsection{Beyond bilinearity}

	It has come to a stage where the manager starts taking into account
	that $AB$ is a product of matrices, not of integers or of polynomials.
	Whilst the $(GA)(BH)$ framework exploits the fact that
	$A$ consists of rows and $B$ consists of columns,
	one can also see $A$ as a list of columns and $B$ as a list of rows.
	
	As part of the new perspective,
	\emph{polynomial code} in \cite{YMA17,YMA20} can be paraphrased as below:
	The manager prepares for each worker $s$ a row vector $g_s∈𝔽^{1×ℓ}$
	and a column vector $h_s∈𝔽^{n×1}$, and asks it to compute $(g_sA)(Bh_s)$.
	The manager can decode $AB$ if the span of $h_sg_s$
	over the punctual workers is the space of all matrices.
	This criterion is derived from $(g_sA)(Bh_s)=\tr(g_sABh_s)=\tr((h_sg_s)(AB))$.
	Since $\tr(¯X^⊤Y)$ is the Frobenius inner product of matrices $X$ and $Y$,
	one can decode $Y$ (in this case, $AB$) if
	a set of $X$'s (in this case, $h_sg_s$) span the space of matrices.
	To accomplish the best result,
	the authors let $ζ_s∈𝔽$ be distinct scalars, and let
	\[*g_s≔\bma{1&ζ_s&⋯&ζ_s^{ℓ-1\vphantom|}},\qquad
		h_s≔\bma{1&ζ_s^ℓ&⋯&ζ_s^{(n-1)ℓ}}^⊤.\]*
	Here, the height of $A$ is $ℓ$ and the width of $B$ is $n$.
	As $h_sg_s$ consists of contiguous powers of $ζ_s$,
	any subset set of size $ℓn$ span the space of the matrices.
	Hence the recovery threshold of polynomial coding is said to be $ℓn$.
	
	On the same track,
	\emph{MatDot code} in \cite{DFHJCG20}
	(which shares the same design elements with coded convolution
	in \cite[Theorem~4]{YMA17})
	can be paraphrased as below:
	The manager prepares for each worker $s$ a column vector $x_s∈𝔽^{m×1}$
	and a row vector $y_s∈𝔽^{1×m}$, and asks it to compute $(Ax_s)(y_sB)$.
	The manager can recover $C≔AB$ if the identity matrix is
	in the span of $x_sy_s$ over the punctual workers.
	To this end, the authors let $ζ_s∈𝔽$ be distinct scalars, and let
	\[*x_s≔\bma{1&ζ_s&⋯&ζ_s^{n-1}}^⊤,\qquad y_s≔\bma{ζ_s^{n-1}&⋯&ζ_s&1}.\]*
	This implies that each $x_sy_s$ is a Toeplitz matrix, and any $2n-1$ many $x_sy_s$
	span the space of Toeplitz matrices, in which the identity matrix lies.
	The number $2n-1$ is referred to as
	the \emph{recovery threshold} of MatDot coding.
	
	Polynomial code and MatDot code generalize.
	\emph{PolyDot code} (PDC) in \cite{DFHJCG20}
	lets worker $s$ compute $(g_sAx_s)(y_sBh_s)$ for some
	vectors $g_s,h_s,x_s,y_s$ that contain powers of $ζ_s$.
	When $ℓ$ (the height of $A$) coincides with $n$ (the width of $B$),
	they are able to achieve a threshold of $ℓ(2m-1)n$.
	Likewise, \emph{entangled polynomial code} (EPC) in \cite{YMAS18,YMA20}
	assigns worker $s$ to compute a similar product $(g_sAx_s)(y_sBh_s)$ with
	$g_s,h_s,x_s,y_s$ consisting of a different set of powers of $ζ_s$.
	The threshold they achieve is $ℓmn+m-1$,
	which is about half of PDC's threshold.
	
	When $m=1$, both PDC and EPC degenerate to polynomial coding,
	and the thresholds become $ℓ(2m-1)n=ℓmn+m-1=ℓn$.
	This quantity happens to be the bilinear rank (in fact, the border rank)
	of matrix multiplication of type $⟨ℓ,1,n⟩$ \cite[Lemma~7.1]{Blaser13}.
	That is to say, the threshold $ℓn$ is optimal in this $m=1$ case.
	When $m≥2$, especially when $ℓ=m=n≥2$,
	EPC's threshold $ℓmn+m-1$ is better than PDC's $ℓ(2m-1)n$.
	Meanwhile, the bilinear rank is subcubic---%
	$\ran⟨m,m,m⟩≈m^ω$ for some $ω<2.3729$ \cite{LeGall14,AW20}.
	This evinces some room for improvement.
	For instance, \cite[section~VI]{YMA20} provided a different scheme
	that has threshold $2R-1$, where $R$ is the bilinear rank.
	
	This work starts from some fast matrix multiplications,
	inserts redundancies into those algorithms,
	and then concatenates them like a TPC does to its constituent codes.
	As long as our strategy stays close enough to the subcubic complexity,
	the threshold dominates the cubic ones of PDC and EPC,
	and is competitive against the $2R-1$ construction for small matrix sizes.
	Even if our thresholds exceed the competitors', there is a chance
	that our manager declares completion before theirs.
	This reflects the fact that TPC, even though having a bad minimum distance,
	performs fairly well in the stochastic, chaotic wild.

\subsection{Complexity convention}

	Throughout this work, the matrices $A$ and $B$ are thought to have entries
	coming from an algebra $𝔸$ over a large enough field $𝔽$---%
	for instance, $𝔽$ is complex numbers and $𝔸$ is $1{,}024$-by-$1{,}024$ matrices.
	We presume that storage, communication, logic, arithmetic in $𝔽$,
	addition in $𝔸$, and scalar multiplication $𝔽×𝔸→𝔸$ are negligibly cheap;
	and focus on the number of entry multiplications $𝔸×𝔸→𝔸$
	as the measurement of complexity.
	Unless stated otherwise,
	one worker carries out one entry multiplication and replies.
	Other cheap operations can be executed by either the manager or the workers.

\subsection{Organization of the paper}

	%The remaining of the paper is organized as follows.
	\Cref{sec:fast} briefs Strassen's fast multiplication.
	\Cref{sec:check} brings in the idea about
		how to inject checksums into fast multiplication.
	\Cref{sec:Laderman} presents Laderman's instruction
		for $3×3$ matrices followed by checksums.
	\Cref{sec:matroid} elaborates on the relation
		of error-correction to the matroid structure.
	\Cref{sec:recurse} concatenates checked algorithms
		and elaborates why it boosts resiliency.
	\Cref{sec:unify} unionizes the concatenations
		to protect better then the sum of its parts.
	\Cref{sec:sample} contains plots for probabilistic analysis.

\section{Strassen the Fast Matrix Multiplication}\label{sec:fast}

	The main theme of the present paper is to calculate the matrix-matrix multiplication
	\[*
		\bma{\A11&\A12\\\A21&\A22}⋆
		\bma{\B11&\B12\\\B21&\B22}=
		\bma{\C11&\C12\\\C21&\C22}
	\]*
	defined by
	\[*\C ik≔∑_{j=1}^2\A ij\B jk.\]*
	Hereafter, both juxtaposition and star $⋆$ mean multiplication.
	We put $⋆$ at the ones we want to highlight (usually the calculation bottleneck).
	
	Let $\J11,\J12,\J21,\J22$ be the matrices
	with a $1$ at the specified position, and zeros elsewhere;
	they form a basis of $\Mat(2,2)$, the vector space of $2$-by-$2$ matrices.
	Let $\J11ˇ,\J12ˇ,\J21ˇ,\J22ˇ$ be the dual basis.
	Then, the matrix multiplication as a linear map
	\[*⋆：\Mat(2,2)×\Mat(2,2)⟶\Mat(2,2)\]*
	lifts to the \emph{multiplication tensor}
	\begin{align*}
		⟨2,2,2⟩	&	∈\Mat(2,2)ˇ⊗\Mat(2,2)ˇ⊗\Mat(2,2),		\\
		⟨2,2,2⟩	&	≔(\J11ˇ⊗\J11ˇ+\J12ˇ⊗\J21ˇ)⊗\J11			
						+(\J11ˇ⊗\J12ˇ+\J12ˇ⊗\J22ˇ)⊗\J12		\\*
				&	\qquad+(\J21ˇ⊗\J11ˇ+\J22ˇ⊗\J21ˇ)⊗\J21	
						+(\J21ˇ⊗\J12ˇ+\J22ˇ⊗\J22ˇ)⊗\J22.	
	\end{align*}
	For general parameters, the multiplication tensor is
	\[*⟨ℓ,m,n⟩≔∑_{i=1}^ℓ∑_{j=1}^m∑_{k=1}^n\J ijˇ⊗\J jkˇ⊗\J ik
		∈\Mat(ℓ,m)ˇ⊗\Mat(m,n)ˇ⊗\Mat(ℓ,n).\]*
	
	The question that is deeply connected to the computational complexity is,
	What is the rank of $⟨ℓ,m,n⟩$?
	The \emph{rank} is the minimum number $r$ such that
	the multiplication tensor can be expressed as
	\[⟨ℓ,m,n⟩=∑_{s=1}^rX_sˇ⊗Y_sˇ⊗Z_s\eqlabel{for:tensor}\]
	for some $(X_sˇ,Y_sˇ,Z_s)∈\Mat(ℓ,m)ˇ×\Mat(m,n)ˇ×\Mat(ℓ,n)$ for each $1≤s≤r$.
	The rank $r$ is a measure for the number of entry multiplications
	an algorithm needs to evaluates $A⋆B$.
	The less the rank is, the faster the algorithm evaluates.
	From the definition, it is clear that $\ran⟨2,2,2⟩≤8$, or $\ran⟨ℓ,m,n⟩≤ℓmn$.
	The inequality is, surprisingly, not tight.

\subsection{An improvement}

	A manager hires seven workers and requests them to work on
	these seven entry multiplications (one line for each worker):
	\begin{alignat*}2
		\S1		≔&&		(\A11+\A22)		&⋆	(\B11+\B22),	\steplabel{for:Str1}\\
		\S2		≔&&		(\A21+\A22)		&⋆		\B11,		\\
		\S3		≔&&			\A11		&⋆	(\B12-\B22),	\\
		\S4		≔&&			\A22		&⋆	(-\B11+\B21),	\\
		\S5		≔&&		(\A11+\A12)		&⋆		\B22,		\\
		\S6		≔&&		(-\A11+\A21)	&⋆	(\B11+\B12),	\\
		\S7		≔&\;&	(\A12-\A22)		&⋆	(\B21+\B22).	\steplabel{for:Str7}
	\end{alignat*}
	The next theorem by V.~Strassen marked
	the epoch of fast matrix multiplication (FMM).
	
	\begin{thm}\label{thm:Strassen}
		\cite{Strassen69}
		$\S1,\S2…\S7$ determine $\C11,\C12,\C21,\C22$.
	\end{thm}
	
	\begin{proof}
		Verify, by the definitions,
		\def\cdot{\phantom{{}+\S0}}
		\begin{align*}
				\steplabel{for:C11}
			\C11	&	=\S1	·		·		+\S4	-\S5	·		+\S7,	\\
			\C12	&	=·		·		\S3		·		+\S5	·		·,		\\
			\C21	&	=·		\S2		·		+\S4	·		·		·,		\\
			\C22	&	=\S1	-\S2	+\S3	·		·		+\S6	·.	
				\steplabel{for:C22}
		\end{align*}
		See \cite{GM17} and its references for more intuitions.
		\Crefrange{for:Str1}{for:Str7} and \crefrange{for:C11}{for:C22}
		together are referred to as the \emph{Strassen algorithm}.
	\end{proof}
	
	Remark:
	How the intermediate variables---the $S$'s---are defined
	gives rise to these \emph{algorithmic tensors} in $\Mat(2,2)ˇ⊗\Mat(2,2)ˇ$:
	\begin{alignat*}2
		⟨\S1⟩	≔&&		(\J11ˇ+\J22ˇ)	&⊗	(\J11ˇ+\J22ˇ),	\steplabel{ten:Str1}\\
		⟨\S2⟩	≔&&		(\J21ˇ+\J22ˇ)	&⊗		\J11ˇ,		\\
		⟨\S3⟩	≔&&			\J11ˇ		&⊗	(\J12ˇ-\J22ˇ),	\\
		⟨\S4⟩	≔&&			\J22ˇ		&⊗	(-\J11ˇ+\J21ˇ),	\\
		⟨\S5⟩	≔&&		(\J11ˇ+\J12ˇ)	&⊗		\J22ˇ,		\\
		⟨\S6⟩	≔&&		(-\J11ˇ+\J21ˇ)	&⊗	(\J11ˇ+\J12ˇ),	\\
		⟨\S7⟩	≔&\;&	(\J12ˇ-\J22ˇ)	&⊗	(\J21ˇ+\J22ˇ).	\steplabel{ten:Str7}
	\end{alignat*}
	Therefore, \cref{thm:Strassen} is equivalent to the next corollary.
	
	\begin{cor}
		$⟨2,2,2⟩$ possesses rank $7$ or less.
		More precisely,
		\[*⟨2,2,2⟩∈\spa(⟨\S1⟩,⟨\S2⟩…⟨\S7⟩)⊗\Mat(2,2).\]*
	\end{cor}
	
	\begin{proof}
		$⟨2,2,2⟩$, the multiplication tensor in $\Mat(2,2)ˇ⊗\Mat(2,2)ˇ⊗\Mat(2,2)$, is
		\begin{align*}
			&	(⟨\S1⟩+⟨\S4⟩-⟨\S5⟩+⟨\S7⟩)⊗\J11	+(⟨\S3⟩+⟨\S5⟩)⊗\J12					\\*
			&	\qquad+(⟨\S2⟩+⟨\S4⟩)⊗\J21		+(⟨\S1⟩+⟨\S3⟩-⟨\S2⟩+⟨\S6⟩)⊗\J22,	
		\end{align*}
		which is equal to 
		\begin{align*}
			&	⟨\S1⟩⊗(\J11+\J22)+⟨\S2⟩⊗(\J21-\J22)+⟨\S3⟩⊗(\J12-\J22)				\\*
			&	\qquad+⟨\S4⟩⊗(\J11+\J21ˇ)ˇ+⟨\S5⟩⊗(\J12-\J11)+⟨\S6⟩⊗\J22+⟨\S7⟩⊗\J11.	
		\end{align*}
		Now $⟨\S1⟩⊗(\J11+\J22)$ and the other six summands are of the form
		$X_sˇ⊗Y_sˇ⊗Z_s$ as in \cref{for:tensor}, hence the upper bound of $7$.
		
		Remark:
		It turns out that $7$ is tight \cite{Pan81}.
		Another remark:
		Traditionally, it is $⟨\S1⟩⊗(\J11+\J22)$ that is defined to be
		an algorithmic tensor that constitutes the Strassen algorithm,
		and this theorem should have been stated as
		``$⟨2,2,2⟩∈$ the span of the seven algorithmic tensors''
		without the ``${}⊗\Mat(2,2)$'' part.
		The remark at the end of \cref{pro:8outof9} explains why we choose not to.
	\end{proof}
	
	To sum up, seven entry products control the product of two $2$-by-$2$ matrices.
	The problem is, the workers that will be calculating $\S1$--$\S7$
	may cease to respond, hence the necessity to build
	a code around the $S$'s to ensure erasure-resiliency.
	As pointed out in the introduction,
	repetition is sufficient but inefficient.
	Coded multiplication is shown to be superior,
	and is the solution applied in the sequel.

\section{Strengthen for Distributed Computation}\label{sec:check}

	In this section, we define some ``random'' linear combinations of $\C ij$.
	These coefficients ought to be large, random numbers in $𝔽$,
	but we fix some small realizations to work with.
	When those combinations are linearly independent of
	Strassen's blueprint---the $S$'s---they become promising checksums
	that help the manager derive any missing datum.
	
	We will add three checksums to the Strassen algorithm, one at a time.

\subsection{First checksum}

	The manager makes two extra workers compute,
	along with $\S1$--$\S7$, these as backups
	\begin{alignat*}2
		\S8		≔&&		(\A11+2\A21)	&⋆	(-\B11+\B12),	\\
		\S9		≔&\;&	(\A12+2\A22)	&⋆	(-\B21+\B22).	
	\end{alignat*}
	
	We can now make our foremost contribution.
	
	\begin{pro}\label{pro:8outof9}
		Symbols $\S1,\S2…\S9$ form a $[9,8,2]$-MDS code.
		Any eight out of the nine determine $\C11,\C12,\C21,\C22$.
		To put it another way, one of the nine workers may straggle/crash but
		the manager will still manage to assemble $C$ from the other workers' responses.
	\end{pro}
	
	\begin{proof}
		Expand the following associativity equation
		\[*\bma{1&2}C\bma{-1\\1}=\(\bma{1&2}A\)⋆\left(B\bma{-1\\1}\right).\]*
		The \RHS/ is
		\[*\bma{\A11+2\A21&\A12+2\A22}⋆\bma{-\B11+\B12\\-\B21+\B22},\]*
		which is an inner product whose first summand/entry multiplication is $\S8$,
		and the second is $\S9$.
		In the meantime, the \LHS/ of the equation is $-\C11+\C12-2\C21+2\C22$,
		where each $\C ij$ is a linear combination of $\S1$--$\S7$.
		Expand the \LHS/ down to the $S$'s;
		they induce a parity-check equation
		\[\S1	-4\S2	+3\S3	-3\S4	+2\S5	+2\S6	-\S7
			=\S8+\S9.\eqlabel{equ:Q89}\]
		This yields a single parity-check code among symbols $\S1$--$\S9$.
		Therefore, any missing symbol can be restored by the rest.
		And then, the manager assimilates $\C11$ to $\C22$ via
		the second half of the Strassen algorithm (\crefrange{for:C11}{for:C22}).
		
		Remark:
		The proposition is equivalent to saying
		\[*⟨2,2,2⟩∈\spa\(†any eight out of †⟨\S1⟩,⟨\S2⟩…⟨\S9⟩\)⊗\Mat(2,2)\]*
		where $⟨\S8⟩,⟨\S9⟩$ are defined in the way \crefrange{ten:Str1}{ten:Str7} are.
	\end{proof}
	
	We name this scheme \emph{Pluto code},
	after the International Astronomical Union reclassified Pluto as a dwarf planet,
	lowering the number of planets in the Solar System from nine to eight.
	
	Pluto code can go one step further;
	make use of two additional workers to oppose to one more erasure.

\subsection{Second checksum}

	Have two more workers compute, besides $\S1$--$\S9$,
	\begin{alignat*}2
		\S{10}	≔&&		(3\A11-\A21)	&⋆	(\B11+2\B12),	\\
		\S{11}	≔&\;&	(3\A12-\A22)	&⋆	(\B21+2\B22).	
	\end{alignat*}
	The upcoming proposition explains the resiliency provided by the first four backups.
	
	\begin{pro}\label{pro:9outof11}
		Any nine out of $\S1,\S2…\S{11}$ determine $\C11,\C12,\C21,\C22$.
		To put it another way, two of the eleven computations
		may fail without endangering the overall result.
	\end{pro}
	
	\begin{proof}
		First of all, expand a modified equation
		\[*\bma{3&-1}C\bma{1\\2}=\(\bma{3&-1}A\)⋆\left(B\bma{1\\2}\right).\]*
		The \LHS/ is $3\C11+6\C12-\C21-2\C22$, the \RHS/ being
		\[*\bma{3\A11-\A21&3\A12-\A22}⋆\bma{\B11+2\B12\\\B21+2\B22}.\]*
		They induce an equation
		\[\S1	+\S2	+4\S3	+2\S4	+3\S5	-2\S6	+3\S7
			=\S{10}+\S{11}.\eqlabel{equ:Qab}\]
		Consequence:
		$\S1$--$\S7$, $\S{10}$, and $\S{11}$ form a single parity-check code.
		
		Now look at the parity-check matrix for $\S1$--$\S{11}$:
		\[\setcounter{MaxMatrixCols}{11}\bma{
			1	&	-4	&	3	&	-3	&	2	&	2	&	-1	&-1	&-1	&	&	\\
			1	&	1	&	4	&	2	&	3	&	-2	&	3	&	&	&-1	&-1	
		}\eqlabel{mat:H11}\]
		The first row corresponds to \cref{equ:Q89}, the second to \eqref{equ:Qab}.
		It is clear that this matrix does not check an MDS code.
		That being said, one sees that
		all $2$-by-$2$ minors of the first seven columns are nonzero.
		So this code fixes two erasures in the first seven symbols%
		---the symbols we care about.
		Once $\S1$--$\S7$ are secured, continue the original Strassen algorithm.
		
		Takeaway:
		We do not need the individual values of $\S8$ and $\S9$, but their sum.
		The same role is played by $\S{10}$ and $\S{11}$.
		Remark:
		This proposition is equivalent to saying
		\[*⟨2,2,2⟩∈\spa\(†any nine out of †⟨\S1⟩,⟨\S2⟩…⟨\S{11}⟩\)⊗\Mat(2,2)\]*
		where $⟨\S{10}⟩,⟨\S{11}⟩$ are defined in a similar manner.
		We will not use the tensor language in the remaining of the paper.
		But it is worth keeping in mind that the determination of $\C ij$
		is nothing more than the membership of the corresponding tensors.
	\end{proof}
	
	\begin{cor}\label{cor:MDS2}
		These nine symbols $\S1,\S2…\S7,(\S8+\S9),(\S{10}+\S{11})$
		amount to a $[9,7,3]$-MDS code.
	\end{cor}
	
	\begin{proof}
		The parity-check matrix
		\[*\bma{
			1	&	-4	&	3	&	-3	&	2	&	2	&	-1	&-1	&	\\
			1	&	1	&	4	&	2	&	3	&	-2	&	3	&	&-1	
		}\]*
		is now MDS.
		(The last four columns of \cref{mat:H11} are merged into two.)
	\end{proof}
	
	\begin{cor}
		Tolerating one error, $\S1,\S2…\S{11}$ determine
		$\C11$, $\C12$, $\C21$, and $\C22$.
	\end{cor}
	
	\begin{proof}
		An aftermath of the minimum distance being $3$.
		
		Remark:
		This work does not put too much emphasis on the error correction aspect.
		This corollary is to remind ourselves that the ability to fill in erasures
		is bound to the ability to detect and correct errors.
	\end{proof}
	
	Whenever either $\S{10}$ or $\S{11}$ is erased,
	we can repair the erased one using any seven out of $\S1$--$\S7$ and $(\S8+\S9)$.
	Whenever both $\S{10},\S{11}$ are erased, we can repair neither
	because all we can learn is $(\S{10}+\S{11})$.
	This is another way to contrast the fact that $\S1$--$\S{11}$ are not MDS
	with the fact that $\S1$--$\S7$, $(\S8+\S9)$, and $(\S{10}+\S{11})$ are MDS.
	
	The phenomenon above is nowhere near a fatal deficit because the general direction
	is to repair the first seven symbols to pave the way for Strassen.
	Notwithstanding that, \cref{sec:recurse} exposes that repairing
	the backup symbols $\S8$--$\S{11}$ could be helpful to the overall job.
	Before that, we attempt to install redundancy for one last time.

\subsection{Third checksum}

	The manager can create as many backups as desired.
	So far the pattern seems to be:
	To tolerate one more erasure, hire two more workers.
	The pattern persists when the third checksum is added,
	whereas new difficulties come into sight.
	
	Let workers compute, in addition to $\S1$--$\S{11}$,
	\begin{alignat*}2
		\S{12}	≔&&		(2\A11-3\A21)	&⋆	(2\B11+\B12),	\\
		\S{13}	≔&\;&	(2\A12-3\A22)	&⋆	(2\B21+\B22).	
	\end{alignat*}
	Then an equation connecting $\S1$--$\S7$ to $\S{12}$ and $\S{13}$ is obtained from
	\[*\bma{2&-3}C\bma{2\\1}=\(\bma{2&-3}A\)⋆\left(B\bma{2\\1}\right).\]*
	The parity-check matrix for
	$\S1$--$\S7,(\S8+\S9),(\S{10}+\S{11}),(\S{12}+\S{13})$ is thus:
	\[\bma{	
		1	&	-4	&	3	&	-3	&	2	&	2	&	-1	&-1	&	&	\\
		1	&	1	&	4	&	2	&	3	&	-2	&	3	&	&-1	&	\\
		1	&	-3	&	-1	&	-2	&	-2	&	-3	&	4	&	&	&-1	\\[-1ex]
		•	&	•	&		&	•	&		&		&		&	&	&	\\[-1ex]
		•	&		&	•	&		&	•	&		&		&	&	&	\\[-1ex]
		•	&		&		&		&		&	•	&	•	&	&	&	\\[-5ex]
	}\rule[-7ex]{0pt}{0pt}\eqlabel{mat:Qcd}\]
	There are three zero minors, as indicated by the bullets.
	
	Some minors vanish not because we chose inappropriate $\S{12}$ and $\S{13}$.
	Rather, this is because ${(\S2-\S1)},\S3,{(\S4+\S1)},\S5,\S6,\S7$
	determine $\C11$--$\C22$ (look again at \crefrange{for:C11}{for:C22}).
	In other words, the best we can learn about $\S1,\S2,\S4$ given
	the other $S$'s and $C$'s are the linear combinations $(\S2-\S1)$ and $(\S4+\S1)$.
	For the other two zero minors, similar aggregation can be stated---%
	one is to replace $\S1,\S3,\S5$ by linear combinations $(\S3+\S1)$ and $(\S5-\S1)$;
	the other is replacing $\S1,\S6,\S7$ by $(\S6+\S1)$ and $(\S7+\S1)$.
	
	In conclusion, we cannot hope to recover, say,
	$(\S1,\S2,\S4)$ as a triple given $\S3$ and $\S5$--$\S{13}$,
	just like we cannot recover $(\S8,\S9)$ as a pair given the others.
	That does not stop us from inferring $\C11$--$\C22$, though.
	The next lemma and proposition sum up the situation.
	
	\begin{lem}\label{lem:MDS3}
		Each of the following three nonuples ($\,9$-tuples) is a $[9,6,4]$-MDS code: \\
		\smaller%%%%%
		\def\backups{(\S8+\S9),(\S{10}+\S{11}),(\S{12}+\S{13})}
		$\((\S2-\S1),\S3,(\S4+\S1),\S5,\S6,\S7,\backups\)$;\AB
		$\(\S2,(\S3+\S1),\S4,(\S5-\S1),\S6,\S7,\backups\)$;\AB
		$\(\S2,\S3,\S4,\S5,(\S6+\S1),(\S7+\S1),\backups\)$.
	\end{lem}
	
	\begin{proof}
		They share the same parity-check matrix
		\[*\bma{
			-4	&	3	&	-3	&	2	&	2	&	-1	&-1	&	&	\\
			1	&	4	&	2	&	3	&	-2	&	3	&	&-1	&	\\
			-3	&	-1	&	-2	&	-2	&	-3	&	4	&	&	&-1	
		}\]*
		(deleting the first column from \cref{mat:Qcd}) and this is MDS.
	\end{proof}
	
	\begin{pro}\label{pro:10outof13}
		Any ten out of $\S1,\S2…\S{13}$ determine $\C11,\C12,\C21,\C22$.
		To rephrase it, Pluto code tolerates
		three erasures out of thirteen computations.
	\end{pro}
	
	\begin{proof}
		By the lemma.
	\end{proof}
	
	It is no coincidence that the rows of the parity-check matrix
	are permutations of each other up to signs.
	We are informed by \cite{Burichenko14} that a group action transfers
	$-\S2,\S3,\S4,-\S5,\S6,\S7$ to $\S3,\S7,-\S5,\S6,\S4,-\S2$.
	We choose carefully the definitions of $\S8,\S9$ and derive,
	according to the group law, the definitions of $\S{10}$--$\S{13}$.
	More on that in \cref{app:triangle}.
	
	It is also intentional that the parity-check matrix
	contains but multiples of $2$ and $3$.
	That implies that the one parity-check $\S1$--$\S9$
	is MDS over all rings whose $2$ and $3$ are units.
	For $\S1$--$\S{11}$ to become MDS (in the manner of \cref{cor:MDS2}),
	there are $\binom92=45$ determinants to be inverted.
	Up to symmetries and signs, there are $23$ unique determinants,
	and they are multiples of the first seven primes assuming our choices.
	(Sidenote: seven is not optimal!)
	%$1$--$7$ and $9，10，11，13，14，18，19$.
	For $\S1$--$\S{13}$ to become MDS (in the regard of \cref{lem:MDS3}),
	one can automate the checking that
	the determinants comprise the first seven primes.
	
	From this point onward, it is not hard to see that
	every two more workers help conquer one more erasure/straggler.
	We stop at three checksums because the code rate is getting worse.
	If a manager wants to calculate $2$-by-$2$ matrix products with as many as
	thirteen workers, it may make use of EPC and enjoy a threshold of nine.
	See also \cref{app:RS} for an approach to build Pluto with EPC-flavor checksums.

\subsection{Tanner graph}

	We conclude this section with Tanner graphs in \cref{fig:Tanner911,fig:Tanner13}.
	Tanner graphs refine and visualize the relation among $\S1$--$\S{13}$
	stated in \cref{pro:8outof9}, \cref{cor:MDS2}, and \cref{lem:MDS3}.
	That is all we want to say about $2$-by-$2$ matrices.
	
	The next section develops Pluto codes upon $3$-by-$3$ matrices.

\begin{figure}
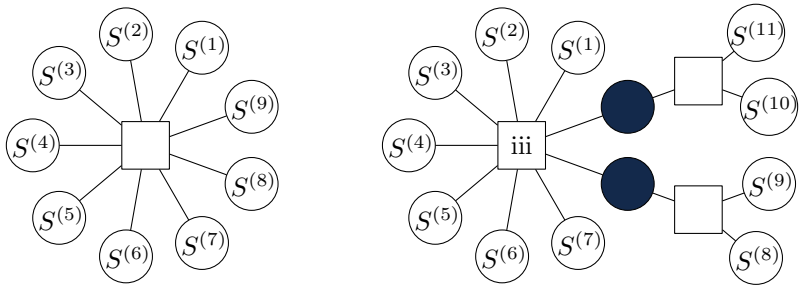

	\tikz[Tanner]{
		\draw node[cn](C){}
			foreach\s in{1,...,9}{
				(\s*40+20:1.5)node(S\s)[vn]{$\S\s\!$}(C)--(S\s)
			}
		;
		\tikzset{xshift=5cm}
		\draw node[cn](C){iii}
			foreach\s in{1,...,7}{
				(\s*40+20:1.5)node(S\s)[vn]{$\S\s\!$}(C)--(S\s)
			}
		;
		\foreach\s in{0,1}{
			\draw
				(\s*40-20:1.5)node[hn](H\s){}(C)--(H\s)
				+(\s*40-20:1)node[cn](C\s){}(H\s)--(C\s)
				foreach\t in{0,1}{
					(C\s)+(\s*20+\t*60-40:1)node[vn](V\s\t)
					{$\S{\the\numexpr8+\s*2+\t}\!$}(C\s)--(V\s\t)
				}
			;
		}
	}
	\caption{
		To the left:
		The Tanner graph that links $\S1$--$\S9$ (cf.\ \cref{pro:8outof9}).
		To the right:
		The Tanner graph that links $\S1$--$\S{11}$ (cf. \cref{cor:MDS2}).
		The check node labeled iii is an ``MDS check'' of minimum distance $3$---%
		it repairs any two variable nodes given the rest.
		Shaded are hidden (punctured) variable nodes,
		and they are $(\S8+\S9)$ and $(\S{10}+\S{11})$, respectively.
	}\label{fig:Tanner911}
\end{figure}

\begin{figure}
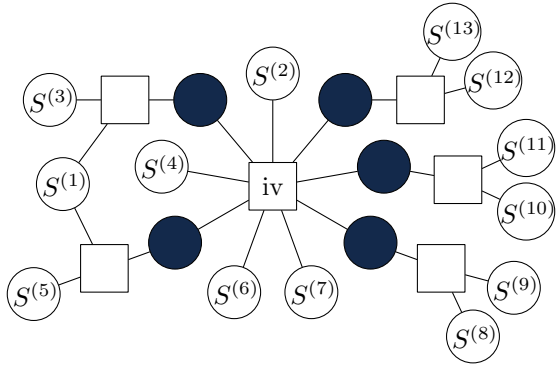

	\tikz[Tanner]{
		\draw
			node[cn](C){iv}
			foreach\s in{2,4,6,7}{
				(\s*40+10:1.5)node(S\s)[vn]{$\S\s\!$}(C)--(S\s)
			}
			(S4)+(190:4/3)node[vn](S1){$\S1\!$}
			foreach\s in{3,5}{
				(\s*40+10:1.5)node(H\s)[hn]{}(C)--(H\s)
				(H\s)+(\s*10+150:1)node[cn](C\s){}(S1)--(C\s)--(H\s)
				(C\s)+(\s*10+150:1)node[vn](S\s){$\S\s\!$}(C\s)--(S\s)
			}
			foreach\s in{8,10,12}{
				(\s*20+170:1.5)node[hn](H\s){}(C)--(H\s)
				(H\s)+(\s*5-60:1)node[cn](C\s){}(H\s)--(C\s)
				(C\s)+(\s*20-225:1)node[vn](S\s){$\S\s\!$}(C\s)--(S\s)
				(C\s)+(\s*20-175:1)node[vn](S\s){$\S{\the\numexpr\s+1}\!$}(C\s)--(S\s)
			}
		;
	}
	\caption{
		A Tanner graph for $\S1$--$\S{13}$.
		Cf.\ the second nonuple in \cref{lem:MDS3}.
		The hidden variables to the left are $(\S3+\S1)$ and $(\S5-\S1)$;
		to the right are $(\S8+\S9)$, $(\S{10}+\S{11})$, and $(\S{12}+\S{13})$.
	}\label{fig:Tanner13}
\end{figure}

\section{Laderman Algorithm and Checksums}\label{sec:Laderman}

	In this section, let us turn our attention to $3$-by-$3$ matrix multiplication:
	\[
		\bma{\A11&\A12&\A13\\\A21&\A22&\A23\\\A31&\A32&\A33}⋆
		\bma{\B11&\B12&\B13\\\B21&\B22&\B23\\\B31&\B32&\B33}=
		\bma{\C11&\C12&\C13\\\C21&\C22&\C23\\\C31&\C32&\C33}
	\]
	Laderman \cite{Laderman76} holds the current record that,
	instead of $27$ entry multiplications, $23$ suffice.
	Here are the definitions of $\L1$--$\L{23}$:
	{\def\diamondsuit(#1){(\hbox to140pt{\smaller[2]$#1$})}%%%%%
	\def\clubsuit(#1){(\hbox to146pt{\smaller[2]$#1$})}
	\begin{alignat*}2
		\L1		≔&&	♢(\A11+\A12+\A13-\A21-\A22-\A32-\A33)	&⋆	\B22,				\\
		\L2		≔&&				(\A11-\A21)		&⋆	(-\B12+\B22),					\\
		\L3		≔&&				\A22	&⋆	♣(-\B11+\B12+\B21-\B22-\B23-\B31+\B33),	\\
		\L4		≔&&			(-\A11+\A21+\A22)	&⋆	(\B11-\B12+\B22),				\\
		\L5		≔&&				(\A21+\A22)		&⋆	(-\B11+\B12),					\\
		\L6		≔&&					\A11		&⋆		\B11,						\\
		\L7		≔&&			(-\A11+\A31+\A32)	&⋆	(\B11-\B13+\B23),				\\
		\L8		≔&&				(\A11-\A31)		&⋆	(-\B13+\B23),					\\
		\L9		≔&&				(\A31+\A32)		&⋆	(-\B11+\B13),					\\
		\L{10}	≔&&	♢(\A11+\A12+\A13-\A22-\A23-\A31-\A32)	&⋆	\B23,				\\
		\L{11}	≔&&				\A32	&⋆	♣(-\B11+\B13+\B21-\B22-\B23-\B31+\B32),	\\
		\L{12}	≔&&			(-\A13+\A32+\A33)	&⋆	(\B22+\B31-\B32),				\\
		\L{13}	≔&&				(\A13-\A33)		&⋆	(\B22-\B32),					\\
		\L{14}	≔&&					-\A13		&⋆		-\B31	,					\\
		\L{15}	≔&&				(\A32+\A33)		&⋆	(-\B31+\B32),					\\
		\L{16}	≔&&			(-\A13+\A22+\A23)	&⋆	(\B23+\B31-\B33),				\\
		\L{17}	≔&&				(\A13-\A23)		&⋆	(\B23-\B33),					\\
		\L{18}	≔&&				(\A22+\A23)		&⋆	(-\B31+\B33),					\\
		\L{19}	≔&&					\A12		&⋆		\B21,						\\
		\L{20}	≔&&					\A23		&⋆		\B32,						\\
		\L{21}	≔&&					-\A21		&⋆		-\B13,						\\
		\L{22}	≔&&					-\A31		&⋆		-\B12,						\\
		\L{23}	≔&\;&				\A33		&⋆		\B33.						
	\end{alignat*}}%
	So far, this is a character-by-character copy from \cite{Laderman76} except that
	$\L8$, $\L{14}$, $\L{21}$, and $\L{22}$ have both their factors changing signs,
	which leave $\L8$, $\L{14}$, $\L{21}$, and $\L{22}$ themselves unchanged.
	This highlights a dihedral symmetry among the $23$ assignments
	as discussed in \cref{app:dihedral}.
	Listed below is the manager's recipe to cook $C$;
	we reorder $L$'s to honor said symmetry:
	{\everymath{\displaystyle}%%%%%
	\PMS\fakeL{width("$L\series{00}$")}
	\def\L#1{\rlap{$L\series{#1}$}\kern\fakeL pt}
	\PMS\halfL{width("${}+L\series{00}$")/2}
	\def~{\kern\halfL pt}
	\begin{align*}
		\C11	&	=\L6	+\L{14}	~	~	~	~	+\L{19},					\\
		\C12	&	=\L{14}	+\L6	+\L4	+\L5	+\L1	+\L{15}	+\L{12},	\\
		\C21	&	=\L6	+\L{14}	+\L{16}	+\L{17}	+\L3	+\L2	+\L4,		\\
		\C13	&	=\L{14}	+\L6	+\L7	+\L9	+\L{10}	+\L{18}	+\L{16},	\\
		\C31	&	=\L6	+\L{14}	+\L{12}	+\L{13}	+\L{11}	+\L8	+\L7,		\\
		\C22	&	=~	\L6		~	+\L4	+\L5	+\L{20}	+\L2,				\\
		\C23	&	=~	\L{14}	~	+\L{16}	+\L{17}	+\L{21}	+\L{18},			\\
		\C33	&	=~	\L6		~	+\L7	+\L9	+\L{23}	+\L8,				\\
		\C32	&	=~	\L{14}	~	+\L{12}	+\L{13}	+\L{22}	+\L{15}.			
	\end{align*}}%
	These $23+9$ assignments are collectively called the \emph{Laderman algorithm}.
	We next demonstrate how to install redundancy.

\subsection{First checksum}

	To interleave backups with the Laderman algorithm,
	the following equation is utilized:
	\[*\bma{1&2&3}C\bma{2\\-1\\3}=\(\bma{1&2&3}A\)⋆\left(B\bma{2\\-1\\3}\right)\]*
	The \LHS/ is a linear combination of $\C11$--$\C33$,
	so it can be rewritten as a linear combination of $\L1$--$\L{23}$.
	The \RHS/ is an inner product that costs three entry multiplications;
	number them $\L{24}$, $\L{25}$, and $\L{26}$, respectively.
	More formally,
	\begin{alignat*}2
		\L{24}	≔&&			(\A11+2\A21+3\A31)	&⋆	(2\B11-\B12+3\B13),				\\
		\L{25}	≔&&			(\A12+2\A22+3\A32)	&⋆	(2\B21-\B22+3\B23),				\\
		\L{26}	≔&\;&		(\A13+2\A23+3\A33)	&⋆	(2\B31-\B32+3\B33).				
	\end{alignat*}
	Then the \LHS/
	\[*\arraycolsep2.8pt\setcounter{MaxMatrixCols}{25}
		\bma{
			-1 & 2 & 4 & 1 & -3 & 21 & 18 & 15 & 12 & 3 & 6 &
			2 & 3 & 17 & -4 & 13 & 10 & 9 & 2 & -2 & 6 & -3 & 9
		}
		\bma{\L1\\[-.5ex]⋮\\[.5ex]\L{23}}
	\]*
	equals the \RHS/ $(\L{24}+\L{25}+\L{26})$.
	Therefore, symbols $\L1$--$\L{26}$ form a $[26,25,2]$-MDS code.
	This summarize the counterpart of \cref{pro:8outof9}.
	
	\begin{pro}
		$\L1$--$\L{26}$ form a $[26,25,2]$-MDS code.
		Any $25$ of them determine $\C11$--$\C33$.
	\end{pro}

\subsection{Second checksum}

	In a similar manner, consider this equation:
	\[*\bma{2&-1&3}C\bma{1\\3\\2}=\(\bma{2&-1&3}A\)⋆\left(B\bma{1\\3\\2}\right)\]*
	Index the three summands/entry products on the \RHS/
	by $\L{27}$, $\L{28}$, and $\L{29}$.
	Notationally,
	\begin{alignat*}2
		\L{27}	≔&&			(2\A11-\A21+3\A31)	&⋆	(\B11+3\B12+2\B13),				\\
		\L{28}	≔&&			(2\A12-\A22+3\A32)	&⋆	(\B21+3\B22+2\B23),				\\
		\L{29}	≔&\;&		(2\A13-\A23+3\A33)	&⋆	(\B31+3\B32+2\B33).				
	\end{alignat*}
	They induce a parity-check equation---that
	\[*\arraycolsep2.8pt\setcounter{MaxMatrixCols}{25}
		\bma{
			6 & -4 & -1 & 2 & 3 & 17 & 13 & 9 & 10 & 4 & 3 & 18
			& 12 & 21 & 15 & 1 & -3 & 2 & 2 & -3 & -2 & 9 & 6
		}
		\bma{\L1\\[-.5ex]⋮\\[.5ex]\L{23}}
	\]*
	equals $(\L{27}+\L{28}+\L{29})$.
	
	The two equations constitute a $2$-by-$25$ parity-check matrix
	that happens to be MDS.
	Thus the symbols $\L1$--$\L{23}$ and $(\L{24}+\L{25}+\L{26})$
	together with $(\L{27}+\L{28}+\L{29})$ form a $[25,23,3]$-MDS code.
	This concludes the counterpart of \cref{cor:MDS2}.
	
	\begin{pro}
		$\L1$--$\L{23}$, $(\L{24}+\L{25}+\L{26})$, and $(\L{27}+\L{28}+\L{29})$
		form a $[25,23,3]$-MDS code.
		Any $23$ of them determine $\C11$--$\C33$.
	\end{pro}
	
	We used symmetry to derive the checksums here, too.
	More elaboration on that in \cref{app:dihedral}.
	We invite readers to check out the paper \cite{Burichenko15}
	for the full symmetry of the Laderman algorithm, especially
	if they want to optimize the prime factors of the minors/determinants.
	Readers should also check \cite{HKS20} for other families of $3$-by-$3$ algorithms
	that could have performed better than Laderman by any means.

\subsection{More checksums}

	The story does not stop here.
	Every three extra workers grant the manager
	the ability to withstand one more erasure/straggler.
	Reminder:
	EPC has a threshold of $29$;
	the manager should switch to EPC after two more rounds
	(i.e., when there are $35$ workers in total).
	
	Nonetheless, it gets harder to characterize the codes
	as they are farther away from being MDS.
	For instance, the Laderman algorithm with three checksums induces
	sixteen triples that behave like $\S1,\S2,\S4$ in the last section.
	They are enumerated in \cref{fig:shuriken}.
	
	As much as some triples break the MDS property,
	this phenomenon is working in our favor because missing such three symbols
	requires only two checksums to infer the final output.
	The next section generalizes the concept from the perspective of matroids.

\begin{figure}
	\tikz[scale=2]{
		\draw[nodes={circle,draw,minimum size=2.5em,inner sep=1}]
			(0,2)node(C11){$\C11$}(2,2)node(C12){$\C12$}
					(1,1)node(C**){$\C{•}{•}$}
			(0,0)node(C21){$-\C21$}(2,0)node(C22){$-\C22$}
		;
		\draw[nodes={inner sep=1,auto,'},every edge/.append style={->}]
			(C22)edge[bend left]node[']{$\S1$}(C11)
			(C21)edge[bend right]node{$\S2$}(C22)
			(C22)edge[bend right]node{$\S3$}(C12)
			(C21)edge[bend left]node[']{$\S4$}(C11)
			(C11)edge[bend left]node[']{$\S5$}(C12)
			(C22)edge[bend right]node[]{$\S6$}(C**)
			(C**)edge[bend right]node{$\S7$}(C11)
		;
	}
	\caption{
		$\C{•}{•}$ is $-\C11-\C22+\C21+\C22=\bma{-1&1}C\bma{1&1}^⊤$.
		The decoding part of Strassen induces a matroid structure,
		and is represented by this directed graph.
	}\label{fig:incidence}
\end{figure}

\begin{figure}
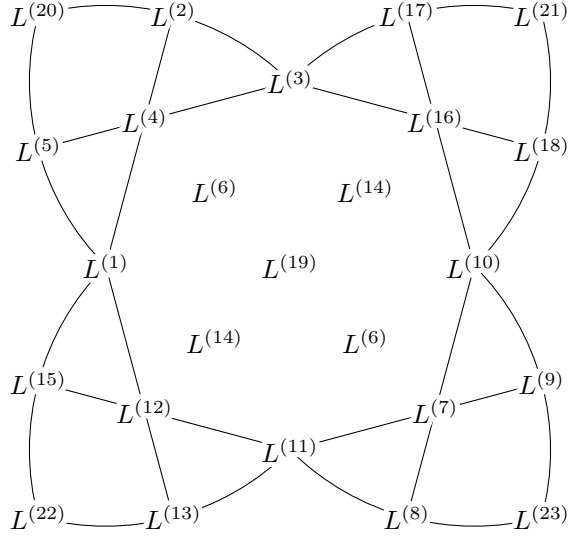

	\tikz[xscale=-1]{
		\PMS\a{sqrt(6)}
		\PMS\b{1+sqrt(3)}
		\PMS\c{2*sqrt(3)-2}
		\PMS\d{2*sqrt(3)}
		\def\blackonwhite$#1${%
			\pdfliteral{q 1 J 1 j 1 Tr}\pgfsetlinewidth{4}%
			\color{white}\rlap{$#1$}\pdfliteral{Q}\color{black}$#1$%
		}
		\def\NL#1 {node(L#1){\blackonwhite$\L{#1}$}}
		\def\TRI#1 #2 #3 #4 #5 {
			(-30:2)--(150:\c)(30:\c)--(210:2)(-30:2)arc(0:60:\d)arc(120:180:\d)
			(0:0)\NL#1 (30:\c)\NL#2 (90:2)\NL#3 (150:\c)\NL#4 (210:2)\NL#5
		}
		\draw[every to/.style={bend right=15}]
			{[rotate=-45,yshift=\b cm]	\TRI	4	5	20	2	3	}
			{[rotate= 45,yshift=\b cm]	\TRI	16	17	21	18	10	}
			{[rotate=135,yshift=\b cm]	\TRI	7	9	23	8	11	}
			{[rotate=225,yshift=\b cm]	\TRI	12	13	22	15	1	}
		;
		\draw(-1,1)node{$\L{14}$}(1,1)node{$\L6$}
			         (0,0)node{$\L{19}$}
			(-1,-1)node{$\L6$}(1,-1)node{$\L{14}$}
		;
	}
	\caption{
		Inspired by the Fano plane.
		Symbols on the same arc are having nullity $1$---
		they behave like $\S1,\S2,\S4$ in \cref{lem:MDS3}.
		For example, $\L1$, $\L2$, and $\L4$ do not lead to the MDS property
		until being replaced by two linear combinations thereof,
		namely $(\L1+\L4)$ and $(\L2+\L4)$.
		Three symbols, $\L6,\L{14},\L{19}$, do not involve in any arc;
		their placements in this diagram reveal how
		the dihedral group order $8$ acts on all the $L$'s.
		C.f.\ \cref{app:dihedral}.
	}\label{fig:shuriken}
\end{figure}

\section{Matroid Theoretical Perspective}\label{sec:matroid}

	A subset like $\{\S1,\S2\}$ satisfies the following property:
	if the manager knows $\{\S1,\S2\}^∁$ and $\C11$--$\C22$,
	then it can infer both $\S1$ and $\S2$.
	We say the subset $\{\S1,\S2\}$ has rank $2$ and nullity $0$.
	In similar fashion, the subset $\{\S1,\S2,\S3\}$ is such that,
	if the manager knows $\{\S1,\S2,\S3\}^∁$ and $\C11$--$\C22$,
	then it can infer all three of $\S1,\S2,\S3$;
	we say it has rank $3$ and nullity $0$.
	On the other hand, the subset $\{\S1,\S2,\S4\}$
	is said to have rank $2$ and nullity $1$ because
	the manager can infer but two linear combinations thereof.
	
	For general subsets, we let the rank be
	how many linear combinations the manager learns.
	And let the nullity be the degree of ambiguity/equivocation;
	it equals the size/cardinality of the subset minus the rank.
	Thus, for example, the ground set $\{\S1…\S7\}$ has rank $4$ and nullity $3$.

\subsection{Strassen matroid}

	Can we predict how well the manager learns easier?
	A systematic treatment to describe the behavior of the subsets is as below:
	compress \crefrange{for:C11}{for:C22} into a linear transformation represented by:
	\[*\def\cdot{\color{427}0}\bma{
		1	&	·	&	·	&	1	&	-1	&	·	&	1	\\
		·	&	·	&	1	&	·	&	1	&	·	&	·	\\
		·	&	1	&	·	&	1	&	·	&	·	&	·	\\
		1	&	-1	&	1	&	·	&	·	&	1	&	·	
	}\]*
	Any subset of columns therein possesses properties like \emph{cardinality},
	\emph{rank}, \emph{nullity}, \emph{corank}, etc.\ as in the matroid theory.
	And those terms coincide with our usage of rank and nullity in the last paragraph.
	
	In fact, this Strassen-derived matrix/matroid can be made a (directed) graph.
	First negate the bottom two rows,
	then stack an extra row such that all columns add up to zero.
	The result is:
	\[\def\cdot{\color{427}0}\bma{
		1	&	·	&	·	&	1	&	-1	&	·	&	1	\\
		·	&	·	&	1	&	·	&	1	&	·	&	·	\\
		·	&	-1	&	·	&	-1	&	·	&	·	&	·	\\
		-1	&	1	&	-1	&	·	&	·	&	-1	&	·	\\
		·	&	·	&	·	&	·	&	·	&	1	&	-1	
	}\eqlabel{mat:Strinf}\]
	The meaning of the new row, if anything, is how $\S1$--$\S7$
	is mapped to $-\C11-\C22+\C21+\C22=\bma{-1&1}C\bma{1&1}^⊤$.
	Of \cref{fig:incidence} this is the incidence matrix.
	
	Its \emph{corank-nullity polynomial} is thus defined.
	Remark:
	Historically, the same concept is captured under different names---%
	Whitney numbers, Tutte polynomial, rank generating function, to name a few.
	We find nullity the proper property as it characterizes
	how many checksums are needed to compensate a missing subset.
	And corank is the dual of nullity.
	We compute the polynomial of \cref{mat:Strinf}.
	
	\begin{pro}\label{pro:StrCoNu}
		The corank--nullity polynomial for Strassen is, over all fields,
		$x^4+7x^3+21x^2+32x+3x^2y+20+15xy+18y+3xy^2+7y^2+y^3$.
	\end{pro}
	Some interpretations of the monomials are:
	$x^4$ is the empty set;
	$7x^3$ means all one-subsets can be restored by the rest;
	$21x^2$ means all two-subsets can be restored by the rest;
	$3x^2y$ is the three three-subsets that cannot,
	as indicated below \cref{mat:Qcd};
	and $y^3$ is the ground set.
	
	Uniqueness remark:
	\cite{deGroot78} showed that the Strassen algorithm
	is essentially unique up to the trivial $\operatorname{GL}$-actions
	(for instance, replacing $A,B$ by $AR,R^{-1}B$).
	Thus the decoding matrix, or \cref{mat:Strinf}, is unique up to signs and reordings.
	Since the decoding matrix is a directed graph,
	which is not sensitive to the field characteristic,
	we conclude that the corresponding matroid is constant across all fields.
	An implication is that the corank-nullity polynomial
	in \cref{pro:StrCoNu} stays unchanged regardless of the context.

\subsection{Laderman polynomial}

	Similar to the previous subsection,
	compression of the Laderman algorithm ends up with:
	\[*\def\cdot{\color{427}0}\def\c#1{\clap{$^{#1}$}}\bma{
		·&·&·&·&·&1&·&·&·&·&·&·&·&1&·&·&·&·&1&·&·&·&·	\\
		1&·&·&1&1&1&·&·&·&·&·&1&·&1&1&·&·&·&·&·&·&·&·	\\
		·&·&·&·&·&1&1&·&1&1&·&·&·&1&·&1&·&1&·&·&·&·&·	\\
		·&1&1&1&·&1&·&·&·&·&·&·&·&1&·&1&1&·&·&·&·&·&·	\\
		·&1&·&1&1&1&·&·&·&·&·&·&·&·&·&·&·&·&·&1&·&·&·	\\
		·&·&·&·&·&·&·&·&·&·&·&·&·&1&·&1&1&1&·&·&1&·&·	\\
		·&·&·&·&·&1&1&1&·&·&1&1&1&1&·&·&·&·&·&·&·&·&·	\\
		·&·&·&·&·&·&·&·&·&·&·&1&1&1&1&·&·&·&·&·&·&1&·	\\
		·&·&·&·&·&1&1&1&1&·&·&·&·&·&·&·&·&·&·&·&·&·&1	\\
		&&&\c4&&&&\c8&&&&\c{12}&&&&\c{16}&&&&\c{20}		\\[-2ex]
	}\]*
	We compute its corresponding polynomial.
	
	\begin{thm}
		Laderman's corank-nullity polynomial reads, over generic fields,	\\
		\pgfplotstableread[header=false]{
			127348	188514	148229	78493	30079	8505	1755	253		23		1	
			427578	442448	242994	87583	21554	3530	350		16		0		0	
			779465	569150	222202	56082	9181	896		40		0		0		0	
			990980	507027	139626	24664	2708	164		4		0		0		0	
			961856	339280	63942	7652	528		16		0		0		0		0	
			743144	176918	21588	1640	60		0		0		0		0		0	
			466076	73168	5276	230		4		0		0		0		0		0	
			238812	24150	878		16		0		0		0		0		0		0	
			99647	6341	88		0		0		0		0		0		0		0	
			33451	1300	4		0		0		0		0		0		0		0	
			8835	198		0		0		0		0		0		0		0		0	
			1770	20		0		0		0		0		0		0		0		0	
			253		1		0		0		0		0		0		0		0		0	
			23		0		0		0		0		0		0		0		0		0	
			1		0		0		0		0		0		0		0		0		0	
		}\tableCorankNullity
		\def\prespace#1#2{%
			\noexpand\phantom{$\ifnum#1=0\else
				\ifnum#1=1\else\ifnum#1=2x\else x^{#1}\fi\fi
				\ifnum#2=0\else\ifnum#2=1y\else y^{#2}\fi\fi
			\fi$}%
		}
		\def\monomial#1#2{%
			\noexpand\rlap{$
				\ifnum#1=0\else\ifnum#1=1x\else x^{#1}\fi\fi
				\ifnum#2=0\else\ifnum#2=1y\else y^{#2}\fi\fi
				\ifnum#2=14.\fi
			$}%
		}
		\pgfplotstabletypeset[
			every head row/.style={output empty row},
			every column/.style={column type={@{}r@{}l}},
			assign cell content/.code={
				\edef\coef{\ifnum#1=127348 \else+\fi#1}
				\ifnum\coef=0
					\edef\pgfmarshalnext{\noexpand\pgfkeysalso{@cell content=&}}
				\else
					\ifnum\coef=1\def\coef{+}\fi
					\edef\pgfmarshalnext{
						\noexpand\pgfkeysalso{@cell content=
							\prespace\pgfplotstablecol\pgfplotstablerow$\coef$&
							\monomial\pgfplotstablecol\pgfplotstablerow
						}
					}
				\fi
				\pgfmarshalnext
			}
		]\tableCorankNullity\kern-\linewidth	\\
		For characteristic $3$, add a correction term of $-4+4xy$.
		For characteristic $2$, add a correction term of $(-1+xy)
			(4x^3+52x^2+290x+8x^2y+740+124xy+786y+20xy^2+465y^2+176y^3+40y^4+4y^5)$.
		(Notice the movement $1→xy$ of the masses.)
	\end{thm}
	
	Some interpretations of the monomials are:
	$x^9$ is the empty set;
	$23x^8$ means all one-subsets can be restored by the rest;
	$\binom{23}2x^7$ means all two-subsets can be restored by the rest;
	$16x^7y$ is the sixteen three-subsets that cannot,
	as arranged in \cref{fig:shuriken};
	and $y^{14}$ is the ground set.
	
	Characteristic remark:
	The course to determine the corank-nullity polynomials is two-fold.
	Step one:
	We evaluate all full minors;
	they range from $-2$ to $3$.
	This means that the only places it could ``go wrong''
	are when the characteristic is $2$ or $3$.
	Step two:
	We evaluate the corank-nullity polynomials over $ℚ,𝔽_2,𝔽_3$,
	and take the differences to see ``what goes wrong''.
	
	Since the incidence matrix of a directed graph has full minors ranging
	from $-1$ to $1$, the Laderman matroid cannot be a directed graph by contraposition.
	The best we can hope for is some contractions-restrictions that lead to
	directed graphs, and we did find some instances in \cref{app:minor}.
	
	The corank-nullity polynomials help us estimate
	the probabilistic performance of Pluto coding.
	See \cref{fig:3x3} for the prediction for the parity-checked Laderman algorithms.
	For the parity-checked Strassen algorithms (\cref{fig:2x2}), the estimate
	falls below the actual performance because $\S6$--$\S{19}$
	can fix each other, which is not taken into consideration here.
	
	There are much more that can be said regarding this topic;
	we leave them for future works.
	The next section discusses the concatenation of Pluto codes---how to build
	compound strategies for larger matrices out of the Strassen and Laderman algorithms.

\section{Recursive Strategy, Iterative Recovery}\label{sec:recurse}

	To forge (say) $\S1$, a worker is to multiply $\A11+\A22$ with $\B11+\B22$.
	The left factor/multiplier $\A11+\A22$ would be a huge matrix
	granted that the original matrix $A$ is a gigantic one.
	The same judgement applies to the right factor/multiplicand $\B11+\B22$%
	---it would be huge if $B$ is gigantic.
	To multiply two huge matrices, the worker applies
	the Strassen algorithm one more time.
	First and the foremost, it subdivides $\A11+\A22$ into four large blocks.
	Assume this convention for the indices of the large blocks:
	\[*\A ij=\bma{
		\AA ij11	&	\AA ij12	\\	\AA ij21	&	\AA ij22	
	}
	\qquad†and†\qquad
	\B ij=\bma{
		\BB ij11	&	\BB ij12	\\	\BB ij21	&	\BB ij22	
	}.\]*
	That leads to
	\def\Aaa{\AA{©1}{1©}11{+}\AA{®2}{2®}11}	\def\Aab{\AA{©1}{1©}12{+}\AA{®2}{2®}12}
	\def\Aba{\AA{©1}{1©}21{+}\AA{®2}{2®}21}	\def\Abb{\AA{©1}{1©}22{+}\AA{®2}{2®}22}
	\def\Baa{\BB{©1}{1©}11{+}\BB{®2}{2®}11}	\def\Bab{\BB{©1}{1©}12{+}\BB{®2}{2®}12}
	\def\Bba{\BB{©1}{1©}21{+}\BB{®2}{2®}21}	\def\Bbb{\BB{©1}{1©}22{+}\BB{®2}{2®}22}
	\[*\A{©1}{1©}+\A{®2}{2®}=\bma{\Aaa&\Aab\\\Aba&\Abb}\]*
	and the same expansion with all the $A$'s replaced by $B$.
	
	From here, the worker decomposes the huge multiplication
	behind $\S1$ into seven large products.
	The result, for the record, is
	{\smaller[1]%%%%%
	\begin{alignat*}2
		\SS11		≔&&		\((\Aaa)+(\Abb)\)	&⋆	\((\Baa)+(\Bbb)\),		\\
		\SS12		≔&&		\((\Aba)+(\Abb)\)	&⋆		(\Baa),				\\
		\SS13		≔&&				(\Aaa)		&⋆	\((\Bab)-(\Bbb)\),		\\
		\SS14		≔&&				(\Abb)		&⋆	\(-(\Baa)+(\Bba)\),		\\
		\SS15		≔&&		\((\Aaa)+(\Aab)\)	&⋆		(\Bbb),				\\
		\SS16		≔&&		\(-(\Aaa)+(\Aba)\)	&⋆	\((\Baa)+(\Bab)\),		\\
		\SS17		≔&\;&	\((\Aab)-(\Abb)\)	&⋆	\((\Bba)+(\Bbb)\).		
		\shortintertext{\normalsize
			Cf.\ \crefrange{for:Str1}{for:Str7}.
			Per demand, the backups will be defined to be}
		\SS18		≔&&		\((\Aaa)+2(\Aab)\)	&⋆	\(-(\Baa)+(\Bab)\),		\\
		\SS19		≔&&		\((\Aba)+2(\Abb)\)	&⋆	\(-(\Bba)+(\Bbb)\),		\\
		\SS1{10}	≔&&		\(3(\Aaa)-(\Aab)\)	&⋆	\((\Baa)+2(\Bab)\),		\\
		\SS1{11}	≔&&		\(3(\Aba)-(\Abb)\)	&⋆	\((\Bba)+2(\Bbb)\),		\\
		\SS1{12}	≔&&		\(2(\Aaa)-3(\Aab)\)	&⋆	\(2(\Baa)+(\Bab)\),		\\
		\SS1{13}	≔&&		\(2(\Aba)-3(\Abb)\)	&⋆	\(2(\Bba)+(\Bbb)\).		
	\end{alignat*}}%
	Clearly, any eight out of $\SS11$--$\SS19$ determine $\S1$.
	Any nine out of $\SS11$--$\SS1{11}$ determine $\S1$.
	Any ten out of $\SS11$--$\SS1{13}$ determine $\S1$.
	
	Imagine that each $\SS1t$ will be taken care of by one subordinate device,
	which we call a \emph{second-level worker}, or simply a \emph{$2$-worker}.
	When a second-level worker divides and distributes its sub-task further,
	the next layer are filled with \emph{third-level workers}, or \emph{$3$-workers}.
	And so on.
	The next subsection elaborates on the relation among
	the first- and second-level workers, which easily generalizes to all levels.

\subsection{TPC structure inside}\label{sec:TPC}

	Some abstract notation better reveals
	the structure of TPC hidden within the square of Strassen.
	Let $\al sij$ and $\be sij∈𝔽$ (where $1≤s≤13$ and $1≤i,j≤2$) be
	the scalars appearing in the parity-checked Strassen algorithms
	\[*\S s≔（∑_{ij}\al sij\A ij）⋆（∑_{ij}\be sij\B ij）.\]*
	Then, the left factor of $\S s$ becomes
	\[*∑_{ij}\al sij\A ij
		=∑_{ij}\al sij\bma{
			\AA ij11	&	\AA ij12	\\
			\AA ij21	&	\AA ij22	
		}
		=\bma{
			∑_{ij}\al sij\AA ij11	&	∑_{ij}\al sij\AA ij12	\\
			∑_{ij}\al sij\AA ij21	&	∑_{ij}\al sij\AA ij22	
		}
	\]*
	and the right factor of $\S s$ becomes
	\[*∑_{ij}\be sij\B ij
		=∑_{ij}\be sij\bma{
			\BB ij11	&	\BB ij12	\\
			\BB ij21	&	\BB ij22	
		}
		=\bma{
			∑_{ij}\be sij\BB ij11	&	∑_{ij}\be sij\BB ij12	\\
			∑_{ij}\be sij\BB ij21	&	∑_{ij}\be sij\BB ij22	
		}.
	\]*
	
	The $s$th worker, for each $1≤s≤13$, hereby applies another layer of Strassen
	to the multiplication that defines $\S s$.
	Let $\SS st$ be the $t$th sub-product that will be assigned
	to the $t$th second-level worker, for $1≤t≤13$.
	The left factor of $\SS st$ is thus
	\[*∑_{kl}\al tkl·\(†the †kl†-sub-block of the left factor of †\S s\)
		=∑_{kl}\al tkl∑_{ij}\al sij\AA ijkl.\]*
	The right factor of $\SS st$ is, by the same reasoning,
	\[*∑_{kl}\be tkl·\(†the †kl†-sub-block of the right factor of †\S s\)
		=∑_{kl}\be tkl∑_{ij}\be sij\BB ijkl.\]*
	In summary,
	\[\SS st=（∑_{ijkl}\al sij\al tkl\AA ijkl）⋆（∑_{ijkl}\be sij\be tkl\BB ijkl）.
		\eqlabel{for:SSful}\]
	It does not take too long to see that the expansion of $\SS st$ is symmetric
	between the outer ($s$--$i$-$j$) layer and the inner ($t$--$k$-$l$) layer.
	The symmetry is a free byproduct that helps us spot relations among $\SS st$.
	
	Merely from the fact that $\SS s1$--$\SS s9$ is a tactic
	to calculate $\S s$ with backups, we deduce that (cf.\ \cref{equ:Q89})
	{\smaller%%%%%
	\[*\SS s1	-4\SS s2	+3\SS s3	-3\SS s4	+2\SS s5	+2\SS s6	-\SS s7
		=\SS s8+\SS s9.\]*}%
	By symmetry (cf.\ \cref{for:SSful}), we deduce that
	{\smaller%%%%%
	\[*\SS 1t	-4\SS 2t	+3\SS 3t	-3\SS 4t	+2\SS 5t	+2\SS 6t	-\SS 7t
		=\SS 8t+\SS 9t.\]*}%
	\Cref{equ:Qab} generalizes likewise:
	{\smaller%%%%%
	\begin{align*}
		\SS s1	+\SS s2	+4\SS s3	+2\SS s4	+3\SS s5	-2\SS s6	+3\SS s7
			&	=\SS s{10}+\SS s{11}, \\
		\SS 1t	+\SS 2t	+4\SS 3t	+2\SS 4t	+3\SS 5t	-2\SS 6t	+3\SS 7t
			&	=\SS {10}t+\SS {11}t. 
	\end{align*}}%
	Two more equations involving $\SS s{12},\SS s{13},\SS{12}t,\SS{13}t$ are omitted.

	Consequence:
	Put $\SS11$--$\SS99$ in an array like \cref{fig:Elias} does.
	Then each row is a standalone $[9,8,2]$-MDS code;
	each column is a standalone $[9,8,2]$-MDS code.
	Together, they weave a web of protection like a TPC does.
	
	To decode, the manager examines every axis (an axis is a row or a column).
	If an axis contains nine known symbols, it is considered finished.
	If an axis contains eight known symbols and one unknown,
	the manager invokes the MDS relation to predict the missing one, and put it
	in the array despite the corresponding worker still running.
	If an axis contains two or more unknowns, skip for now.
	The manager executes the check iteratively
	while it is waiting for the workers to submit their answers;
	it declares completion if it learns $\SS11$--$\SS77$.
	
	Hereafter, we refer to this calculation strategy as ``$9·9$''.
	Its performance is analyzed in the sequel.

\begin{figure}
	\tikz[yscale=.5]{
		\begin{scope}[rounded corners=1em]
			\fill[Dark yellow](7-.1,.1)rectangle(9+.1,-9+.1);
			\fill[Gray-blue](.1,-7+.1)rectangle(9-.1,-9-.1);
			\fill[Dark yellow,opacity=.5](7-.1,.1)rectangle(9+.1,-9+.1);
		\end{scope}
		\foreach\s in{1,...,9}{
			\foreach\t in{1,...,9}{
				\draw(\t-.5,.5-\s)node{$\SS\s\t$};
			}
		}
	}
	\caption{
		The array of intermediate symbols to be figured out
		in multiplying $4×4$ matrices.
		Backups are shaded.
		Each row is $[9,8,2]$-MDS;
		each column is $[9,8,2]$-MDS.
		Cf.\ \cite[Fig.~1]{Elias54}.
	}\label{fig:Elias}
\end{figure}

\subsection{Stopping sets}

	Speaking of the recovery threshold (or equivalently, the minimum distance),
	this two-dimensional array of single parity-check codes does not possess a good one
	because four erasures ruin it (the square from $\SS11$ to $\SS22$).
	Despite the poor parameter, $9·9$ resists probabilistic errors pretty well.
	To see this, we characterize exactly the completion criterion.
	
	Construct a bipartite graph as follows:
	the s-part and the t-part each has nine vertices.
	For every cell in the array,
	an unknown $\SS st$ corresponds to an edge incident with
	vertex $s$ of the s-part and vertex $t$ of the t-part;
	a known $\SS st$ corresponds to a non-edge.
	Then, the manager can learn a new $\SS st$
	if that edge is incident with a degree-one vertex
	(that is, if that $\SS st$ is the lonely unknown in its column or row).
	Repeat this knowledge-propagation process;
	the manager stops after running out of degree-$1$ vertices.
	At this moment, if any of $\SS11$--$\SS77$ remains as an edge,
	the manger cannot declare completion;
	otherwise it can.
	We summarize this paragraph as follows.
	
	\begin{thm}\label{thm:2core}
		Assume $9·9$, the strategy depicted in \cref{fig:Elias}.
		The manager declares completion when the $2$-core of the corresponding
		bipartite graph does not contain any of the edges $\SS11$--$\SS77$.
	\end{thm}
	
	In LDPC (low-density parity-check) code terminology, a set of unknowns
	that prohibits the decoding from progressing is called a \emph{stopping set}.
	The smallest (in terms of cardinality) stopping sets in this case are $4$-cycles.
	As it turns out, four randomly-chosen edges are unlikely to line up as a $4$-cycle.
	Hence the said strategy withstands, with high probability,
	four erasures, sometimes even more.
	For the picture of a more detailed distribution, see \cref{sec:sample}.
	
	Move on to a new strategy ``$11·11$'';
	more precisely, $11$ workers compute $\S1$--$\S{11}$ by each hiring
	$11$ second-level workers to compute $\SS s1$--$\SS s{11}$.
	The bipartite graph can be defined similarly with a slight modification:
	Start with eleven vertices on both sides,
	and place edges based on whether $\SS st$ is known.
	Then we merge the eighth and the ninth vertices, for each side;
	and we merge the tenth and the eleventh vertices, for each side.
	Finally, we remove repeated edges.
	Thus, for example, an edge connects
	the eighth vertex in the s-part to the ninth vertex in the t-part
	if $\SS8{10}$, $\SS8{11}$, $\SS9{10}$, \emph{or} $\SS9{11}$ is unknown.
	Below goes the generalization of \cref{thm:2core}.
	
	\begin{thm}\label{thm:3core}
		For strategy $11·11$, the manager announces completion
		when the $3$-core of the corresponding bipartite graph
		does not contain any of the edges $\SS11$--$\SS77$.
	\end{thm}

\subsection{Rectangular/cuboid generality}\label{sec:rectangle}

	Earlier in this section, we demonstrated that a parity-checked Strassen algorithm
	can be \emph{tensored} with itself to build a calculation scheme
	for greater matrices, whose decoding enjoys a TPC-flavor template.
	It happens that Pluto codes based on the Laderman algorithm can be tensored, too.
	In actuality, all Pluto codes can be tensored,
	including those that are already tensor products.
	Therefore, for example, one can tensor $\S1$--$\S9$ with $\L1$--$\L{26}$
	to construct a Pluto code that multiplies $6×6$ matrices.
	One can tensor $\S1$--$\S9$ with $\SS11$--$\SS99$ to multiply $8×8$ matrices.
	And so forth.
	
	The FMM literature uses the notion $⟨2,2,2;7⟩$ to convey the message that
	seven entry multiplications determine the product of two $2×2$ matrices.
	We abuse this notion;
	write $⟨2,2,2;9⟩$ to express the calculation scheme $\S1$--$\S9$;
	write $⟨2,2,2;11⟩$ to express $\S1$--$\S{11}$;
	and write $⟨2,2,2;13⟩$ to express $\S1$--$\S{13}$.
	Readers can guess what are $⟨3,3,3;23⟩$
	along with $⟨3,3,3;26⟩$ as well as $⟨3,3,3;29⟩$.
	In the new language we restate that tensoring is a versatile mechanism.
	
	\begin{thm}
		Let $⟨ℓ,m,n;r⟩$ and $⟨ℓ',m',n';r'⟩$ be two Pluto codes.
		(They do not have to have any protection;
		raw FMM works, too.)
		Then there exists a Pluto code $⟨ℓℓ',mm',nn';rr'⟩$
		that assumes the TPC template.
	\end{thm}
	
	\begin{proof}
		Mimic what \cref{sec:TPC} does to $⟨2,2,2;9⟩$ and $⟨2,2,2;9⟩$.
	\end{proof}
	
	Pluto codes that are the tensor products of others are said to be \emph{composite}.
	Otherwise they are said to be \emph{prime}.
	A number of prime Pluto codes follow.
	
	\cite{HK71} proposed
	\[*\left⟨ℓ,2,n;\left⌈÷{3ℓn+\max(ℓ,n)}2\right⌉\right⟩,\]*
	in particular $⟨2,2,3;11⟩$ and $⟨3,2,3;15⟩$.
	(And \cite{Alekseyev85} proved that $⟨2,2,3;11⟩$ is optimal over all fields).
	We might as well insert checksums to these algorithms
	to construct Pluto codes
	\[*⟨2,2,3;13⟩≅⟨3,2,2;13⟩,\quad⟨2,3,2;14⟩,
		\quad⟨3,2,3;17⟩,\quad⟨2,3,3;18⟩≅⟨3,3,2;18⟩.\]*
	(Details omitted.)
	Suitable tensor products of these Pluto codes
	multiply $6×6$, $12×12$, or even $18×18$ matrices:
	\begin{align*}
		⟨6,6,6;234⟩		&	=⟨2,2,3;13⟩⊗⟨3,3,2;18⟩,	\\
		⟨6,6,6;238⟩		&	=⟨2,3,2;14⟩⊗⟨3,2,3;17⟩,	\\
		⟨12,12,12;2366⟩	&	=⟨2,2,3;13⟩⊗⟨2,3,2;14⟩⊗⟨3,2,2;13⟩,	\\
		⟨18,18,18;5508⟩	&	=⟨2,3,3;18⟩⊗⟨3,2,3;17⟩⊗⟨3,3,2;18⟩.	
	\end{align*}
	Note that the schoolbook complexity is $⟨18,18,18;5832⟩$,
	meaning that the Pluto code $⟨18,18,18;5508⟩$ relies on fewer workers;
	what is more, not all of them must reply to the manager.
	
	Whereas Pluto code $⟨12,12,12;2366⟩$ is blown away by the schoolbook method
	$⟨12,12,12;1728⟩$, other building blocks seem a lot better:
	\begin{align*}
		⟨12,12,12;1344⟩	&	=⟨3,3,4;32⟩⊗⟨4,4,3;42⟩,	\\
		⟨12,12,12;1353⟩	&	=⟨3,4,3;33⟩⊗⟨4,3,4;41⟩.
	\end{align*}
	Here $⟨3,3,4;29⟩$ and $⟨3,4,4;38⟩$ are taken from \cite{Smirnov13}.
	
	Explanation:
	Whether a composite Pluto code outperforms the naïve approach depends
	heavily on whether the bilinear ranks of the prime components
	are much less than the naïve complexities to begin with.
	For $⟨2;2;2;7⟩$ and $⟨2,2,3;11⟩$,
	adding as few as one checksum exceeds the naïve complexity.
	On the other hand, $⟨2,3,3;15⟩$ and $⟨3,3,3;23⟩$ afford one additional checksum,
	let alone even greater building blocks $⟨3,3,4;29⟩$ and $⟨3,4,4;38⟩$ who afford two.
	In the long run, the tensor powers of the ``wealthy'' ones
	show an advantage over the naïve methods, the latter including EPC and PDC.
	We give two gargantuan examples that support this argument:
	Pluto $⟨3,3,3;26⟩^{⊗3}=⟨27,27,27;17576⟩$ versus naïve $⟨27,27,27;19683⟩$.
	Pluto $⟨3,3,4;32⟩⊗⟨3,4,3;33⟩⊗⟨4,3,3;32⟩=⟨36,36,36;33792⟩$ versus	\\
	naïve $⟨36,36,36;46656⟩$.
	
	Outside this subsection, we identify a prime Pluto by the number of workers.
	Hence $7,9,11,\dotsc$ are those of type $⟨2,2,2⟩$;
	and $23,26,29,\dotsc$ are those of type $⟨3,3,3⟩$;
	and suchlike.
	There is one conflict at $⟨2,2,2;13⟩$ and $⟨2,2,3;13⟩$;
	they will be distinguishable from the context, otherwise we will specify.

\subsection{Tree and flexibility}

	Another approach to nest and visualize Pluto codes
	is by means of hierarchy trees, as illustrated in \cref{fig:tree}.
	In the figure, a manager is in charge of calculating $C≔AB$.
	It divides the matrices into $\A11$--$\B22$,
	and then assigns the calculations of $\S1$--$\S9$ to nine workers.
	The $s$th worker (for each $1≤s≤9$) will further divide
	the left and right factors of $\S s$ into smaller pieces,
	and each piece will be handled by a second-level worker.
	Per personal considerations, workers may subdivide differently---%
	one worker might subdivide into $3×3$ pieces whilst others still do $2×2$.
	Workers are also free to decide how many backups they want, if at all.
	
	Doing so, they mimic irregular LDPC codes with endless possibilities.
	Not only are irregular LDPC codes more flexible,
	but also strictly better than the regular ones.
	Intuitively speaking, the low-degree check nodes
	in an irregular LDPC damage the code rate but converge faster,
	while the high-degree check nodes honor the code rate but converge slower.
	The state-of-the-art irregular LDPC is a proper mixture
	of high- and low-degree check nodes that max out the said trade-off.
	
	Later in \cref{sec:sample}, we will experiment some strategies that are
	mostly regular/\AB homogeneous tensor products of prime Pluto codes.
	Those have an obvious disadvantage that the number of
	required $2$-workers are somewhat rigid, which forces the manager
	to prepare \emph{exactly} that certain amount of $2$-workers.
	What we offer in this subsection is that,
	by attaching superfluous $2$-workers to random $1$-workers
	or discarding putative $2$-workers from arbitrary $1$-workers,
	the total number of $2$-workers may deviate from the designed quantity,
	whilst the aftereffect on the performance is supposedly indistinguishable.
	A more in-depth comparison on that subject is to be conducted.
	
	The $3$-by-$3$ discussion ends here.
	In the next section, we introduce another powerful technique to forge Pluto codes,
	which also brings irregularity/inhomogeneity into the graph.

\begin{figure}
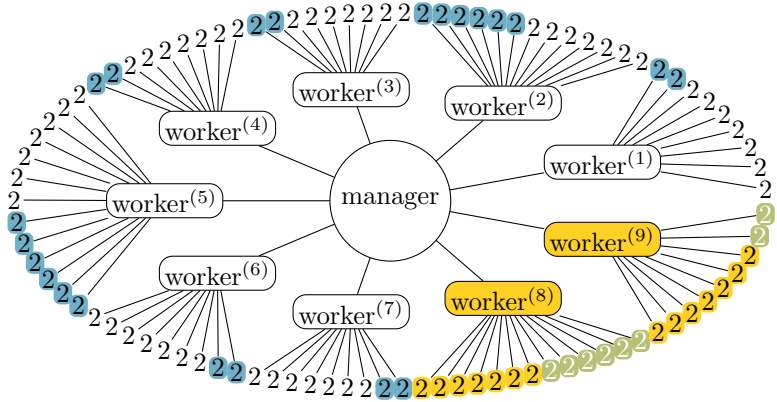
%\labortrue
	\makeatletter
	\tikzset{
		worker/.style={draw,rounded corners=1em/2,inner sep=2},
		2worker/.style={
			inner sep=1,rounded corners=1em/4,
			/utils/exec={
				\xdef\t{\the\numexpr\t+1}
				\EA\pgfutil@in@\s{258}
				\ifpgfutil@in@
					\ifnum\t=14\xdef\s{\the\numexpr\s+1}\xdef\t{1}\fi
				\else
					\ifnum\t=10\xdef\s{\the\numexpr\s+1}\xdef\t{1}\fi
				\fi
				\ifnum\s>7\pgfkeysalso{fill=Dark yellow}\fi
				\ifnum\t>7
					\pgfkeysalso{fill=Gray-blue}
					\ifnum\s>7
						\pgfkeysalso{fill=Gray-blue!50!Dark yellow,text=white}
					\fi
				\fi
			}
		}
	}
	\tikz[yscale=1/2]{
		\draw node(m)[circle,draw]{manager};
		\def\tt{0}\def\s{1}
		\foreach\s in{1,...,9}{
			\tikzset{rotate=\s*40-20,shift={(3,0)}}
			\ifnum\s>7\tikzset{nodes={fill=Dark yellow}}\fi
			\draw node(s\s)[worker]{worker$\series\s$}(s\s)--(m);
		}
		\tikzset{
			decoration={
				markings,
				mark=between positions .5/93 and 1 step 1/93 with {
					\coordinate(t\tt);
					\xdef\tt{\the\numexpr\tt+1}
				}
			}
		}
		\draw[yscale=4/5,decorate]circle(5);
		\def\t{0}
		\foreach\tt in{0,...,92}{
			\pgfpointanchor{t\tt}{center}
			\xdef\tx{\the\pgf@x}\xdef\ty{\the\dimexpr\pgf@y*5/4}
			\draw(\tx,\ty)node(t)[2worker]{$2$}(t)--(s\s);
		}
	}
	\caption{
		A manager breaks each of the two gigantic matrices $A,B$ into four huge blocks,
		associating $\S1$--$\S9$ to nine workers.
		The shaded workers are the backups ($\S8$ and $\S9$).
		For each $1≤s≤9$, the $s$th worker breaks the huge blocks further,
		assigning $\SS s1$--$\SS s9$ to nine $2$-workers (labeled ``2'').
		The shaded $2$-workers are the backups;
		the white-on-green ones are the backups of backups.
		Exception:
		Individual workers ($2$nd, $5$th, and $8$th) can hire
		four more $2$-workers and assign them $\SS s{10}$--$\SS s{13}$.
	}\label{fig:tree}
\end{figure}

\section{Union/Gluing Strategies}\label{sec:unify}

	This section is devoted to how to stack/overlay a TPC strategy
	as generated in the last section onto another strategy.
	We need some preparation.
	
	Recall $4$-by-$4$ matrices
	\[*A≔\bma{
		\AA1111&\AA1112	&&		\AA1211&\AA1212	\\
		\AA1121&\AA1122	&&		\AA1221&\AA1222	\\[2ex]
		\AA2111&\AA2112	&&		\AA2211&\AA2212	\\
		\AA2121&\AA2122	&&		\AA2221&\AA2222	
	}\]*
	and $B$ (indexed similarly).
	(Note that, under the usual convention,
	$\AA ijkl$ would have been indexed as $\A{(2i-2+k)}{(2j-2+l)}$.)
	To calculate $A⋆B$, the manager might repeat a parity-checked Strassen algorithm
	twice, or introduce checksums to the squared Strassen algorithm.
	\emph{It might as well do both.}
	The former, squaring a parity-checked algorithm, is what \cref{fig:Elias} did.
	The latter, parity-checking a squared algorithm, is detailed below.
	
	Consider this parity-check equation:
	\[*\bma{1&1&1&1}C\bma{3\\1\\1\\2}
		=\(\bma{1&1&1&1}A\)⋆\left(B\bma{3\\1\\1\\2}\right)\eqlabel{equ:motive}\]*
	The \LHS/ is a linear combination of entries of $C$,
	so it is a linear combination of $\SS11$--$\SS77$.
	The \RHS/ is a $1×4$ vector dotting a $4×1$ vector.
	Label the four entry products therein
	$ß{50}$, $ß{51}$, $ß{52}$, and $ß{53}$, respectively.
	Id est,
	\[*ß{50}≔(\AA1111+\AA1121+\AA2111+\AA2121)⋆(3\BB1111+\BB1112+\BB1211+2\BB1212),\]*
	and $ß{51}$--$ß{53}$ are defined in the same way.
	Then a single parity-check code pops up.
	\begin{pro}\label{pro:52outof53}
		$\SS 11$--$\SS77$ and $ß{50}$--$ß{53}$ make a $[53,52,2]$-MDS code.
	\end{pro}
	The following array catalogs the coefficients of $\SS11$--$\SS77$ in this relation:
	\[\bma{
		7	&	1	&	6	&	8	&	-1	&	3	&	4	\\
		1	&	3	&	-2	&	4	&	-3	&	-1	&	2	\\
		6	&	-2	&	8	&	4	&	2	&	4	&	2	\\
		8	&	4	&	4	&	12	&	-4	&	2	&	6	\\
		-1	&	-3	&	2	&	-4	&	3	&	1	&	-2	\\
		3	&	-1	&	4	&	2	&	1	&	2	&	1	\\
		4	&	2	&	2	&	6	&	-2	&	1	&	3	
	}\eqlabel{mat:H7x7}\]
	Why it is symmetric is left to readers.
	
	By perturbing the coefficients in \cref{equ:motive} dexterously,
	we might generate four more backups, which can be denoted by $ß{54}$--$ß{57}$.
	The game continues in a way that every four new \textit\ss's
	offer insurance against one more erasure.
	
	So far, we have invented four types of backups:
	\begin{itemize}
		\item	$\SS st$ for $1≤s≤7<t$;
		\item	$\SS st$ for $1≤t≤7<s$;
		\item	$\SS st$ for $7<s,t$; and
		\item	$ß{49}$--$ß{53}$ (and potentially four more).
	\end{itemize}
	It is up to the manager to decide how much insurance it intends to receive
	and which combinations of the insurance is cost-efficient.
	We analyze two union strategies in the next subsection.

\subsection{Cutting short cycles}

	First consider the union of $7·9$ and $49+4$.
	This union means that a total of $67$ workers work out $\SS11$--$\SS79$
	and $ß{50}$--$ß{53}$.
	For convenience, this scheme is denoted by $7·9∪53$ from now on.
	Tensoring takes precedence over unioning.
	
	By \cref{pro:8outof9}, one erasure in $\SS s1$--$\SS s9$
	(for every $1≤s≤9$) can be fixed.
	When there are more erasures than that, the $\textit\ss$-part kicks in.
	Here goes a running example:
	Suppose all but $\SS11$, $\SS12$, and $\SS21$ are known,
	then $\SS22$--$\SS29$ fix $\SS21$.
	The manager is left with two unknowns, $\SS11$ and $\SS12$,
	constrained by two relations coming from
	$\SS11$--$\SS19$ and from $\SS11$--$\SS77$ plus $ß{50}$--$ß{53}$.
	One hopes that these two relations are ``linearly independent''---that
	they impose uniquely solvable system of linear equations in $\SS11$ and $\SS12$.
	That indeed holds true, and not only for this particular example.
	
	\begin{thm}\label{thm:2cycle}
		$7·9∪53$ withstands one erasure on every axis
		parallel to the ``$\,9\!$''-direction.
		Beyond that, it withstands one more erasure.
	\end{thm}
	
	\begin{proof}
		Suppose that there are still two unknowns after invoking \cref{pro:8outof9}.
		
		If both unknowns are among $\SS18$--$\SS79$ and $ß{50}$--$ß{53}$,
		then all of $\SS11$--$\SS77$ are known, which yield the correct outcome.
		If exactly one unknown is among $\SS11$--$\SS77$
		(call it $\SS st$ for $1≤s,t≤7$), then the other unknown must be
		either of the backup on the same row (call it $ \SS s{τ}$ for $8≤τ≤9$),
		otherwise \cref{pro:8outof9} would have eliminated $\SS st$.
		This means that $ß{50}$--$ß{53}$ are all known,
		and the manager can invoke \cref{pro:52outof53} to infer $\SS st$.
		The last case is when two unknowns are
		in the same row of the first seven columns.
		
		It suffices to verify that, for the two remaining unknowns
		$\SS st$ and $\SS s{τ}$ (where $1≤s,t,τ≤7$), the matrix
		\[*\bma{
			(s,t)†th entry of \cref{mat:H7x7}†	&(s,τ)†th entry of \cref{mat:H7x7}†	\\
			t†th coefficient of \cref{equ:Q89}†	&τ†th coefficient of \cref{equ:Q89}†
		}\]*
		is invertible.
		A brute-force approach confirms that this is true.
	\end{proof}
	
	Now consider $9·9$ union $53$.
	This strategy has a total of $85$ workers working on
	$\SS11$--$\SS99$ and $ß{50}$--$ß{53}$ in multiplying $4$-by-$4$ matrices.
	From \cref{thm:2core} we see that $\SS11$--$\SS99$ save against a single erasure
	in any axis, and is helpless facing cycles.
	Now the $53$-part cuts the cycles.
	
	\begin{thm}\label{thm:4cycle}
		$9·9∪53$ resists any erasure pattern wherein the $2$-core is edgeless.
		Moreover, knowing $ß{50}$--$ß{53}$ resists a $4$-cycle as the $2$-core.
	\end{thm}
	
	\begin{proof}
		When the $2$-core is empty, \cref{thm:2core} takes care of it.
		Suppose the $2$-core is a $4$-cycle.
		If the cycle overlaps with $\SS11$--$\SS77$ at at most one edge,
		\cref{pro:52outof53} rules out the unknown edge.
		If the $4$-cycle overlaps with $\SS11$--$\SS77$ at two edges,
		\cref{thm:2cycle} handles the case regardless
		whether the two edges share the same $s$-value or $t$-value.
		The last possibility is when the $4$-cycle lie entirely within $\SS11$--$\SS77$.
		
		\def\entry{†th of \eqref{mat:H7x7}†}
		\def\coef{†th coef of \eqref{equ:Q89}†}
		\def\dual#1/#2·#3 {\tfrac{#1}{#2†-coef†·#3†-coef of \eqref{equ:Q89}†}}
		It suffices to verify that, for four unknowns
		$\SS st,\SS s{τ},\SS{σ}t,\SS{σ}{τ}$ (where $1≤s,σ,t,τ≤7$), the matrix
		\[\bma{
			(s,t)\entry		&(s,τ)\entry	&(σ,τ)\entry	&(σ,t)\entry	\\
			t\coef			&τ\coef											\\
							&s\coef			&σ\coef							\\
							&				&τ\coef			&t\coef			\\
			s\coef			&				&				&σ\coef			
		}\eqlabel{mat:4cycle}\]
		is full-rank.
		Note that none of the coefficients are zero.
		Therefore the last four rows span the dual space of
		\[*\bma{
			\dual1/s·t		&\dual-1/s·τ	&\dual1/σ·τ		&\dual-1/σ·t		
		}.\]*
		It remains to show that this vector is not
		orthogonal to the first row of \cref{mat:4cycle}.
		Symbolically, we want
		\[*\dual(s,t)\entry/s·t	-\dual(s,τ)\entry/s·τ	
			+\dual(σ,τ)\entry/σ·τ	-\dual(σ,t)\entry/σ·t	≠0.\]*
		A brute-force approach confirms this is true for all $1≤s,σ,t,τ≤7$.
	\end{proof}
	
	For the union strategy $7·9∪9·7∪53$ as depicted in \cref{fig:union},
	the resiliency is presumably the same as stated in \cref{thm:4cycle}.
	The only difference is that $\SS88$--$\SS99$ are permanently unknown.
	
	It is not difficult to see that, say, $9·11∪11·9∪57$ is also a valid union.
	The analysis of such higher-order protections is left for future works.
	We next go for higher-dimensional arrays.

\begin{figure}
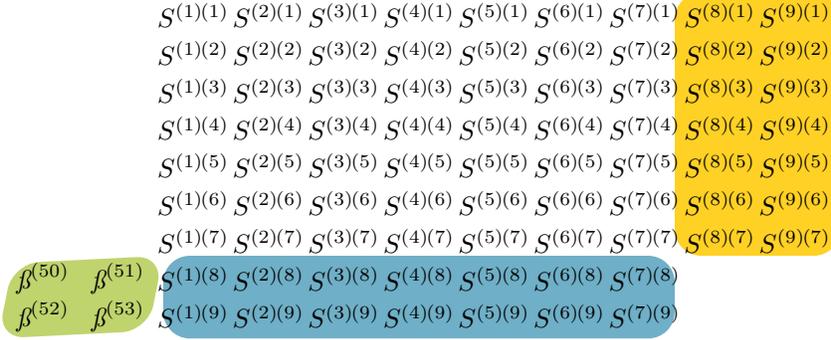

	\tikz[yscale=.5]{
		\begin{scope}[rounded corners=1em]
			\fill[Dark yellow](7-.1,.1)rectangle(9+.1,-7+.1);
			\fill[Gray-blue](.1,-7+.1)rectangle(7-.1,-9-.1);
			\fill[Citron](-2+.1,-7-.1)--(+.1,-7+.1)--(-.1,-9+.1)--(-2-.1,-9-.1)--cycle;
		\end{scope}
		\foreach\s in{1,...,9}{
			\foreach\t in{1,...,7}{
				\draw(\s-.5,.5-\t)node{$\SS\s\t$};
			}
		}
		\foreach\s in{1,...,7}{
			\foreach\t in{8,9}{
				\draw(\s-.5,.5-\t)node{$\SS\s\t$};
			}
		}
		\draw[shift={(-2.5,-6.5)}]
			(1,-1)node{$ß{50}$}(2,-1)node{$ß{51}$}
			(1,-2)node{$ß{52}$}(2,-2)node{$ß{53}$}
		;
	}
	\caption{
		The union of three strategies: $7·9$, $9·7$, and $49+4$.
		The gluing identifies $7·7$ (more concretely, $\SS11$--$\SS77$).
		Each of the first seven rows is $[9,8,2]$-MDS;
		so is every length-$9$ column $[9,8,2]$-MDS.
		The unshaded part and the parallelogramical-shaded part
		 together is $[53,52,2]$-MDS.
	}\label{fig:union}
\end{figure}

\subsection{Three-dimensional worker array}

	An $8×8$ matrix is
	\begin{itemize}
		\item	a $2×2$ matrix of $2×2$ matrices of $2×2$ matrices,
		\item	a $2×2$ matrix of $4×4$ matrices,
		\item	a $4×4$ matrix of $2×2$ matrices, and
		\item	a special data structure where
				the first $2×2$ and the third $2×2$	\\
				are treated as a $4×4$.
	\end{itemize}
	Different factorizations create endless compositions of error-correcting tactics.
	We discuss some in this section.
	
	\Cref{fig:cube} pictures a union of four strategies.
	In the middle is a $9×9×9$ array,
	which represent the third tensor power of $\S1$--$\S9$.
	Apart from the big cube there are three $2×2×9$ ``pillars''.
	Each $2×2$ slice of a pillar is a copy of $ß{50}$--$ß{53}$
	attached to the $7×7$ slice of $\SS11$--$\SS77$
	that shares a common coordinate;
	each pillar consists of $9$ such slices.
	
	Here is how we refer to those strategies.
	Label the axes $x$, $y$, and $z$ according to the right-hand law.
	The copies of $ß{50}$--$ß{53}$ perpendicular to the $x$-axis belongs to $9·53$.
	The replicas of $ß{50}$--$ß{53}$ perpendicular to the $z$-axis belongs to $53·9$.
	The duplicates of $ß{50}$--$ß{53}$ perpendicular to the $y$ axis
	belongs to $ρ(53·9)$, where $ρ$ is the rotation that sends
	the $x$-, $y$-, and $z$-axes to the $z$-, $x$-, and $y$-axes
	(in \cref{fig:cube}, it rotates $120^∘$ clockwise).
	In particular, $ρ(9·53)=53·9$ and $(ρ∘ρ)(53·9)=9·53$.
	
	Besides the aforementioned strategies, \cref{fig:cube}
	contains several interesting sub-strategies and inspires others.
	We cherry-pick some here.
	\begin{itemize}
		\item	$9·9·9$ can fix a single erasure on any axis.
				A smallest stopping set is a $2×2×2$ cube (eight unknowns).
		\item	$\Cyc7·7·9≔7·7·9∪7·9·7∪9·7·7$ is the counterpart of $7·9∪9·7$.
				A smallest stopping set is a ``claw'' with coordinates
				$(1,1,1)$, $(1,1,9)$, $(1,9,1)$, and $(9,1,1)$ (four unknowns).
		\item	$\Cyc7·53≔7·53∪53·7∪ρ(53·7)$ can fix a single erasure on any slice.
				A smallest stopping set is a $2×2$ square (four unknowns).
		\item	$\Cyc7·57≔7·57∪57·7∪ρ(57·7)$ can fix two erasures on any slice.
				A smallest stopping set is a $2×2×2$ cube minus a diagonal
				(six unknowns).
		\item	(\cref{fig:cube} itself) $9·9·9∪\Cyc9·53≔9·9·9∪9·53∪53·9∪ρ(53·9)$
				is such that every slice behaves like \cref{thm:4cycle}.
				A smallest stopping set is a $2×2×2$ cube plus
				all ``$ß{53}$'' whose slice intersects the cube (fourteen unknowns).
		\item	$7·9·7∪7·53∪53·7$ is such that every slice
				behaves like \cref{thm:2cycle}.
				A smallest stopping set is two on one of the $9$-axis
				plus the two ``$ß{53}$'' on the two slices that contain the axis.
	\end{itemize}
	
	We now call the end of the minimum distance flavor analysis
	and turn to a probabilistic analysis.

\begin{figure}
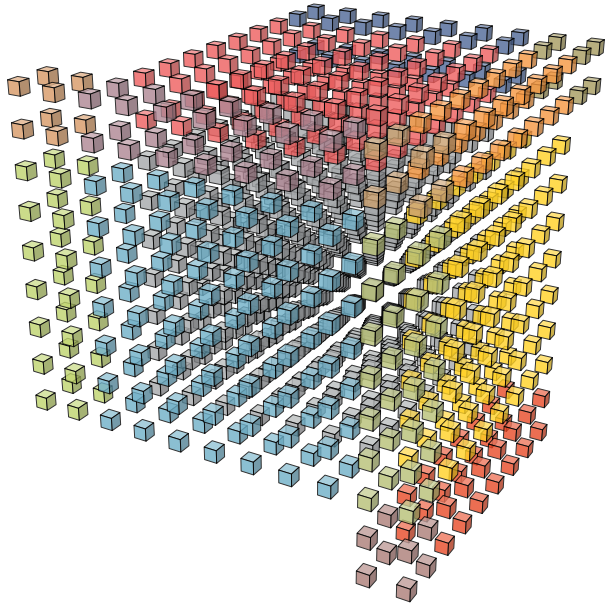
%\labortrue
	\makeatletter
	\colorlet{Scube0}{Cool Gray 6}
	\colorlet{Scube1}{Gray-blue}
	\colorlet{Scube2}{Dark yellow}
	\colorlet{Scube4}{Salmon}
	\colorlet{Scube3}{Scube1!50!Scube2}
	\colorlet{Scube5}{Scube1!50!Scube4}
	\colorlet{Scube6}{Scube2!50!Scube4}
	\colorlet{Scube7}{Scube3!66.66!Scube4}
	\colorlet{Scube8}{Citron}
	\colorlet{Scube9}{Scube8!50!Scube4}
	\colorlet{Scube10}{Periwinkle}
	\colorlet{Scube11}{Scube10!50!Scube2}
	\colorlet{Scube12}{UCO}
	\colorlet{Scube13}{Scube12!50!Scube1}
	\pgfplotsforeachungrouped\index in{0,...,13}{
		\colorlet{Scube\index+}{Scube\index!75!white}
		\colorlet{Scube\index-}{Scube\index!75!black}
	}
	\def\cameradistance{36}
	\def\perspectivexyz#1#2#3{%
		\pgfpointxyz{#1-4}{#2-4}{#3-4}%%% (4,4,4) is imaginary center
		\PMS\depth{\rcarot*\pgftemp@x+\rcbrot*\pgftemp@y+\rccrot*\pgftemp@z}%
		\PMS\depthrescale{\cameradistance/(\cameradistance-\depth)/2}%
		\pgf@x=\depthrescale\pgf@x%
		\pgf@y=\depthrescale\pgf@y%
	}
	\def\dec{.3333}
	\def\X{\x\dec}\def\Y{\y\dec}\def\Z{\z\dec}
	\def\remember#1#2{%
		\pgf@process{#2}%
		\edef#1{\noexpand\pgfqpoint{\the\pgf@x}{\the\pgf@y}}%
	}
	\def\filldrawcube{%
			\remember\xyZ{\perspectivexyz\x\y\Z}\remember\xYZ{\perspectivexyz\x\Y\Z}%
		\remember\XyZ{\perspectivexyz\X\y\Z}\remember\XYZ{\perspectivexyz\X\Y\Z}%
												\remember\xYz{\perspectivexyz\x\Y\z}%
		\remember\Xyz{\perspectivexyz\X\y\z}\remember\XYz{\perspectivexyz\X\Y\z}%
		\pgfpathmoveto{\XYZ}\pgfpathlineto{\xYZ}%
		\pgfpathlineto{\xyZ}\pgfpathlineto{\XyZ}%
		\pgfpathclose\pgfsetfillcolor{Scube\paletteindex+}\pgfusepath{fill,draw}%
		\pgfpathmoveto{\XYZ}\pgfpathlineto{\XyZ}%
		\pgfpathlineto{\Xyz}\pgfpathlineto{\XYz}%
		\pgfpathclose\pgfsetfillcolor{Scube\paletteindex}\pgfusepath{fill,draw}%
		\pgfpathmoveto{\XYZ}\pgfpathlineto{\xYZ}%
		\pgfpathlineto{\xYz}\pgfpathlineto{\XYz}%
		\pgfpathclose\pgfsetfillcolor{Scube\paletteindex-}\pgfusepath{fill,draw}%
	}
	\tdplotsetmaincoords{90-23.5}{120} % Formosa 120E 23.5N
	\tikz[tdplot_main_coords,opacity=.8,line width=.2]{
		\iflabor
		\foreach\x in{9,10}{
			\foreach\y in{0,1}{
				\foreach\z in{2,...,10}{
					\PMT\paletteindex{8+(\z>8)}
					\filldrawcube
				}
			}
		}
		\foreach\x in{0,1}{
			\foreach\y in{2,...,10}{
				\foreach\z in{9,10}{
					\PMT\paletteindex{10+(\y>8)}
					\filldrawcube
				}
			}
		}
		\foreach\x in{2,...,10}{
			\foreach\y in{9,10}{
				\foreach\z in{0,1}{
					\PMT\paletteindex{12+(\x>8)}
					\filldrawcube
				}
			}
		}
		\foreach\x in{2,...,10}{
			\message{<\x>}
			\foreach\y in{2,...,10}{
				\foreach\z in{2,...,10}{
					\PMT\paletteindex{(\x>8)+2*(\y>8)+4*(\z>8)}
					\filldrawcube
				}
			}
		}
		\else
		\pgfpathmoveto{\perspectivexyz{10\dec}{10\dec}0}
		\pgfpathlineto{\perspectivexyz2{10\dec}0}
		\pgfpathlineto{\perspectivexyz0{10\dec}{10\dec}}
		\pgfpathlineto{\perspectivexyz02{10\dec}}
		\pgfpathlineto{\perspectivexyz{10\dec}0{10\dec}}
		\pgfpathlineto{\perspectivexyz{10\dec}02}
		\pgfpathclose
		\pgfusepath{draw}
		\fi
	}
	\caption{
		The union of four strategies---%
		$9·9·9$ and $53·9$ along with $9·53$ together with $ρ(9·53)$.
		Slices in every direction fall back to \cref{fig:Elias} union \cref{fig:union},
		or equivalently, $9·9$ union $53$.
	}\label{fig:cube}
\end{figure}

\section{Probabilistic Analysis}\label{sec:sample}

	In this section, we want to go beyond the MDS guarantees of Pluto coding
	and explore its performance in a probabilistic fashion.
	We want to compare Pluto to prior constructions, particularly EPC and PDC.
	Below is the setup for the comparison.
	
	Assume that each worker finishes its task at a random time,
	and that no two workers report their results at the same time.
	(If so, break tie by any means.)
	We further assume that the workers can report in any order,
	and every order (i.e., every permutation of workers) has an equal chance to occur.
	Define the \emph{recovery count} to be the number of done workers
	when the manager receives enough data to assemble the matrix $C≔AB$.
	The recovery count is thus a random variable depending on the permutation.
	
	Recall that EPC has recovery threshold $ℓmn+m-1$ and PDC has $ℓ(2m-1)n$.
	What that means in this context is that, the recovery counts
	of the said codes are guaranteed to be at most the thresholds.
	According to how EPC and PDC are defined, it is unlikely that
	their recovery counts will ever be strictly less than the aforementioned thresholds.
	We assume equality for the ease of comparison.
	
	With the rule set, we can now sample some Pluto codes' recovery counts
	a couple of times and plot the empirical cumulative distribution functions.
	Starting from \cref{fig:2x2}, we do this for $2×2$, $4×4$, and $8×8$ matrices
	to showcase how is Strassen doing.
	We repeat this for $6×6$, $9×9$, and $12×12$ matrices to showcase Laderman.
	Rectangular tactics mentioned in \cref{sec:rectangle} are also included.
	It is worth noting that in \cref{fig:2x2,fig:3x3},
	the exact probability is calculated, so the vertical axis is in percentage.
	
	Among these figures, please identify strategies using the numbers of workers.
	For instance, $7+2$ means $⟨2,2,2;9⟩$, i.e., $\S1$--$\S9$;
	and $9·26$ means the tensor strategy $⟨2,2,2;9⟩⊗⟨3,3,3;26⟩$.
	Et cetera.

\section{Conclusions and Discussion}

	Through the present paper we propose a method to multiply matrices
	that relies on workers that suffer from straggling.
	We suggest the manager to begin with an FMM algorithm
	and create checksums via the tautological equation
	\[*gCh=(gA)(Bh).\]*
	The resulting calculation scheme is named Pluto code;
	more precisely, it is a prime Pluto code.
	A composite Pluto code is then obtained from ``tensoring primes''.
	Furthermore, we can take unions of Pluto codes
	granted that the dimensions match.
	
	\Cref{tab:prime} puts some prime Pluto codes together for comprehension.
	Its columns are:
	matrix dimension/type;
	schoolbook complexity;
	bilinear rank;
	number of workers;
	exponent as in FMM literature;
	Pluto's threshold;
	EPC's threshold \cite{YMA20};
	PDC's threshold \cite{DFHJCG20};
	EPC's secondary threshold \cite[section~VI]{YMA20}.
	For practical dimensions, the last quantity is worse than EPC's primary threshold.
	\Cref{tab:compose} book-keeps some composite Pluto codes
	to illustrate how to tensor primes.
	
	We say a Pluto code is \emph{wealthy} if the number of
	required workers does not exceed the naïve complexity.
	Tensoring wealthy Pluto codes will always yield a wealthy one.
	In other words, Pluto codes provide erasure-resiliency
	(or \emph{straggler-mitigation} as said in some references) in the \emph{high-rate}
	region, where the manager prepares fewer workers than the naïve algorithm needs.
	In the past, Laderman and rectangular algorithms are only of academic interest
	because Strassen is simply better in terms of exponents and overheads.
	Nonetheless, the Pluto construction now motivates the search for, say,
	$⟨3,3,3;21⟩$ and $⟨4,4,4;48⟩$, so they can afford inserting more checksums.
	
	A more comprehensive elaboration on \emph{code rate} goes as follows:
	Fix $ℓ,m,n≥2$.
	The plain multiplication requires $ℓmn$ entry multiplications.
	FMM brings this number down to $\ran⟨ℓ,m,n⟩$.
	A wealthy Pluto code takes advantage of the gap $ℓmn-\ran⟨ℓ,m,n⟩$
	by asking the saved workers to build up backups.
	Comparing wealthy Pluto to naïve, we argue that, in return for
	the complicated linear combinations, Pluto endures lazy workers.
	This region is what we call the high-rate region.
	In this region, neither EPC nor PDC functions;
	actually, no known codes seem to be considering this region.
	
	Beyond $ℓmn+m-1$, EPC starts functioning.
	In this region, termed \emph{low-rate} region, EPC's threshold outperforms Pluto's.
	However, Pluto owns at least two advantages:
	One, it could still outcompete EPC under a probabilistic model.
	Two, whereas both EPC and Pluto rely heavily on linear combinations,
	the Pluto ones are ``shorter'' and ``lighter''.
	``Shorter'' implies that, over characteristic zero fields,
	Pluto poses less rounding errors than EPC does;
	and ``lighter'' means that Pluto coding works over smaller finite fields.
	Lastly, the TPC structure within composite Pluto respects worker locality
	in the network topology---especially, the decoding is highly parallelizable.
	See \cref{fig:worker} for an image summary.
	
	\cite[section~VI]{YMA20} proposed an EPC variant with threshold $2\ran⟨ℓ,m,n⟩-1$,
	which looks asymptotically better then $ℓmn+m-1$.
	According to \cite{Sedoglavic20}, however, the former is not less than the latter
	until $ℓ=m=n=28$, at which point we see $2·10{,}556-1<21{,}979$.
	Note that this construction now assumes a Reed--Solomon code of length $>20{,}000$.
	Over real numbers, the decoding of such a long code suffers from numerical errors;
	we do not expect that to be used in practical scenarios (but see \cite{FC19}).
	In spite of that, Pluto provides code lengths $10{,}556+28$, $10{,}556+56$,
	$10{,}556+84$, et seq, which transcends EPC again.
	
	In \cite{GWCR18} it is commented that around $5\%$ of workers straggle,
	and this is for AWS Lambda.
	For high-performance computation, the straggle probability is even more minuscule.
	We foresee that Pluto applies to these scenarios.

%\labortrue
\begin{figure}%\labortrue
	\iflabor
	\tikz{
		\def\plotnum{0}\def\numplots{0}
		\begin{axis}[
			legend pos=south east,
			every axis plot/.style={xshift=(2*\plotnum+1-\numplots)*1.5},
			yticklabel={\PMT\pers{round(\tick*100)}\ifnum\pers=100 1\else\pers\%\fi}
		]
			\addlegendentry{$7+2$}\addplot+coordinates{
				(7,1/36)(8,1)};
			\addlegendentry{$7+4$}\addplot+coordinates{
				(7,1/330)(8,7/55)(9,1)};
			\addlegendentry{EPC $9$}\addplot+coordinates{(9,0)(9,1)};
			\addlegendentry{$7+6$}\addplot+coordinates{
				(7,1/1716)(8,1/33)(9,3/11)(10,1)};
			\addlegendentry{$7+8$}\addplot+coordinates{
				(7,1/6435)(8,29/2145)(9,141/1001)(10,2591/3003)(11,1)};
			\addlegendentry{$7+10$}\addplot+coordinates{
				(7,1/19448)(8,32/2431)(9,316/2431)(10,7363/9724)(11,5847/6188)(12,1)};
			\addlegendentry{$7+12$}\addplot+coordinates{
				(7,1/50388)(8,682/37791)(9,6768/46189)(10,31703/46189)(11,1771/1989)
				(12,12334/12597)(13,1)};
			\addlegendentry{PDC $12$}\addplot+coordinates{(12,0)(12,1)};
			\addlegendentry{EPC $13$}\addplot+coordinates{(13,0)(13,1)};
		\end{axis}
	}
	\fi
	\caption{
		Pluto codes to multiply $2×2$ matrices and the CDFs of their recovery counts.
		The vertical axis is the cumulative probability.
		The horizontal axis is the recovery count (always integer);
		plots are shifted horizontally to avoid collision.
		The vertical bar at $9$ is EPC's threshold;
		the bar at $12$ is PDC's threshold;
		at $13$ is EPC's secondary threshold.
	}\label{fig:2x2}
\end{figure}
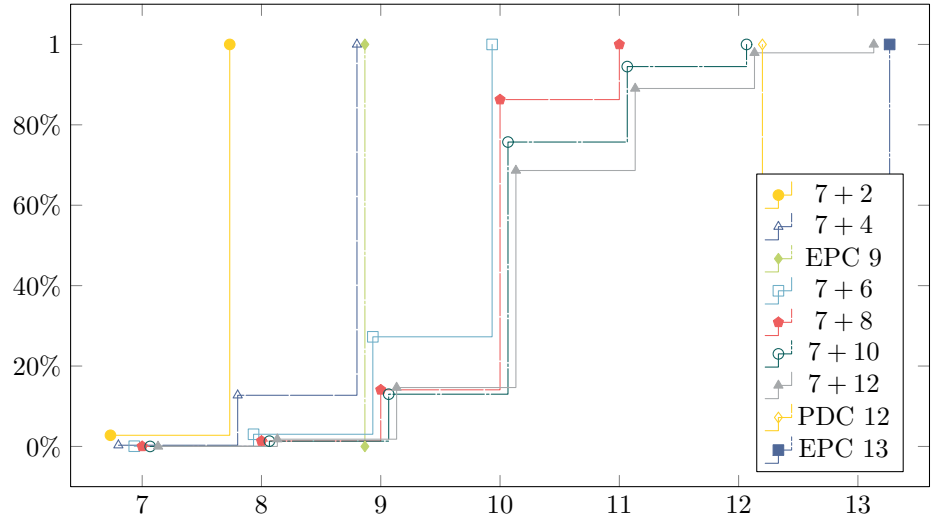

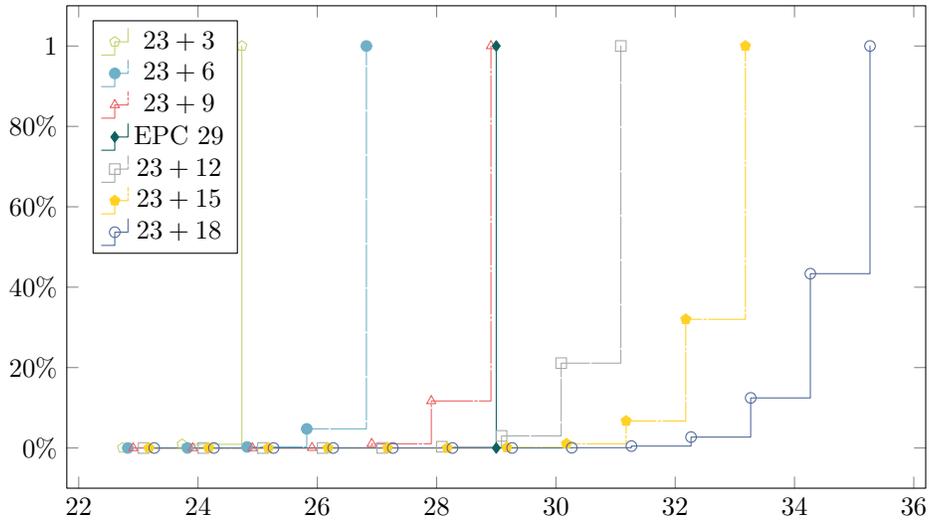
\begin{figure}%\labortrue
	\iflabor
	\tikz{
		\def\plotnum{0}\def\numplots{0}
		\begin{axis}[
			legend pos=north west,
			every axis plot/.style={xshift=(2*\plotnum+1-\numplots)*1},
			yticklabel={\PMT\pers{round(\tick*100)}\ifnum\pers=100 1\else\pers\%\fi}
		]
			\addlegendentry{$23+3$}\addplot+coordinates{
				(23,1/2600)(24,3/325)(25,1)};
			\addlegendentry{$23+6$}\addplot+coordinates{
				(23,1/475020)(24,2/39585)(25,61/23751)(26,1/21)(27,1)};
			\addlegendentry{$23+9$}\addplot+coordinates{
				(23,1/28048800)(24,1/1168700)(25,35/1121952)(26,275/453096)
				(27,263/25172)(28,526/4495)(29,1)};
			\addlegendentry{EPC $29$}\addplot+coordinates{(29,0)(29,1)};
			\addlegendentry{$23+12$}\addplot+coordinates{
				(23,1/834451800)(24,1/34768825)(25,79/91789698)(26,2/121737)
				(27,175/672452)(28,10511/3362260)(29,4885/162316)(30,8567/40579)(31,1)};
			\addlegendentry{$23+15$}\addplot+coordinates{
				(23,1/15471286560)(24,1/644636940)(25,55/1353737574)(26,55/73175004)
				(27,13649/1203322288)(28,64223/472733756)(29,108529/81505820)
				(30,516829/48903492)(31,282635/4206752)(32,883649/2760681)(33,1)};
			\addlegendentry{$23+18$}\addplot+coordinates{
				(23,1/202112640600)(24,1/8421360025)(25,97/34359148902)
				(26,1603/31716137448)(27,839/1136779120)(28,153821/17620076360)
				(29,337607/3949327460)(30,1103485/1579730984)(31,1340589/280274852)
				(32,9514938/350343565)(33,793664/6369883)(34,374981/864690)(35,1)};
		\end{axis}
	}
	\fi
	\caption{
		Pluto codes to multiply $3×3$ matrices and the CDFs of their recovery count.
		Plots are shifted horizontally to avoid collision.
		The line at $29$ is EPC's threshold.
		PDC's and EPC's secondary thresholds are at $45$, too far to be included.
	}\label{fig:3x3}
\end{figure}

\begin{figure}%\labortrue
	\iflabor
	\tikz{
		\begin{axis}[
			legend pos=north west,
		]
			\addlegendentry{$49+4$}\addplot+coordinates{
				(50,8)(51,217)(52,50000)};
			\addlegendentry{$49+8$}\addplot+coordinates{
				(52,1)(53,58)(54,1172)(55,50000)};
			\addlegendentry{$49+12$}\addplot+coordinates{
				(55,11)(56,217)(57,3053)(58,50000)};
			\addlegendentry{Charon $63$}\addplot+coordinates{
				(56,4)(57,84)(58,2570)(59,49244)(60,49896)(61,50000)};
			\addlegendentry{$7·9$}\addplot+coordinates{
				(54,3)(55,59)(56,685)(57,3331)(58,9613)(59,20038)(60,32602)(61,43711)
				(62,50000)};
			\addlegendentry{$49+16$}\addplot+coordinates{
				(57,3)(58,50)(59,603)(60,5852)(61,50000)};
			\addlegendentry{$7·9∪53$}\addplot+coordinates{
				(56,4)(57,24)(58,295)(59,2071)(60,7869)(61,18762)(62,32065)(63,42866)
				(64,48513)(65,50000)};
			\addlegendentry{$49+20$}\addplot+coordinates{
				(59,1)(60,17)(61,161)(62,1353)(63,9431)(64,50000)};
			\addlegendentry{$49+24$}\addplot+coordinates{
				(62,7)(63,57)(64,414)(65,2608)(66,13310)(67,50000)};
			\addlegendentry{EPC $67$}\addplot+coordinates{(67,0)(67,50000)};
		\end{axis}
	}
	\fi
	\caption{
		Pluto for $4×4$ matrices---first half.
		The vertical axis is the empirical accumulated number of samples.
		Charon is a 63-worker tactic constructed in \cref{app:Charon}:
		it applies FMM to the computation of checksums to save workers.
	}\label{fig:4x4win}
\end{figure}
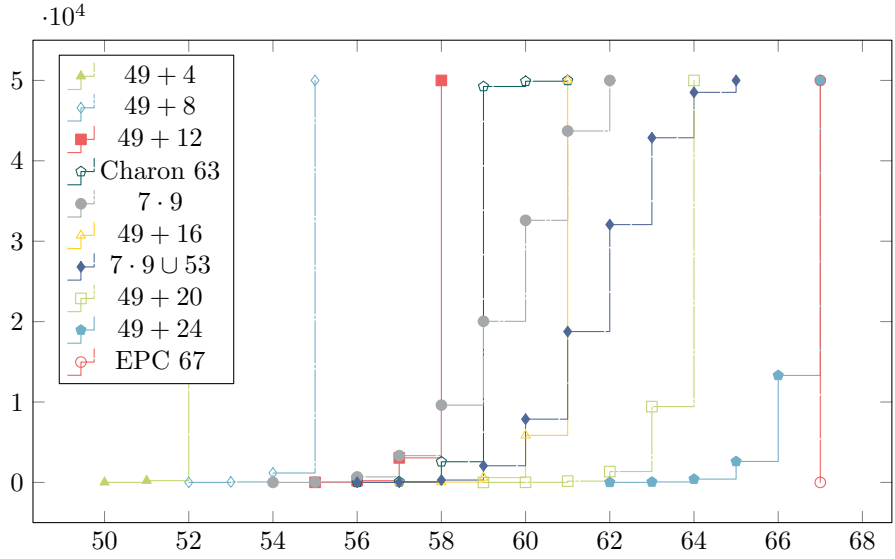

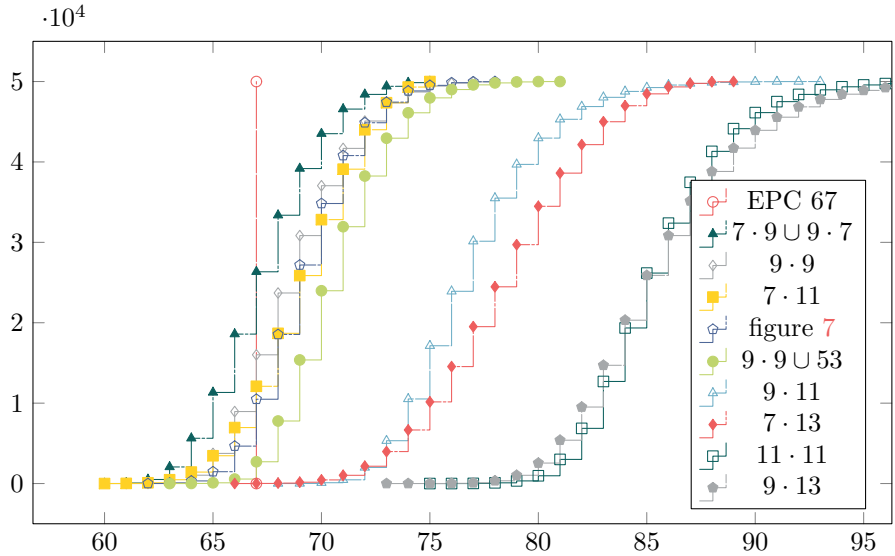
\begin{figure}%\labortrue
	\iflabor
	\tikz{
		\begin{axis}[
			legend pos=south east,
		]
			\xdef\accumulatenumplots{\the\numexpr\accumulatenumplots-1}
			\addlegendentry{EPC $67$}\addplot+coordinates{(67,0)(67,50000)};
			\addlegendentry{$7·9∪9·7$}\addplot+coordinates{
				(60,6)(61,70)(62,506)(63,2064)(64,5632)(65,11330)(66,18588)(67,26339)
				(68,33372)(69,39178)(70,43515)(71,46575)(72,48399)(73,49407)(74,49864)
				(75,50000)};
			\addlegendentry{$9·9$}\addplot+coordinates{
				(61,1)(62,10)(63,88)(64,1037)(65,3749)(66,8957)(67,16012)(68,23702)
				(69,30841)(70,37043)(71,41685)(72,45089)(73,47342)(74,48699)(75,49435)
				(76,49809)(77,49960)(78,50000)};
			\addlegendentry{$7·11$}\addplot+coordinates{
				(60,2)(61,21)(62,98)(63,475)(64,1427)(65,3459)(66,6963)(67,12098)
				(68,18685)(69,25867)(70,32825)(71,39100)(72,44005)(73,47367)(74,49314)
				(75,50000)};
			\addlegendentry{\cref{fig:union}}\addplot+coordinates{
				(62,5)(63,59)(64,328)(65,1474)(66,4656)(67,10487)(68,18544)(69,27182)
				(70,34830)(71,40791)(72,44876)(73,47438)(74,48833)(75,49525)(76,49860)
				(77,49977)(78,50000)};
			\addlegendentry{$9·9∪53$}\addplot+coordinates{
				(63,2)(64,14)(65,94)(66,561)(67,2711)(68,7780)(69,15374)(70,23980)
				(71,31945)(72,38248)(73,42953)(74,46113)(75,47970)(76,49005)(77,49588)
				(78,49830)(79,49953)(80,49996)(81,50000)};
			\addlegendentry{$9·11$}\addplot+coordinates{
				(68,4)(69,13)(70,104)(71,465)(72,2003)(73,5329)(74,10522)(75,17140)
				(76,23910)(77,30138)(78,35503)(79,39694)(80,42966)(81,45304)(82,46889)
				(83,48035)(84,48770)(85,49251)(86,49565)(87,49761)(88,49891)(89,49949)
				(90,49981)(91,49996)(92,49999)(93,50000)};
			\addlegendentry{$7·13$}\addplot+coordinates{
				(66,3)(67,15)(68,55)(69,166)(70,453)(71,1026)(72,2155)(73,3977)(74,6664)
				(75,10157)(76,14547)(77,19514)(78,24471)(79,29710)(80,34483)(81,38609)
				(82,42152)(83,45001)(84,46991)(85,48468)(86,49333)(87,49774)(88,49966)
				(89,50000)};
			\addlegendentry{$11·11$}\addplot+[update limits=false]coordinates{
				(75,1)(76,2)(77,19)(78,90)(79,329)(80,968)(81,3003)(82,6864)(83,12697)
				(84,19351)(85,26179)(86,32390)(87,37473)(88,41315)(89,44138)(90,46136)
				(91,47501)(92,48398)(93,48970)(94,49347)(95,49578)(96,49745)(97,49854)
				(98,49912)(99,49946)(100,49962)(101,49979)(102,49988)(103,49993)
				(104,49997)(105,49998)(107,49999)(108,50000)};
			\addlegendentry{$9·13$}\addplot+[update limits=false]coordinates{
				(73,1)(74,3)(75,6)(76,26)(77,97)(78,345)(79,1027)(80,2548)(81,5391)
				(82,9516)(83,14708)(84,20318)(85,25891)(86,30843)(87,35153)(88,38815)
				(89,41722)(90,43932)(91,45566)(92,46851)(93,47773)(94,48427)(95,48899)
				(96,49249)(97,49498)(98,49669)(99,49785)(100,49867)(101,49923)
				(102,49954)(103,49970)(104,49984)(105,49991)(106,49997)(107,50000)};
		\end{axis}
	}
	\fi
	\caption{
		Pluto for $4×4$ matrices---second half.
		Whereas these recovery counts are in general worse than EPC's 67,
		we argue that Pluto's pre- and post-processes are more lightweight.
	}\label{fig:4x4bad}
\end{figure}

\begin{figure}%\labortrue
	\iflabor
	\tikz{
		\begin{axis}[
			legend pos=south east,
			scaled y ticks=base 10:-3,
		]
			\addlegendentry{$7·29$}\addplot+coordinates{
				(188,6)(189,62)(190,272)(191,829)(192,1985)(193,3710)(194,6188)
				(195,9078)(196,12031)(197,14830)(198,17116)(199,18743)(200,19656)
				(201,20000)};
			\addlegendentry{$9·23$}\addplot+coordinates{
				(188,2)(189,11)(190,19)(191,41)(192,121)(193,270)(194,524)(195,977)
				(196,1682)(197,2759)(198,4251)(199,6139)(200,8426)(201,10840)(202,13321)
				(203,15829)(204,17808)(205,19254)(206,20000)};
			\addlegendentry{$9·26$}\addplot+coordinates{
				(202,7)(203,28)(204,67)(205,150)(206,370)(207,760)(208,1422)(209,2309)
				(210,3584)(211,5135)(212,6841)(213,8602)(214,10485)(215,12192)
				(216,13673)(217,15047)(218,16202)(219,17110)(220,17908)(221,18508)
				(222,18960)(223,19320)(224,19571)(225,19753)(226,19855)(227,19932)
				(228,19969)(229,19990)(230,19998)(231,20000)};
			\addlegendentry{$(13·18)$}\addplot+coordinates{
				(204,5)(205,26)(206,125)(207,374)(208,892)(209,1727)(210,2890)(211,4312)
				(212,6011)(213,7851)(214,9750)(215,11651)(216,13385)(217,14881)
				(218,16121)(219,17112)(220,17891)(221,18524)(222,19044)(223,19379)
				(224,19602)(225,19747)(226,19843)(227,19919)(228,19965)(229,19992)
				(230,19997)(231,20000)};
			\addlegendentry{$(14·17)$}\addplot+coordinates{
				(207,1)(208,9)(209,45)(210,148)(211,416)(212,917)(213,1832)(214,3077)
				(215,4682)(216,6406)(217,8279)(218,10189)(219,12012)(220,13650)
				(221,15079)(222,16286)(223,17227)(224,17985)(225,18549)(226,19022)
				(227,19357)(228,19600)(229,19754)(230,19863)(231,19936)(232,19968)
				(233,19987)(234,19999)(235,20000)};
			\addlegendentry{EPC $221$}\addplot+coordinates{(221,0)(221,20000)};
			\addlegendentry{$9·29$}\addplot+coordinates{
				(215,2)(217,16)(218,49)(219,157)(220,423)(221,908)(222,1677)(223,2827)
				(224,4309)(225,6016)(226,7782)(227,9670)(228,11414)(229,12980)
				(230,14330)(231,15490)(232,16508)(233,17289)(234,17956)(235,18424)
				(236,18826)(237,19138)(238,19355)(239,19545)(240,19665)(241,19770)
				(242,19833)(243,19889)(244,19925)(245,19953)(246,19968)(247,19982)
				(248,19989)(249,19993)(250,19998)(251,19999)(254,20000)};
			\addlegendentry{$11·23$}\addplot+coordinates{
				(215,2)(216,3)(217,9)(218,24)(219,54)(220,87)(221,139)(222,204)(223,333)
				(224,493)(225,717)(226,1034)(227,1420)(228,1942)(229,2573)(230,3297)
				(231,4126)(232,5130)(233,6222)(234,7410)(235,8618)(236,9923)(237,11250)
				(238,12472)(239,13628)(240,14756)(241,15804)(242,16703)(243,17491)
				(244,18183)(245,18748)(246,19182)(247,19514)(248,19753)(249,19910)
				(250,19976)(251,20000)};
			\addlegendentry{$11·26$}\addplot+[update limits=false]coordinates{
				(225,1)(227,3)(228,4)(229,13)(230,27)(231,65)(232,124)(233,243)(234,439)
				(235,729)(236,1147)(237,1727)(238,2481)(239,3453)(240,4552)(241,5790)
				(242,7127)(243,8553)(244,9856)(245,11138)(246,12364)(247,13536)
				(248,14564)(249,15521)(250,16284)(251,16974)(252,17571)(253,18033)
				(254,18434)(255,18745)(256,19000)(257,19225)(258,19414)(259,19549)
				(260,19654)(261,19734)(262,19805)(263,19852)(264,19895)(265,19931)
				(266,19956)(267,19971)(268,19985)(269,19990)(270,19996)(272,19997)
				(273,19999)(275,20000)};
%			\addlegendentry{$11·29$}\addplot+[update limits=false]coordinates{
%				(243,1)(244,5)(245,9)(246,22)(247,59)(248,140)(249,269)(250,507)
%				(251,913)(252,1501)(253,2271)(254,3253)(255,4420)(256,5814)(257,7366)
%				(258,8935)(259,10574)(260,12029)(261,13448)(262,14670)(263,15722)
%				(264,16575)(265,17330)(266,17917)(267,18383)(268,18781)(269,19096)
%				(270,19322)(271,19499)(272,19613)(273,19717)(274,19775)(275,19833)
%				(276,19866)(277,19897)(278,19924)(279,19941)(280,19959)(281,19966)
%				(282,19978)(283,19984)(284,19991)(285,19995)(287,19997)(288,19998)
%				(290,19999)(292,20000)};
		\end{axis}
	}
	\fi
	\caption{
		$6×6$.
		Between the parentheses are the rectangular syntheses
		$⟨2,2,3;13⟩⊗⟨3,3,2;18⟩$ and $⟨2,3,2;14⟩⊗⟨3,2,3;17⟩$,.
	}\label{fig:6x6}
\end{figure}
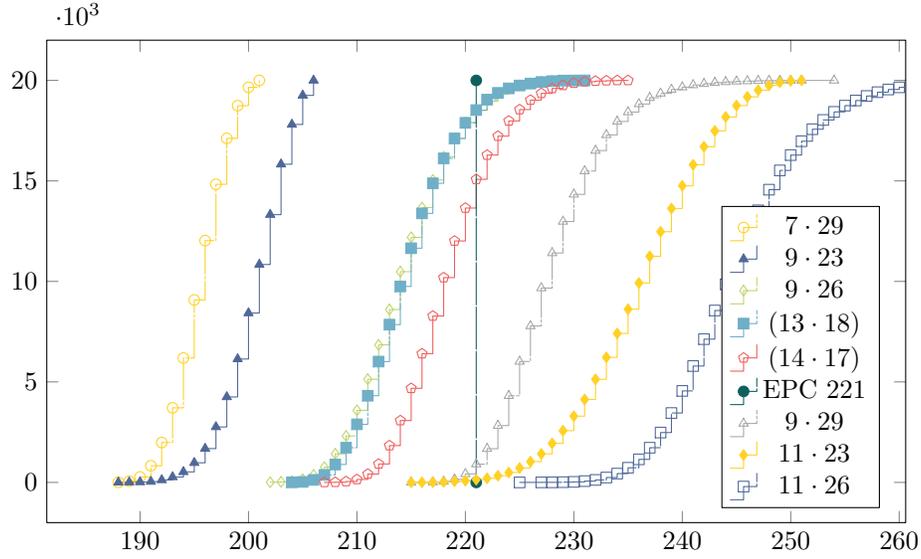

\begin{figure}%\labortrue
	\iflabor
	\tikz{
		\begin{axis}[
			legend pos=south east,
			scaled y ticks=base 10:-3,
		]
			\addlegendentry{$\Cyc7·53$}\addplot+coordinates{
				(404,4)(405,28)(406,158)(407,577)(408,1468)(409,2734)(410,4189)
				(411,5655)(412,6909)(413,7894)(414,8636)(415,9139)(416,9501)(417,9716)
				(418,9844)(419,9916)(420,9959)(421,9985)(422,9997)(423,9999)
				(424,10000)};
			\addlegendentry{two $53$ a $57$}\addplot+coordinates{
				(424,3)(425,18)(426,47)(427,165)(428,447)(429,1025)(430,1956)(431,3182)
				(432,4511)(433,5786)(434,6926)(435,7844)(436,8548)(437,9039)(438,9395)
				(439,9625)(440,9761)(441,9841)(442,9900)(443,9940)(444,9966)(445,9988)
				(446,9994)(447,10000)};
			\addlegendentry{$9·53$}\addplot+coordinates{
				(430,2)(431,5)(432,11)(433,17)(434,33)(435,63)(436,104)(437,179)
				(438,287)(439,450)(440,658)(441,931)(442,1278)(443,1680)(444,2136)
				(445,2675)(446,3215)(447,3815)(448,4402)(449,5005)(450,5614)(451,6173)
				(452,6695)(453,7267)(454,7744)(455,8134)(456,8481)(457,8743)(458,8999)
				(459,9231)(460,9405)(461,9534)(462,9654)(463,9741)(464,9818)(465,9873)
				(466,9911)(467,9947)(468,9969)(469,9984)(470,9995)(471,10000)};
			\addlegendentry{a $53$ two $57$}\addplot+coordinates{
				(444,1)(445,5)(446,19)(447,71)(448,184)(449,403)(450,837)(451,1534)
				(452,2486)(453,3600)(454,4818)(455,6002)(456,7093)(457,7920)(458,8591)
				(459,9041)(460,9366)(461,9613)(462,9760)(463,9865)(464,9920)(465,9950)
				(466,9966)(467,9979)(468,9982)(469,9989)(470,9996)(471,9998)
				(473,10000)};
			\addlegendentry{$9·57$}\addplot+coordinates{
				(450,2)(452,5)(453,17)(454,39)(455,69)(456,123)(457,214)(458,345)
				(459,581)(460,867)(461,1240)(462,1672)(463,2199)(464,2778)(465,3386)
				(466,4042)(467,4678)(468,5345)(469,5955)(470,6523)(471,7094)(472,7533)
				(473,7945)(474,8258)(475,8570)(476,8836)(477,9049)(478,9223)(479,9382)
				(480,9515)(481,9622)(482,9694)(483,9759)(484,9800)(485,9846)(486,9872)
				(487,9902)(488,9918)(489,9937)(490,9953)(491,9964)(492,9975)(493,9986)
				(494,9994)(495,9995)(496,9996)(497,9998)(499,10000)};
			\addlegendentry{$\Cyc7·57$}\addplot+coordinates{
				(465,3)(466,11)(467,40)(468,100)(469,198)(470,410)(471,777)(472,1323)
				(473,2108)(474,3076)(475,4195)(476,5298)(477,6386)(478,7325)(479,8142)
				(480,8759)(481,9194)(482,9491)(483,9683)(484,9801)(485,9889)(486,9932)
				(487,9962)(488,9980)(489,9988)(490,9993)(491,9996)(492,9997)(493,9998)
				(496,10000)};
			\addlegendentry{$7·9·9$}\addplot+coordinates{
				(474,2)(475,6)(477,7)(478,8)(479,12)(480,15)(481,28)(482,34)(483,41)
				(484,58)(485,81)(486,104)(487,132)(488,167)(489,212)(490,259)(491,309)
				(492,384)(493,495)(494,597)(495,717)(496,837)(497,981)(498,1131)
				(499,1285)(500,1461)(501,1655)(502,1861)(503,2102)(504,2330)(505,2552)
				(506,2802)(507,3080)(508,3333)(509,3623)(510,3885)(511,4182)(512,4476)
				(513,4788)(514,5063)(515,5304)(516,5568)(517,5833)(518,6112)(519,6377)
				(520,6636)(521,6874)(522,7109)(523,7328)(524,7560)(525,7768)(526,7952)
				(527,8141)(528,8310)(529,8498)(530,8666)(531,8797)(532,8944)(533,9063)
				(534,9194)(535,9279)(536,9346)(537,9435)(538,9508)(539,9558)(540,9624)
				(541,9683)(542,9729)(543,9768)(544,9801)(545,9840)(546,9868)(547,9894)
				(548,9919)(549,9935)(550,9954)(551,9962)(552,9972)(553,9981)(554,9986)
				(555,9989)(556,9994)(557,9996)(558,9998)(559,9999)(560,10000)};
			\addlegendentry{EPC $519$}\addplot+coordinates{(519,0)(519,10000)};
			\addlegendentry{$11·53$}\addplot+coordinates{
				(482,1)(491,2)(492,3)(493,4)(494,11)(495,13)(496,20)(497,37)(498,53)
				(499,67)(500,106)(501,153)(502,208)(503,277)(504,365)(505,481)(506,641)
				(507,815)(508,1019)(509,1263)(510,1564)(511,1886)(512,2215)(513,2539)
				(514,2960)(515,3378)(516,3784)(517,4210)(518,4637)(519,5069)(520,5481)
				(521,5862)(522,6259)(523,6629)(524,6952)(525,7269)(526,7592)(527,7857)
				(528,8111)(529,8333)(530,8541)(531,8748)(532,8901)(533,9040)(534,9168)
				(535,9276)(536,9373)(537,9469)(538,9550)(539,9606)(540,9669)(541,9715)
				(542,9762)(543,9802)(544,9830)(545,9862)(546,9885)(547,9907)(548,9925)
				(549,9934)(550,9945)(551,9960)(552,9964)(553,9971)(554,9979)(555,9986)
				(556,9990)(557,9991)(558,9995)(559,9996)(560,9997)(563,10000)};
			\addlegendentry{$11·57$}\addplot+coordinates{
				(510,1)(515,3)(516,5)(517,7)(518,19)(519,33)(520,56)(521,84)(522,128)
				(523,189)(524,264)(525,348)(526,497)(527,651)(528,827)(529,1087)
				(530,1361)(531,1702)(532,2058)(533,2487)(534,2940)(535,3424)(536,3881)
				(537,4388)(538,4928)(539,5411)(540,5848)(541,6292)(542,6767)(543,7232)
				(544,7586)(545,7920)(546,8219)(547,8481)(548,8737)(549,8966)(550,9136)
				(551,9274)(552,9384)(553,9484)(554,9580)(555,9662)(556,9740)(557,9787)
				(558,9825)(559,9855)(560,9876)(561,9893)(562,9914)(563,9931)(564,9944)
				(565,9953)(566,9962)(567,9974)(568,9980)(569,9982)(570,9985)(571,9990)
				(572,9993)(573,9995)(574,9998)(585,9999)(589,10000)};
			\addlegendentry{$\Cyc7·7·9$}\addplot+[update limits=false]coordinates{
				(502,1)(504,4)(505,5)(506,10)(508,11)(509,15)(510,26)(511,32)(512,39)
				(513,54)(514,64)(515,75)(516,96)(517,109)(518,134)(519,163)(520,198)
				(521,230)(522,266)(523,306)(524,362)(525,411)(526,479)(527,551)(528,632)
				(529,720)(530,812)(531,882)(532,985)(533,1113)(534,1250)(535,1371)
				(536,1504)(537,1647)(538,1781)(539,1928)(540,2071)(541,2201)(542,2388)
				(543,2561)(544,2714)(545,2894)(546,3091)(547,3260)(548,3449)(549,3644)
				(550,3828)(551,4031)(552,4220)(553,4407)(554,4604)(555,4789)(556,4990)
				(557,5174)(558,5343)(559,5534)(560,5727)(561,5911)(562,6078)(563,6251)
				(564,6407)(565,6578)(566,6744)(567,6907)(568,7065)(569,7189)(570,7342)
				(571,7495)(572,7637)(573,7760)(574,7892)(575,8007)(576,8133)(577,8244)
				(578,8354)(579,8464)(580,8560)(581,8659)(582,8763)(583,8838)(584,8926)
				(585,8984)(586,9059)(587,9120)(588,9182)(589,9244)(590,9293)(591,9347)
				(592,9412)(593,9452)(594,9516)(595,9556)(596,9595)(597,9638)(598,9659)
				(599,9691)(600,9725)(601,9760)(602,9788)(603,9821)(604,9841)(605,9855)
				(606,9876)(607,9890)(608,9900)(609,9909)(610,9922)(611,9938)(612,9944)
				(613,9953)(614,9959)(615,9965)(616,9973)(617,9975)(618,9981)(619,9987)
				(620,9990)(621,9994)(622,9996)(623,9997)(624,9999)(627,10000)};
%			\addlegendentry{$\Cyc7·9·9$}\addplot+[update limits=false]coordinates{
%				(522,3)(523,12)(524,24)(525,29)(526,42)(527,60)(528,89)(529,130)
%				(530,186)(531,260)(532,348)(533,467)(534,581)(535,718)(536,881)
%				(537,1056)(538,1249)(539,1452)(540,1670)(541,1895)(542,2101)(543,2361)
%				(544,2611)(545,2846)(546,3121)(547,3363)(548,3619)(549,3866)(550,4121)
%				(551,4356)(552,4545)(553,4761)(554,5010)(555,5231)(556,5440)(557,5679)
%				(558,5868)(559,6041)(560,6239)(561,6435)(562,6607)(563,6790)(564,6950)
%				(565,7104)(566,7267)(567,7398)(568,7531)(569,7661)(570,7783)(571,7885)
%				(572,8009)(573,8131)(574,8242)(575,8329)(576,8424)(577,8518)(578,8611)
%				(579,8698)(580,8769)(581,8849)(582,8920)(583,8992)(584,9055)(585,9112)
%				(586,9166)(587,9220)(588,9267)(589,9301)(590,9356)(591,9391)(592,9426)
%				(593,9468)(594,9498)(595,9525)(596,9556)(597,9585)(598,9615)(599,9638)
%				(600,9658)(601,9669)(602,9692)(603,9709)(604,9731)(605,9756)(606,9772)
%				(607,9784)(608,9803)(609,9816)(610,9828)(611,9843)(612,9853)(613,9861)
%				(614,9873)(615,9882)(616,9888)(617,9897)(618,9903)(619,9908)(620,9920)
%				(621,9923)(622,9929)(623,9935)(624,9941)(625,9946)(626,9950)(627,9955)
%				(628,9956)(629,9959)(630,9960)(631,9964)(632,9965)(633,9967)(634,9970)
%				(635,9973)(636,9976)(637,9977)(638,9979)(639,9980)(640,9986)(642,9987)
%				(643,9988)(644,9990)(645,9993)(647,9994)(651,9995)(653,9996)(656,9998)
%				(659,9999)(664,10000)};
			\addlegendentry{$9·9·9$}\addplot+[update limits=false]coordinates{
				(519,1)(521,2)(522,7)(523,13)(524,16)(525,21)(526,38)(527,59)(528,97)
				(529,132)(530,173)(531,238)(532,320)(533,437)(534,560)(535,711)(536,861)
				(537,1023)(538,1218)(539,1409)(540,1643)(541,1872)(542,2125)(543,2384)
				(544,2632)(545,2883)(546,3107)(547,3376)(548,3628)(549,3896)(550,4140)
				(551,4394)(552,4670)(553,4894)(554,5115)(555,5328)(556,5539)(557,5727)
				(558,5940)(559,6119)(560,6302)(561,6476)(562,6652)(563,6823)(564,6989)
				(565,7139)(566,7278)(567,7420)(568,7543)(569,7679)(570,7790)(571,7896)
				(572,8011)(573,8121)(574,8232)(575,8333)(576,8412)(577,8505)(578,8601)
				(579,8685)(580,8764)(581,8838)(582,8914)(583,8968)(584,9030)(585,9088)
				(586,9141)(587,9198)(588,9248)(589,9280)(590,9322)(591,9373)(592,9410)
				(593,9456)(594,9491)(595,9529)(596,9551)(597,9575)(598,9603)(599,9632)
				(600,9656)(601,9678)(602,9706)(603,9729)(604,9741)(605,9758)(606,9775)
				(607,9793)(608,9811)(609,9823)(610,9828)(611,9835)(612,9842)(613,9856)
				(614,9866)(615,9878)(616,9889)(617,9896)(618,9898)(619,9901)(620,9908)
				(621,9913)(622,9918)(623,9927)(624,9935)(625,9940)(626,9945)(627,9946)
				(628,9949)(629,9951)(630,9954)(631,9962)(632,9964)(633,9967)(634,9968)
				(635,9973)(636,9976)(637,9977)(638,9978)(639,9979)(641,9980)(642,9982)
				(643,9983)(644,9985)(645,9986)(646,9987)(650,9989)(651,9990)(652,9992)
				(654,9994)(656,9997)(657,9998)(663,9999)(667,10000)};
		\end{axis}
	}
	\fi
	\caption{
		$8×8$.
		``Two $53$ a $57$'' means $7·53∪53·7∪ρ(57·7)$.
		``A $53$ two $57$'' means $7·53∪57·7∪ρ(57·7)$.
	}\label{fig:8x8}
\end{figure}
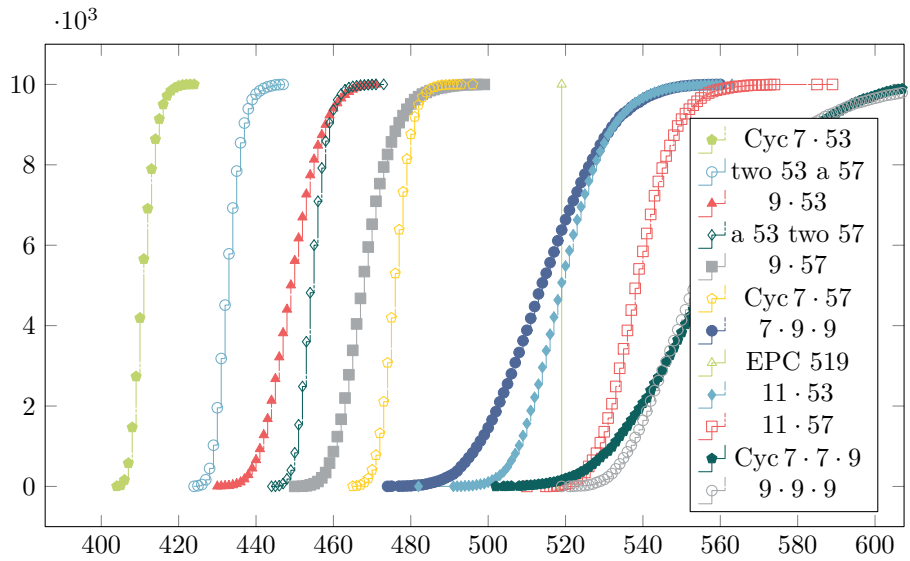

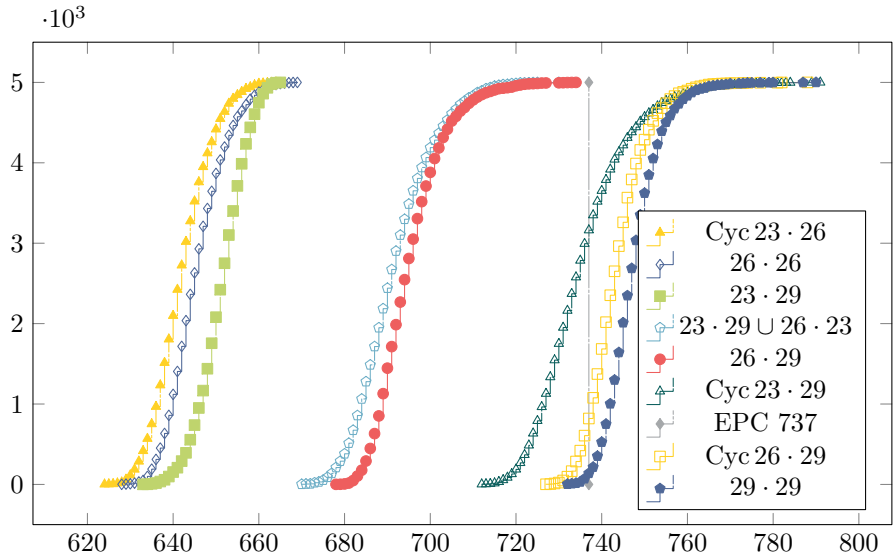
\begin{figure}%\labortrue
	\iflabor
	\tikz{
		\begin{axis}[
			legend pos=south east,
			scaled y ticks=base 10:-3,
		]
			\addlegendentry{$\Cyc23·26$}\addplot+coordinates{
				(624,1)(625,3)(626,5)(627,14)(628,26)(629,61)(630,114)(631,186)(632,286)
				(633,417)(634,569)(635,751)(636,967)(637,1234)(638,1513)(639,1805)
				(640,2098)(641,2423)(642,2725)(643,3019)(644,3275)(645,3521)(646,3760)
				(647,3949)(648,4120)(649,4258)(650,4417)(651,4547)(652,4636)(653,4737)
				(654,4798)(655,4844)(656,4883)(657,4921)(658,4941)(659,4966)(660,4985)
				(661,4988)(662,4994)(663,4998)(664,5000)};
			\addlegendentry{$26·26$}\addplot+coordinates{
				(628,1)(629,4)(630,8)(631,16)(632,31)(633,50)(634,100)(635,190)(636,311)
				(637,452)(638,634)(639,859)(640,1120)(641,1405)(642,1717)(643,2037)
				(644,2364)(645,2631)(646,2930)(647,3210)(648,3432)(649,3646)(650,3868)
				(651,4036)(652,4197)(653,4342)(654,4464)(655,4573)(656,4645)(657,4723)
				(658,4792)(659,4840)(660,4880)(661,4911)(662,4932)(663,4960)(664,4973)
				(665,4985)(666,4992)(667,4996)(668,4999)(669,5000)};
			\addlegendentry{$23·29$}\addplot+coordinates{
				(633,1)(634,2)(635,8)(636,17)(637,36)(638,58)(639,98)(640,154)(641,208)
				(642,290)(643,393)(644,556)(645,738)(646,951)(647,1167)(648,1439)
				(649,1756)(650,2080)(651,2422)(652,2751)(653,3102)(654,3401)(655,3713)
				(656,3988)(657,4236)(658,4441)(659,4605)(660,4748)(661,4860)(662,4933)
				(663,4979)(664,4996)(665,5000)};
			\addlegendentry{$23·29∪26·23$}\addplot+coordinates{
				(670,2)(671,5)(672,7)(673,21)(674,31)(675,53)(676,89)(677,137)(678,210)
				(679,289)(680,381)(681,513)(682,674)(683,857)(684,1063)(685,1278)
				(686,1487)(687,1715)(688,1949)(689,2186)(690,2441)(691,2673)(692,2902)
				(693,3101)(694,3298)(695,3487)(696,3652)(697,3806)(698,3942)(699,4066)
				(700,4189)(701,4282)(702,4368)(703,4455)(704,4533)(705,4610)(706,4666)
				(707,4726)(708,4767)(709,4808)(710,4835)(711,4874)(712,4895)(713,4922)
				(714,4936)(715,4954)(716,4964)(717,4973)(718,4976)(719,4985)(720,4987)
				(721,4994)(722,4995)(723,4996)(724,4998)(725,4999)(730,5000)};
			\addlegendentry{$26·29$}\addplot+coordinates{
				(678,1)(679,2)(680,6)(681,17)(682,48)(683,104)(684,177)(685,291)
				(686,447)(687,632)(688,852)(689,1129)(690,1446)(691,1714)(692,1987)
				(693,2268)(694,2547)(695,2810)(696,3051)(697,3306)(698,3517)(699,3711)
				(700,3881)(701,4053)(702,4185)(703,4315)(704,4417)(705,4519)(706,4577)
				(707,4642)(708,4688)(709,4740)(710,4782)(711,4818)(712,4848)(713,4871)
				(714,4888)(715,4906)(716,4922)(717,4928)(718,4934)(719,4942)(720,4958)
				(721,4968)(722,4977)(723,4982)(724,4984)(725,4987)(726,4992)(727,4993)
				(730,4995)(731,4996)(732,4998)(733,4999)(734,5000)};
			\addlegendentry{$\Cyc23·29$}\addplot+coordinates{
				(712,2)(713,4)(714,8)(715,15)(716,32)(717,55)(718,86)(719,130)(720,187)
				(721,268)(722,368)(723,486)(724,635)(725,791)(726,964)(727,1129)
				(728,1335)(729,1546)(730,1750)(731,1951)(732,2162)(733,2371)(734,2577)
				(735,2785)(736,2994)(737,3162)(738,3349)(739,3519)(740,3653)(741,3789)
				(742,3918)(743,4042)(744,4132)(745,4219)(746,4312)(747,4382)(748,4444)
				(749,4514)(750,4570)(751,4617)(752,4660)(753,4703)(754,4746)(755,4771)
				(756,4794)(757,4818)(758,4836)(759,4858)(760,4877)(761,4894)(762,4906)
				(763,4919)(764,4926)(765,4939)(766,4948)(767,4954)(768,4960)(769,4964)
				(770,4971)(771,4974)(772,4976)(773,4980)(774,4983)(775,4986)(776,4987)
				(777,4989)(778,4993)(780,4994)(781,4995)(782,4996)(783,4997)(784,4998)
				(789,4999)(791,5000)};
			\addlegendentry{EPC $737$}\addplot+coordinates{(737,0)(737,5000)};
			\addlegendentry{$\Cyc26·29$}\addplot+coordinates{
				(727,2)(728,5)(729,12)(730,23)(731,49)(732,96)(733,165)(734,285)
				(735,420)(736,607)(737,820)(738,1077)(739,1373)(740,1685)(741,2018)
				(742,2353)(743,2651)(744,2964)(745,3260)(746,3567)(747,3793)(748,3991)
				(749,4160)(750,4288)(751,4412)(752,4522)(753,4618)(754,4680)(755,4741)
				(756,4793)(757,4830)(758,4865)(759,4882)(760,4899)(761,4921)(762,4939)
				(763,4951)(764,4961)(765,4973)(766,4979)(767,4982)(768,4985)(769,4987)
				(770,4989)(771,4991)(773,4993)(774,4995)(777,4996)(781,4997)(782,4999)
				(788,5000)};
			\addlegendentry{$29·29$}\addplot+coordinates{
				(732,3)(733,9)(734,15)(735,31)(736,71)(737,138)(738,215)(739,352)
				(740,526)(741,750)(742,1003)(743,1299)(744,1641)(745,2011)(746,2348)
				(747,2687)(748,3035)(749,3352)(750,3623)(751,3849)(752,4053)(753,4228)
				(754,4393)(755,4505)(756,4588)(757,4673)(758,4734)(759,4791)(760,4834)
				(761,4874)(762,4902)(763,4919)(764,4930)(765,4950)(766,4956)(767,4965)
				(768,4969)(769,4974)(770,4983)(772,4985)(773,4987)(774,4993)(775,4994)
				(777,4995)(779,4996)(780,4997)(787,4998)(790,5000)};
		\end{axis}
	}
	\fi
	\caption{
		$9×9$.
		Strategy $26·29=754$ is the example in the abstract.
		For none of $5{,}000$ samples does this strategy fall behind EPC.
	}\label{fig:9x9}
\end{figure}

\begin{figure}%\labortrue
	\iflabor
	\tikz{
		\begin{axis}[
			legend pos=south east,
			scaled y ticks=base 10:-2,
		]
			\addlegendentry{$[32·42]$}\addplot+coordinates{
				(1279,1)(1281,2)(1282,3)(1283,4)(1284,13)(1285,22)(1286,30)(1287,48)
				(1288,73)(1289,99)(1290,126)(1291,167)(1292,221)(1293,275)(1294,336)
				(1295,402)(1296,496)(1297,584)(1298,663)(1299,746)(1300,832)(1301,928)
				(1302,1022)(1303,1112)(1304,1189)(1305,1277)(1306,1351)(1307,1429)
				(1308,1493)(1309,1553)(1310,1601)(1311,1658)(1312,1703)(1313,1739)
				(1314,1776)(1315,1806)(1316,1835)(1317,1862)(1318,1877)(1319,1899)
				(1320,1919)(1321,1930)(1322,1940)(1323,1955)(1324,1964)(1325,1973)
				(1326,1976)(1327,1983)(1328,1988)(1329,1991)(1330,1995)(1331,1996)
				(1332,1997)(1333,1998)(1337,1999)(1338,2000)};
			\addlegendentry{$[33·41]$}\addplot+coordinates{
				(1288,1)(1292,6)(1293,12)(1294,22)(1295,35)(1296,45)(1297,64)(1298,92)
				(1299,126)(1300,167)(1301,208)(1302,286)(1303,357)(1304,433)(1305,510)
				(1306,606)(1307,710)(1308,803)(1309,909)(1310,999)(1311,1079)(1312,1166)
				(1313,1247)(1314,1331)(1315,1401)(1316,1459)(1317,1521)(1318,1574)
				(1319,1632)(1320,1676)(1321,1709)(1322,1743)(1323,1775)(1324,1811)
				(1325,1844)(1326,1873)(1327,1893)(1328,1908)(1329,1923)(1330,1939)
				(1331,1949)(1332,1959)(1333,1970)(1334,1977)(1335,1984)(1336,1988)
				(1337,1989)(1339,1994)(1340,1998)(1343,1999)(1344,2000)};
			\addlegendentry{$26·53$}\addplot+coordinates{
				(1312,1)(1313,3)(1314,4)(1315,7)(1316,8)(1317,14)(1318,25)(1319,40)
				(1320,59)(1321,83)(1322,103)(1323,130)(1324,169)(1325,224)(1326,282)
				(1327,345)(1328,418)(1329,508)(1330,585)(1331,681)(1332,749)(1333,853)
				(1334,941)(1335,1018)(1336,1107)(1337,1170)(1338,1243)(1339,1315)
				(1340,1378)(1341,1447)(1342,1506)(1343,1565)(1344,1613)(1345,1665)
				(1346,1705)(1347,1752)(1348,1782)(1349,1800)(1350,1829)(1351,1852)
				(1352,1878)(1353,1901)(1354,1916)(1355,1928)(1356,1947)(1357,1961)
				(1358,1966)(1359,1975)(1360,1982)(1361,1989)(1362,1990)(1363,1993)
				(1364,1997)(1365,1998)(1366,1999)(1368,2000)};
			\addlegendentry{$26·57$}\addplot+coordinates{
				(1384,2)(1385,4)(1386,8)(1387,12)(1388,24)(1389,44)(1390,62)(1391,97)
				(1392,156)(1393,190)(1394,254)(1395,320)(1396,390)(1397,470)(1398,576)
				(1399,671)(1400,753)(1401,852)(1402,957)(1403,1056)(1404,1150)
				(1405,1230)(1406,1307)(1407,1386)(1408,1450)(1409,1512)(1410,1576)
				(1411,1623)(1412,1665)(1413,1710)(1414,1740)(1415,1762)(1416,1790)
				(1417,1816)(1418,1838)(1419,1857)(1420,1882)(1421,1892)(1422,1905)
				(1423,1920)(1424,1933)(1425,1941)(1426,1948)(1427,1953)(1428,1959)
				(1429,1964)(1430,1967)(1431,1975)(1432,1978)(1433,1982)(1434,1984)
				(1435,1987)(1436,1988)(1437,1989)(1438,1990)(1442,1992)(1443,1993)
				(1444,1994)(1447,1995)(1448,1999)(1451,2000)};
			\addlegendentry{$29·53$}\addplot+coordinates{
				(1422,2)(1423,3)(1424,5)(1425,7)(1426,9)(1427,14)(1428,21)(1429,27)
				(1430,36)(1431,52)(1432,77)(1433,111)(1434,143)(1435,189)(1436,233)
				(1437,282)(1438,343)(1439,407)(1440,480)(1441,563)(1442,649)(1443,740)
				(1444,816)(1445,910)(1446,989)(1447,1072)(1448,1159)(1449,1241)
				(1450,1300)(1451,1374)(1452,1431)(1453,1490)(1454,1538)(1455,1588)
				(1456,1620)(1457,1659)(1458,1696)(1459,1724)(1460,1742)(1461,1761)
				(1462,1787)(1463,1808)(1464,1821)(1465,1844)(1466,1858)(1467,1869)
				(1468,1883)(1469,1899)(1470,1914)(1471,1922)(1472,1928)(1473,1932)
				(1474,1939)(1475,1944)(1476,1952)(1477,1958)(1478,1968)(1479,1973)
				(1480,1977)(1481,1980)(1482,1982)(1483,1984)(1484,1986)(1485,1988)
				(1486,1989)(1488,1992)(1489,1994)(1490,1995)(1492,1996)(1493,1997)
				(1496,1998)(1497,1999)(1498,2000)};
			\addlegendentry{$29·57$}\addplot+coordinates{
				(1499,1)(1501,7)(1502,8)(1503,14)(1504,30)(1505,41)(1506,59)(1507,92)
				(1508,130)(1509,176)(1510,249)(1511,324)(1512,397)(1513,478)(1514,559)
				(1515,639)(1516,724)(1517,830)(1518,944)(1519,1040)(1520,1145)
				(1521,1242)(1522,1323)(1523,1401)(1524,1485)(1525,1563)(1526,1619)
				(1527,1677)(1528,1727)(1529,1771)(1530,1817)(1531,1851)(1532,1877)
				(1533,1896)(1534,1917)(1535,1934)(1536,1948)(1537,1960)(1538,1967)
				(1539,1973)(1540,1980)(1541,1984)(1542,1986)(1543,1987)(1544,1990)
				(1545,1994)(1546,1995)(1548,1996)(1550,1997)(1551,1998)(1558,1999)
				(1559,2000)};
			\addlegendentry{$7·9·26$}\addplot+coordinates{
				(1494,1)(1497,2)(1498,4)(1499,6)(1500,7)(1502,11)(1504,12)(1505,14)
				(1506,17)(1507,20)(1508,25)(1509,31)(1510,36)(1511,45)(1512,55)(1513,68)
				(1514,79)(1515,89)(1516,92)(1517,102)(1518,114)(1519,125)(1520,136)
				(1521,148)(1522,166)(1523,187)(1524,206)(1525,221)(1526,240)(1527,262)
				(1528,279)(1529,304)(1530,328)(1531,346)(1532,374)(1533,400)(1534,430)
				(1535,452)(1536,474)(1537,506)(1538,537)(1539,576)(1540,611)(1541,638)
				(1542,667)(1543,703)(1544,730)(1545,762)(1546,794)(1547,819)(1548,853)
				(1549,892)(1550,929)(1551,956)(1552,989)(1553,1022)(1554,1054)
				(1555,1083)(1556,1117)(1557,1153)(1558,1200)(1559,1226)(1560,1260)
				(1561,1298)(1562,1329)(1563,1360)(1564,1388)(1565,1419)(1566,1440)
				(1567,1472)(1568,1493)(1569,1520)(1570,1554)(1571,1570)(1572,1592)
				(1573,1615)(1574,1631)(1575,1649)(1576,1679)(1577,1694)(1578,1713)
				(1579,1730)(1580,1745)(1581,1767)(1582,1786)(1583,1797)(1584,1810)
				(1585,1818)(1586,1828)(1587,1839)(1588,1860)(1589,1869)(1590,1881)
				(1591,1891)(1592,1902)(1593,1908)(1594,1921)(1595,1930)(1596,1937)
				(1597,1943)(1598,1949)(1599,1955)(1600,1959)(1601,1963)(1602,1974)
				(1603,1977)(1604,1980)(1605,1983)(1606,1985)(1607,1989)(1608,1991)
				(1609,1995)(1612,1997)(1613,1998)(1614,1999)(1617,2000)};
			\addlegendentry{$7·9·29$}\addplot+coordinates{
				(1592,2)(1596,3)(1597,6)(1599,7)(1602,8)(1603,12)(1604,14)(1605,18)
				(1606,21)(1607,24)(1608,28)(1609,31)(1610,34)(1611,38)(1612,43)(1613,50)
				(1614,57)(1615,63)(1616,73)(1617,83)(1618,90)(1619,105)(1620,114)
				(1621,123)(1622,136)(1623,145)(1624,161)(1625,178)(1626,195)(1627,213)
				(1628,234)(1629,253)(1630,278)(1631,297)(1632,321)(1633,341)(1634,360)
				(1635,372)(1636,395)(1637,416)(1638,436)(1639,459)(1640,480)(1641,508)
				(1642,533)(1643,558)(1644,590)(1645,619)(1646,642)(1647,665)(1648,695)
				(1649,725)(1650,752)(1651,783)(1652,813)(1653,832)(1654,866)(1655,891)
				(1656,917)(1657,949)(1658,983)(1659,1021)(1660,1056)(1661,1077)
				(1662,1104)(1663,1121)(1664,1144)(1665,1176)(1666,1208)(1667,1230)
				(1668,1255)(1669,1288)(1670,1316)(1671,1336)(1672,1361)(1673,1391)
				(1674,1417)(1675,1436)(1676,1452)(1677,1471)(1678,1486)(1679,1507)
				(1680,1529)(1681,1546)(1682,1559)(1683,1575)(1684,1597)(1685,1617)
				(1686,1642)(1687,1655)(1688,1666)(1689,1685)(1690,1702)(1691,1727)
				(1692,1740)(1693,1750)(1694,1754)(1695,1762)(1696,1772)(1697,1778)
				(1698,1789)(1699,1802)(1700,1808)(1701,1821)(1702,1827)(1703,1838)
				(1704,1845)(1705,1848)(1706,1855)(1707,1858)(1708,1865)(1709,1871)
				(1710,1878)(1711,1885)(1712,1890)(1713,1894)(1714,1898)(1715,1908)
				(1716,1909)(1717,1915)(1718,1918)(1719,1924)(1720,1929)(1722,1933)
				(1724,1938)(1725,1941)(1726,1945)(1727,1950)(1728,1951)(1729,1953)
				(1730,1956)(1731,1961)(1732,1963)(1733,1966)(1734,1970)(1735,1971)
				(1736,1973)(1737,1976)(1738,1978)(1739,1980)(1740,1981)(1741,1983)
				(1743,1984)(1744,1986)(1746,1988)(1749,1989)(1752,1990)(1753,1991)
				(1756,1992)(1757,1993)(1758,1994)(1759,1995)(1762,1997)(1764,1998)
				(1766,1999)(1776,2000)};
			\addlegendentry{$9·9·26$}\addplot+coordinates{
				(1664,1)(1672,2)(1673,3)(1674,4)(1675,9)(1676,11)(1677,13)(1678,15)
				(1679,20)(1680,27)(1681,33)(1682,42)(1683,48)(1684,55)(1685,60)(1686,69)
				(1687,77)(1688,92)(1689,105)(1690,114)(1691,130)(1692,139)(1693,159)
				(1694,174)(1695,193)(1696,213)(1697,236)(1698,257)(1699,286)(1700,306)
				(1701,336)(1702,362)(1703,383)(1704,407)(1705,437)(1706,459)(1707,482)
				(1708,520)(1709,540)(1710,569)(1711,601)(1712,632)(1713,663)(1714,692)
				(1715,722)(1716,744)(1717,778)(1718,802)(1719,828)(1720,856)(1721,879)
				(1722,905)(1723,927)(1724,949)(1725,973)(1726,994)(1727,1017)(1728,1030)
				(1729,1051)(1730,1065)(1731,1086)(1732,1115)(1733,1133)(1734,1154)
				(1735,1176)(1736,1187)(1737,1208)(1738,1225)(1739,1247)(1740,1264)
				(1741,1281)(1742,1296)(1743,1313)(1744,1329)(1745,1345)(1746,1360)
				(1747,1372)(1748,1388)(1749,1404)(1750,1417)(1751,1429)(1752,1440)
				(1753,1455)(1754,1467)(1755,1479)(1756,1490)(1757,1501)(1758,1518)
				(1759,1528)(1760,1538)(1761,1547)(1762,1564)(1763,1568)(1764,1582)
				(1765,1588)(1766,1598)(1767,1609)(1768,1615)(1769,1621)(1770,1633)
				(1771,1640)(1772,1648)(1773,1659)(1774,1666)(1775,1672)(1776,1676)
				(1777,1681)(1778,1686)(1779,1698)(1780,1709)(1781,1715)(1782,1718)
				(1783,1726)(1784,1730)(1785,1743)(1786,1752)(1787,1759)(1788,1765)
				(1789,1769)(1790,1774)(1791,1781)(1793,1786)(1794,1792)(1795,1800)
				(1796,1806)(1797,1814)(1798,1818)(1799,1825)(1800,1828)(1801,1834)
				(1802,1839)(1803,1845)(1804,1849)(1805,1853)(1806,1864)(1807,1869)
				(1808,1873)(1809,1877)(1810,1881)(1811,1885)(1812,1891)(1813,1893)
				(1814,1898)(1815,1899)(1816,1902)(1817,1903)(1818,1907)(1819,1909)
				(1820,1912)(1821,1916)(1822,1921)(1823,1923)(1824,1925)(1825,1927)
				(1826,1930)(1827,1932)(1828,1933)(1829,1935)(1830,1937)(1831,1939)
				(1833,1941)(1834,1942)(1835,1946)(1836,1947)(1837,1950)(1838,1952)
				(1839,1953)(1840,1954)(1841,1956)(1842,1961)(1843,1962)(1846,1963)
				(1847,1964)(1851,1966)(1853,1968)(1854,1970)(1855,1971)(1858,1972)
				(1861,1974)(1862,1976)(1863,1977)(1866,1980)(1867,1982)(1868,1984)
				(1869,1985)(1872,1986)(1878,1988)(1885,1989)(1886,1990)(1887,1991)
				(1888,1992)(1891,1993)(1896,1994)(1900,1995)(1903,1996)(1905,1997)
				(1908,1998)(1911,1999)(1940,2000)};
			\addlegendentry{EPC $1739$}\addplot+coordinates{(1739,0)(1739,2000)};
			\addlegendentry{$9·(13·18)$}\addplot+coordinates{
				(1667,1)(1681,2)(1683,5)(1684,9)(1685,12)(1686,16)(1687,17)(1688,23)
				(1689,30)(1690,33)(1691,45)(1692,54)(1693,67)(1694,82)(1695,98)
				(1696,112)(1697,128)(1698,145)(1699,168)(1700,191)(1701,220)(1702,244)
				(1703,275)(1704,300)(1705,323)(1706,351)(1707,380)(1708,412)(1709,444)
				(1710,480)(1711,510)(1712,542)(1713,579)(1714,611)(1715,639)(1716,663)
				(1717,695)(1718,713)(1719,740)(1720,767)(1721,799)(1722,837)(1723,859)
				(1724,885)(1725,905)(1726,925)(1727,948)(1728,978)(1729,1001)(1730,1028)
				(1731,1051)(1732,1078)(1733,1095)(1734,1114)(1735,1134)(1736,1148)
				(1737,1163)(1738,1181)(1739,1204)(1740,1220)(1741,1241)(1742,1265)
				(1743,1285)(1744,1300)(1745,1318)(1746,1337)(1747,1344)(1748,1355)
				(1749,1362)(1750,1377)(1751,1393)(1752,1409)(1753,1426)(1754,1443)
				(1755,1457)(1756,1463)(1757,1475)(1758,1483)(1759,1495)(1760,1508)
				(1761,1520)(1762,1527)(1763,1538)(1764,1548)(1765,1556)(1766,1564)
				(1767,1573)(1768,1582)(1769,1592)(1770,1608)(1771,1621)(1772,1631)
				(1773,1639)(1774,1643)(1775,1653)(1776,1668)(1777,1674)(1778,1684)
				(1779,1691)(1780,1698)(1781,1707)(1782,1710)(1783,1715)(1784,1723)
				(1785,1732)(1786,1740)(1787,1745)(1788,1749)(1789,1756)(1790,1764)
				(1791,1770)(1792,1777)(1793,1784)(1794,1788)(1795,1790)(1796,1793)
				(1797,1797)(1798,1804)(1799,1809)(1800,1816)(1801,1822)(1802,1828)
				(1803,1832)(1804,1834)(1805,1837)(1806,1841)(1807,1849)(1809,1852)
				(1810,1855)(1811,1859)(1812,1864)(1813,1865)(1814,1866)(1815,1868)
				(1816,1871)(1817,1876)(1818,1877)(1819,1883)(1820,1886)(1821,1892)
				(1822,1896)(1823,1898)(1824,1901)(1825,1903)(1826,1907)(1827,1908)
				(1828,1909)(1829,1910)(1830,1915)(1831,1916)(1832,1921)(1833,1922)
				(1834,1923)(1835,1926)(1836,1928)(1837,1931)(1838,1934)(1839,1935)
				(1841,1940)(1842,1943)(1843,1945)(1845,1946)(1846,1947)(1847,1948)
				(1849,1950)(1850,1952)(1852,1953)(1853,1956)(1854,1957)(1855,1960)
				(1856,1965)(1857,1967)(1858,1968)(1859,1969)(1860,1970)(1861,1973)
				(1862,1974)(1863,1976)(1868,1977)(1874,1978)(1878,1979)(1879,1980)
				(1880,1981)(1881,1982)(1884,1985)(1885,1986)(1886,1987)(1887,1988)
				(1889,1989)(1894,1990)(1898,1992)(1899,1993)(1907,1994)(1909,1995)
				(1910,1996)(1920,1997)(1946,1998)(1950,1999)(1951,2000)};
			\addlegendentry{$9·(14·17)$}\addplot+coordinates{
				(1705,1)(1709,2)(1710,3)(1713,6)(1714,7)(1715,8)(1716,13)(1717,16)
				(1718,19)(1719,23)(1720,29)(1721,38)(1722,45)(1723,52)(1724,70)(1725,83)
				(1726,94)(1727,115)(1728,136)(1729,156)(1730,178)(1731,214)(1732,232)
				(1733,258)(1734,293)(1735,319)(1736,357)(1737,379)(1738,408)(1739,440)
				(1740,462)(1741,504)(1742,535)(1743,570)(1744,603)(1745,634)(1746,665)
				(1747,698)(1748,732)(1749,755)(1750,781)(1751,811)(1752,839)(1753,870)
				(1754,898)(1755,933)(1756,955)(1757,977)(1758,1006)(1759,1024)
				(1760,1046)(1761,1075)(1762,1094)(1763,1111)(1764,1134)(1765,1164)
				(1766,1185)(1767,1204)(1768,1220)(1769,1235)(1770,1249)(1771,1262)
				(1772,1273)(1773,1285)(1774,1299)(1775,1312)(1776,1333)(1777,1346)
				(1778,1361)(1779,1372)(1780,1383)(1781,1404)(1782,1419)(1783,1433)
				(1784,1444)(1785,1462)(1786,1473)(1787,1483)(1788,1498)(1789,1507)
				(1790,1519)(1791,1530)(1792,1540)(1793,1547)(1794,1561)(1795,1574)
				(1796,1584)(1797,1597)(1798,1610)(1799,1617)(1800,1628)(1801,1635)
				(1802,1643)(1803,1649)(1804,1658)(1805,1668)(1806,1672)(1807,1680)
				(1808,1692)(1809,1702)(1810,1707)(1811,1719)(1812,1727)(1813,1737)
				(1814,1747)(1815,1748)(1816,1756)(1817,1763)(1818,1770)(1819,1777)
				(1820,1782)(1821,1787)(1822,1795)(1823,1803)(1824,1807)(1825,1814)
				(1826,1818)(1827,1823)(1828,1826)(1829,1829)(1830,1835)(1831,1842)
				(1832,1848)(1833,1852)(1834,1859)(1835,1862)(1836,1866)(1837,1868)
				(1838,1870)(1839,1872)(1840,1874)(1841,1875)(1842,1877)(1843,1879)
				(1844,1882)(1845,1886)(1846,1890)(1847,1896)(1848,1900)(1849,1904)
				(1851,1909)(1852,1911)(1853,1914)(1854,1915)(1855,1918)(1856,1923)
				(1857,1925)(1858,1928)(1859,1930)(1860,1932)(1861,1934)(1862,1935)
				(1863,1936)(1865,1937)(1866,1940)(1867,1941)(1868,1943)(1869,1945)
				(1872,1947)(1874,1951)(1875,1953)(1876,1954)(1877,1956)(1878,1958)
				(1879,1959)(1883,1960)(1884,1964)(1885,1965)(1886,1966)(1887,1967)
				(1888,1968)(1889,1969)(1890,1970)(1892,1971)(1893,1972)(1895,1973)
				(1901,1974)(1903,1975)(1906,1977)(1907,1978)(1909,1979)(1912,1980)
				(1913,1981)(1916,1983)(1917,1984)(1918,1986)(1923,1987)(1924,1988)
				(1925,1989)(1928,1990)(1929,1991)(1936,1992)(1940,1993)(1943,1994)
				(1948,1995)(1949,1996)(1955,1997)(1972,1998)(1980,1999)(1981,2000)};
		\end{axis}
	}
	\fi
	\caption{
		$12×12$.
		Between the brackets are the rectangular syntheses
		$⟨3,3,4;32⟩⊗⟨4,4,3;42⟩$ and $⟨3,3,4;33⟩⊗⟨4,4,3;41⟩$.
		Between the parentheses are components from \cref{fig:6x6}.
	}\label{fig:12x12}
\end{figure}
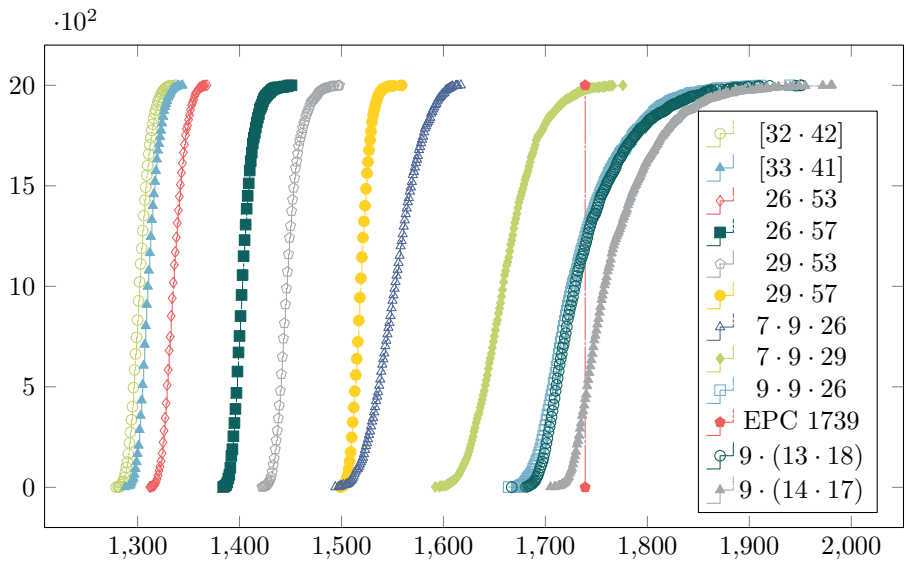

\def\assigncellcontent#1#2#3{
	\PMT\rowindex{#1}	\PMT\columnindex{#2}
	\pgfplotstableset{every row \rowindex\space column \columnindex/.estyle={
		assign cell content/.style={@cell content=#3}
	}}
}
\pgfplotstableread{
	dimen		naive	rank	n-worker	expon	thres	EPC	PDC	EPC2
	⟨2,2,2⟩		8		7		+2=9		3.170	8		9	12	13	
	⟨2,2,2⟩		8		7		+4=11		3.459	9		9	12	13	
	⟨2,2,2⟩		8		7		+6=13		3.700	10		9	12	13	
	⟨2,2,2⟩		8		7		+8=15		3.907	11		9	12	13	
	⟨2,2,2⟩		8		7		+10=17		4.087	12		9	12	13	
	⟨2,2,2⟩		8		7		+12=19		4.248	13		9	12	13	
	⟨2,2,3⟩		12		11		+2=13		3.126	12		13	18	21	
	⟨2,3,2⟩		12		11		+3=14		3.126	13		14	20	21	
	⟨3,2,3⟩		18		15		+2=17		2.980	16		19	27	29	
	⟨2,3,3⟩		18		15		+3=18		2.980	17		20	30	29	
	⟨3,3,3⟩		27		23		+3=26		2.966	25		29	45	45	
	⟨3,3,3⟩		27		23		+6=29		3.065	27		29	45	45	
	⟨3,3,3⟩		27		23		+9=32		3.155	29		29	45	45	
	⟨3,3,3⟩		27		23		+12=35		3.236	31		29	45	45	
	⟨3,3,3⟩		27		23		+15=38		3.311	33		29	45	45	
	⟨3,3,4⟩		36		29		+3=32		2.910	31		38	60	57	
	⟨3,4,3⟩		36		29		+4=33		2.910	32		39	63	57	
	⟨4,3,4⟩		48		38		+3=41		2.890	40		50	80	75	
	⟨3,4,4⟩		48		38		+4=42		2.890	41		51	84	75	
	⟨4,4,4⟩		64		49		+4=53		2.864	52		67	112	97	
	⟨4,4,4⟩		64		49		+8=57		2.916	55		67	112	97	
	⟨4,4,4⟩		64		49		+12=61		2.965	58		67	112	97	
	⟨4,4,4⟩		64		49		+14=63		2.989	61		67	112	97	
}\tablePrime
\begin{table}
	\caption{
		Summary of prime Pluto codes.
		The last row, with $63$ workers,
		is the Charon construction (see \cref{app:Charon}).
	}\label{tab:prime}
	\def\arraystretch{1.2}
	\pgfplotstablegetrowsof\tablePrime\PMT\lastrow{\pgfplotsretval-1}
	\pgfplotsforeachungrouped\r in{0,...,\lastrow}{
		\pgfplotstablegetelem\r{dimen}\of\tablePrime
		\ifx\lastdimen\pgfplotsretval
			\assigncellcontent{\r-1}0{}
			\assigncellcontent{\r-1}6{}
			\assigncellcontent{\r-1}7{}
		\else
			\let\lastdimen\pgfplotsretval
			\def\dimenrepeat{}
		\fi
		\edef\dimenrepeat{\the\numexpr\dimenrepeat-1}
		\assigncellcontent\r0{\NE\multirow{\dimenrepeat}*{$##1$}}
		\assigncellcontent\r6{\NE\multirow{\dimenrepeat}*{$##1$}}
		\assigncellcontent\r7{\NE\multirow{\dimenrepeat}*{$##1$}}
	}
	\pgfplotsforeachungrouped\r in{0,...,\lastrow}{
		\pgfplotstablegetelem\r{expon}\of\tablePrime
		\ifx\lastexpon\pgfplotsretval
			\assigncellcontent{\r-1}4{}
		\else
			\let\lastexpon\pgfplotsretval
			\def\exponrepeat{}
		\fi
		\edef\exponrepeat{\the\numexpr\exponrepeat-1}
		\assigncellcontent\r4{\NE\multirow{\exponrepeat}*{$##1$}}
	}
	\def\shadesign{1}
	\pgfplotsforeachungrouped\r in{0,...,\lastrow}{
		\pgfplotstablegetelem\r{naive}\of\tablePrime
		\ifx\lastnaive\pgfplotsretval
			\assigncellcontent{\r-1}1{}
			\assigncellcontent{\r-1}2{}
			\assigncellcontent{\r-1}8{}
		\else
			\let\lastnaive\pgfplotsretval
			\def\naiverepeat{}\edef\shadesign{\the\numexpr-\shadesign}
		\fi
		\edef\naiverepeat{\the\numexpr\naiverepeat-1}
		\assigncellcontent\r1{\NE\multirow{\naiverepeat}*{$##1$}}
		\assigncellcontent\r2{\NE\multirow{\naiverepeat}*{$##1$}}
		\assigncellcontent\r8{\NE\multirow{\naiverepeat}*{$##1$}}
		\ifnum\shadesign=1
			\pgfplotstableset{every row no \r/.style={before row=\rowcolor{427}}}
		\fi
	}
	\centering\pgfplotstabletypeset[
		columns/naive/.style={column name=naïve},
		columns/n-worker/.style={
			column type/.initial={r<{\pgfplotstableresetcolortbloverhangright}
								@{}l<{\pgfplotstableresetcolortbloverhangleft}},
			column name={\multicolumn2c{\#worker}},
			string replace*={=}{\;$&$=},
		},
		every head row/.style={before row=\toprule,after row=\midrule},
		every last row/.style={after row=\bottomrule},
		assign cell content/.style={@cell content={${}#1$}},
	]\tablePrime
\end{table}

\pgfplotstableread{
	dimension				naive	rank	n-worker		EPC		EPC2	
	⟨2,2,2⟩					8		7		7+2=9			9		13		
	⟨3,3,3⟩					27		23		23+3=26			29		45		
	⟨4,4,4⟩					64		49		49+4=57			67		97		
	⟨2,2,2⟩^{⊗2}			64		49		9^2=81			67		97		
	⟨2,2,2⟩⊗⟨3,3,3⟩			216		161		9·26=234		221		321		
	⟨2,2,3⟩⊗⟨3,3,2⟩			216		161		13·18=234		221		321		
	⟨2,3,2⟩⊗⟨3,2,3⟩			216		161		14·17=238		221		321		
	⟨2,2,2⟩^{⊗3}			512		343		9^3=729			519		685		
	⟨2,2,2⟩⊗⟨4,4,4⟩			512		343		9·53=477		519		685		
	\Cyc⟨2,2,2⟩⊗⟨4,4,4⟩		512		343		\Cyc7·53=427	519		685		
	⟨3,3,3⟩^{⊗2}			729		522		26·26=676		737		1,043	
	⟨3,3,3⟩⊗⟨4,4,4⟩			1,728	1,125	26·53=1,378		1,739	2,249	
	⟨3,3,4⟩⊗⟨4,4,3⟩			1,728	1,125	32·42=1,344		1,739	2,249	
	⟨3,4,4⟩⊗⟨4,3,4⟩			1,728	1,125	33·41=1,353		1,739	2,249	
	⟨2,2,2⟩^{⊗2}⊗⟨3,3,3⟩	1,728	1,125	9^2·26=2,106	1,739	2,249	
	⟨2,2,2⟩⊗⟨2,2,3⟩⊗⟨3,3,2⟩	1,728	1,125	9·13·18=2,106	1,739	2,249	
	⟨2,2,2⟩⊗⟨2,3,2⟩⊗⟨3,2,3⟩	1,728	1,125	9·14·17=2,142	1,739	2,249	
	⟨2,2,3⟩⊗⟨2,3,2⟩⊗⟨3,2,2⟩	1,728	1,125	13·14·13=2,366	1,739	2,249	
	⟨4,4,4⟩^{⊗2}			4,096	2,401	53·53=2,809		4,111	4,801	
	⟨2,2,2⟩^{⊗2}⊗⟨4,4,4⟩	4,096	2,401	9^2·53=4,293	4,111	4,801	
	⟨2,2,2⟩^{⊗4}			4,096	2,401	9^4=6,561		4,111	4,801	
	⟨2,2,2⟩⊗⟨3,3,3⟩^2		5,832	3,200	9·26^2=6,084	5,849	6,399	
	⟨2,2,3⟩⊗⟨3,3,2⟩⊗⟨3,3,3⟩	5,832	3,200	13·18·26=6,084	5,849	6,399	
	⟨2,3,2⟩⊗⟨3,2,3⟩⊗⟨3,3,3⟩	5,832	3,200	14·17·26=6,188	5,849	6,399	
	⟨2,3,3⟩⊗⟨3,2,3⟩⊗⟨3,3,2⟩	5,832	3,200	18·17·18=5,508	5,849	6,399	
%	xxxxxxxx				xxxx	xxxx	xxxxxxxx=xxxx	xxxx	xxxx	
}\tableCompose
\begin{table}
	\caption{
		A partial list of composite Pluto codes;
		one representative for each dimension type.
	}\label{tab:compose}
	\def\arraystretch{1.2}
	\def\shadesign{1}
	\pgfplotstablegetrowsof\tableCompose\PMT\lastrow{\pgfplotsretval-1}
	\pgfplotsforeachungrouped\r in{0,...,\lastrow}{
		\pgfplotstablegetelem\r{naive}\of\tableCompose
		\ifx\lastnaive\pgfplotsretval
			\assigncellcontent{\r-1}1{}
			\assigncellcontent{\r-1}2{}
			\assigncellcontent{\r-1}4{}
			\assigncellcontent{\r-1}5{}
		\else
			\let\lastnaive\pgfplotsretval
			\def\naiverepeat{}\edef\shadesign{\the\numexpr-\shadesign}
		\fi
		\edef\naiverepeat{\the\numexpr\naiverepeat-1}
		\assigncellcontent\r1{\NE\multirow{\naiverepeat}*{$##1$}}
		\assigncellcontent\r2{\NE\multirow{\naiverepeat}*{$##1$}}
		\assigncellcontent\r4{\NE\multirow{\naiverepeat}*{$##1$}}
		\assigncellcontent\r5{\NE\multirow{\naiverepeat}*{$##1$}}
		\ifnum\shadesign=1
			\pgfplotstableset{every row no \r/.style={before row=\rowcolor{427}}}
		\fi
	}
	\centering\pgfplotstabletypeset[
		assign cell content/.style={@cell content={$#1$}},
		columns/naive/.style={column name=naïve,string replace*={,}{{,}},},
		columns/rank/.style={string replace*={,}{{,}}},
		columns/n-worker/.style={
			column type/.initial={
				r<{\pgfplotstableresetcolortbloverhangright}%
				@{}%
				l<{\pgfplotstableresetcolortbloverhangleft}%
			},
			column name={\multicolumn2c{\#worker}},
			string replace*={=}{=$&$\;},
			string replace*={,}{{,}},
		},
		columns/thres/.style={string replace*={,}{{,}}},
		columns/EPC/.style={string replace*={,}{{,}}},
		columns/EPC2/.style={string replace*={,}{{,}}},
		every head row/.style={before row=\toprule,after row=\midrule},
		every last row/.style={after row=\bottomrule},
	]\tableCompose
\end{table}

\begin{figure}
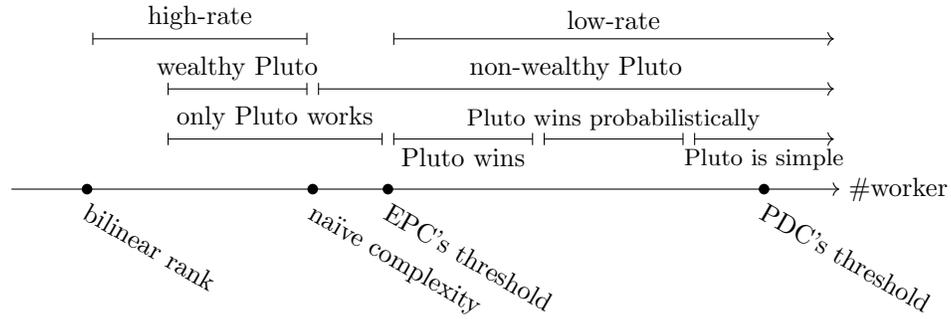

	\tikz{
		\draw[->](0,0)--(11,0)node[right]{\#worker};
		\fill[nodes={below right,rotate=-30}]
			(1,0)circle(2pt)node{bilinear rank}
			(4,0)circle(2pt)node{naïve complexity}
			(5,0)circle(2pt)node{EPC's threshold}
			(10,0)circle(2pt)node{PDC's threshold}
		;
		\tikzset{yscale=2/3,,shorten <=2pt,shorten >=2pt,nodes=above}
		\draw[|-|](1,3)--node{high-rate}(4,3);
		\draw[|->](5,3)--node{low-rate}(11,3);
		\draw[|-|](2,2)--node{wealthy Pluto}(4,2);
		\draw[|->](4,2)--node{non-wealthy Pluto}(11,2);
		\draw[|-|](2,1)--node{only Pluto works}(5,1);
		\draw[|-|](5,1)--node[below]{Pluto wins}(7,1);
		\draw[|-|](7,1)--node{\smaller Pluto wins probabilistically}(9,1);
		\draw[|->](9,1)--node[below]{\smaller Pluto is simple}(11,1);
	}
	\caption{
		The low- and high-rate regions.
	}\label{fig:worker}
\end{figure}

\appendix

\section{EPC/PDC-based Checksums}\label{app:RS}

	We remark that the checksums can be made compatible with EPC/PDC.
	Recall the first checksum
	\[*\bma{1&2}C\bma{-1\\1}=\(\bma{1&2}A\)⋆\left(B\bma{-1\\1}\right).\]*
	We modify it slightly:
	Choose a generic scalar $ζ≠0$, then
	\[2\bma{1&ζ^2}C\bma{1\\ζ^4}
		=\bma{1&ζ^2}A\bma{2\\&2}B\bma{1\\ζ^4}.\eqlabel{equ:EPC}\]
	The \LHS/ is a linear combination of $\S1$--$\S7$.
	The \RHS/ contains a diagonal matrix, which is the sum
	\[*\bma{2\\&2}=\bma{1\\ζ}\bma{1&ζ^{-1}}+\bma{1\\-ζ}\bma{1&÷1{-ζ}}.\]*
	Thus, the \RHS/ of \cref{equ:EPC} equals
	\[*
		\left(\bma{1&ζ^2}A\bma{1\\ζ}\right)⋆\left(\bma{1&÷1{ζ}}B\bma{1\\ζ^4}\right)+
		\left(\bma{1&ζ^2}A\bma{1\\-ζ}\right)⋆\left(\bma{1&÷1{-ζ}}B\bma{1\\ζ^4}\right).
	\]*
	Call the first star $\S8$, the second star $\S9$.
	We get that a linear combination of $\S1$--$\S7$ equals $(\S8+\S9)$.
	Reselecting $ζ$ if necessary, we can make $\S1$--$\S9$ MDS.
	
	For general dimension $⟨ℓ,m,n⟩$, establish a parity-check equation:
	\[*m\bsm{ζ^m&ζ^{2m}&⋯&ζ^{ℓm}}C\bsm{ζ^{ℓm}\\ζ^{2ℓm}\\⋮\\ζ^{ℓmn}}
		=\bsm{ζ^m&ζ^{2m}&⋯&ζ^{ℓm}}A
			\bsm{m\\&⋱\\&&m}B\bsm{ζ^{ℓm}\\ζ^{2ℓm}\\⋮\\ζ^{ℓmn}}\]*
	Use the $m$th root of unity $μ^m=1$ to breakdown the diagonal matrix:
	\[*\bsm{m\\&⋱\\&&m}
		=∑_{k=1}^m\bsm{ζμ^k\\ζ^2μ^{2k}\\⋮\\ζ^mμ^{mk}}
			\bsm{÷1{ζμ^k}&÷1{ζ^2μ^{2k}}&⋯&÷1{ζ^mμ^{mk}}}\]*
	This is nothing but discrete Fourier, $∑_{k=1}^mμ^{k(i-j)}=mδ_{ij}$.
	Lastly, perform the associativity trick.

\section{Cyclic Symmetry of Strassen}\label{app:triangle}

	This appendix archives the group action we used to generate
	Strassen checksums alongside the preparation of this paper.
	For the full symmetry group, plus an index-$2$ extension,
	refer to \cite{Burichenko14}.
	
	Consider the change of basis
	\[*A⟼RAR^{-1},\qquad R≔\bma{-1&1\\-1&0}.\]*
	Then it fixes the trace $\A11+\A22$ and creates orbits
	\begin{gather*}
		\A11⟼\A22-\A12⟼\A22+\A21⟼†cycle†,	\\
		\A11+\A22⟼†stabilized†,	\\
		\A22⟼\A11+\A12⟼\A11-\A21⟼†cycle†.	
	\end{gather*}
	
	Let the same operation apply to $A,B,C$ at once, then \crefrange{for:C11}{for:C22}
	expand to the following ``orbital closure'' that respects symmetry:
	\[*\def\cdot{\color{427}0}
		\bma{
			\C22+\C21	\\	\C11	\\	\C22-\C12	\\[2ex]
			\C11+\C22								\\[2ex]
			\C11-\C21	\\	\C22	\\	\C11+\C12	
		}=\arraycolsep4pt\bma{
			-1 &  · &  · && 1 &&  · & -1 & -1 \\
			 · & -1 &  · && 1 && -1 &  · & -1 \\
			 · &  · & -1 && 1 && -1 & -1 &  · \\[2ex]
			-1 & -1 & -1 && 2 && -1 & -1 & -1 \\[2ex]
			 · & -1 & -1 && 1 && -1 &  · &  · \\
			-1 &  · & -1 && 1 &&  · & -1 &  · \\
			-1 & -1 &  · && 1 &&  · & · & -1  
		}\arraycolsep0pt
		\left[\begin{array}{rl}
				\A11		&	(\B22-\B12)		\\
			(\A22-\A12)		&	(\B22+\B21)		\\
			(\A22+\A21)		&		\B11		\\[2ex]
			(\A11+\A22)		&	(\B11+\B22)		\\[2ex]
			(\A11+\A12)		&		\B22		\\
			(\A11-\A21)		&	(\B11+\B12)		\\
				\A22		&	(\B11-\B21)		
		\end{array}\right]
	\]*
	
	Strassen's ingredients, $\S1$--$\S7$, could have been defined
	as what are in the rightmost vector (for both signs and orbits).
	The rightmost vector equals, in terms of the current notation,
	$\bma{-\S3 & -\S7 & \S2 & \S1 & \S5 & -\S6 & -\S4}^⊤$.

\section{Dihedral Symmetry of Laderman}\label{app:dihedral}

	The following is the group action on $3$-by-$3$ matrices that respects Laderman.
	Declare matrices
	\[*P≔\bma{1&0&0\\0&0&1\\0&1&0},\qquad Q≔\bma{0&0&-1\\0&1&0\\-1&0&0}.\]*
	Then the following action is a rotation of order $4$
	\[*A⟼PB^⊤Q,\qquad B⟼QA^⊤,\qquad C⟼PC^⊤.\]*
	Moreover, the following action is a reflection of order $2$
	\[*A⟼-B^⊤,\qquad B⟼-A^⊤,\qquad C⟼C^⊤.\]*
	As usual, conjugating the rotation by the reflection
	is rotating in the opposite direction.
	
	The rotation rotates \cref{fig:shuriken,fig:sharingan} by $90^∘$ clockwise.
	It acts on $A$'s, $B$'s, $C$'s, and the indices of $L$'s by
	{\def\;{\mskip\thickmuskip\mathopen{}}%%%%%
	\begin{gather*}
		(\A21\;-\B33\;\A31\;-\B32)(\A22\;\B23\;\A32\;\B22)(\A23\;-\B13\;\A33\;-\B12) \\
		(\A31\;-\B32\;\A21\;-\B33)(\A32\;\B22\;\A22\;\B23)(\A33\;-\B12\;\A23\;-\B13) \\
		(\A11\;-\B31)(\A12\;\B21)(\A13\;-\B11)(\C12\C21\C13\C31)(\C22\C23\C33\C32)   \\
		(6\;14)(1\;3\;10\;11)(4\;16\;7\;12)(5\;17\;9\;13)(15\;2\;18\;8)(20\;21\;23\;22).
	\end{gather*}}%
	The reflection reflects \cref{fig:shuriken,fig:sharingan}
	along the matrix-transpose axis.
	It acts on $A$'s, $B$'s, $C$'s, and the indices of $L$'s by
	\[*(\A ij\;\mathopen{}-\B ji)(\C ij\C ji)
		(1\:3)(2\:5)(8\:9)(10\:11)(12\:16)(13\:18)(15\:17)(21\:22).\]*

\section{Minor Aspect of Laderman}\label{app:minor}

	In \cref{sec:matroid}, it can be seen that there are three versions
	of Laderman matroid---over $ℚ$, over $𝔽_2$, and over $𝔽_3$.
	It is not obvious if any of them is a directed graph.
	That said, we observe that some matroid minors turn out to be directed graphs.
	
	Consider linear combinations of $\C11$--$\C33$;
	we found the following combinations ``low-weight''
	(in terms of how many $L$'s are used):
	{\everymath{\displaystyle}%%%%%
	\PMS\fakeL{width("$L\series{00}$")}
	\def\L#1{\rlap{$L\series{#1}$}\kern\fakeL pt}
	\begin{align*}
		\C12-\C11	&	=\L4	+\L5	+\L1	+\L{15}	+\L{12}	-\L{19},	\\
		\C21-\C11	&	=\L{16}	+\L{17}	+\L3	+\L2	+\L4	-\L{19},	\\
		\C13-\C11	&	=\L7	+\L9	+\L{10}	+\L{18}	+\L{16}	-\L{19},	\\
		\C31-\C11	&	=\L{12}	+\L{13}	+\L{11}	+\L8	+\L7	-\L{19},	\\
		\C32-\C12+\C22		&	=\L2	+\L{20}	-\L1	+\L{22}	+\L{13},	\\
		\C22-\C21+\C23		&	=\L{18}	+\L{21}	-\L3	+\L{20}	+\L5,		\\
		\C23-\C13+\C33		&	=\L8	+\L{23}	-\L{10}	+\L{21}	+\L{17},	\\
		\C33-\C31+\C32		&	=\L{15}	+\L{22}	-\L{11}	+\L{23}	+\L9,		\\
		\C11		&	=\L6	+\L{14}	+\L{19},							\\
		\C22		&	=\L6	+\L4	+\L5	+\L2	+\L{20},			\\
		\C12-\C22	&	=\L{14}	+\L1	+\L{15}	-\L2	-\L{20}	+\L{12}.	
	\end{align*}}%
	
	Shortened at $\L6$ and $\L{14}$ and punctured at $\L{19}$, the first eight equations
	constitute a graph on eight vertices, as every variable appears at most twice.
	This graph can be directed by negating every other equation.
	\Cref{fig:sharingan} is the said directed graph;
	a vertex is identified with an equation if the former appears in
	the left-hand side of the latter, with the correct sign.
	
	Punctured at $\L{19}$, $\L2$ and $\L{20}$,
	the last ten equations constitute a directed graph on ten vertices.
	Note that the choice of $\L2$ and $\L{20}$ breaks the symmetry.
	That also means that there are eight ways to puncture,
	all yield the same structure up to isomorphism.

\begin{figure}
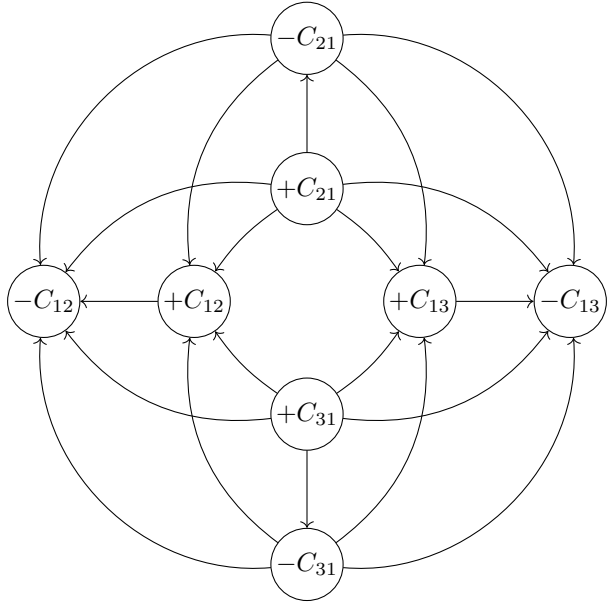

	\tikz{
		\draw[nodes={circle,draw,minimum size=2.5em,inner sep=1}]
			foreach\i in{2,3}{
				foreach\j in{+,-}{
					(0,{(2.5-\i)*(5-\j2)})node(C\i\j){$\j\C\i1$}
					({(\i-2.5)*(5-\j2)},0)node(C\j\i){$\j\C1\i$}
				}
			}
		;
		\draw[every edge/.append style={->}]
			foreach\i in{2,3}{
				(C\i+)edge(C\i-)(C+\i)edge(C-\i)
				foreach\j in{+,-}{
					foreach\k in{2,3}{
						foreach\l in{+,-}{
							(C\i\j)edge[bend left={(\i==\k?-1:1)*(30-\j10-\l10)}](C\l\k)
						}
					}
				}
			}
		;
	}
	\caption{
		A rank-$7$ minor of the Laderman matroid.
	}\label{fig:sharingan}
\end{figure}

\section{Charon Construction}\label{app:Charon}

	Over the entirety of the paper, we rely on auxiliary \emph{vectors}
	$g$ and $h$ to establish the parity-check equation $gCh=(gA)(Bh)$.
	A new idea in one sentence: use FMM in computing the checksum $(gA)⋆(Bh)$.
	We call this scheme \emph{Charon} after a moon of Pluto.
	
	A walking example goes below:
	Let $G∈𝔽^{2×4}$;
	let $A,B∈𝔸^{4×4}$;
	let $H∈𝔽^{4×2}$.
	Then $GCH=(GA)(BH)∈𝔸^{2×2}$ is worth (almost) four checksums.
	Normally, it would have cost $16$ workers to buy four checks.
	But in this case, $(GA)⋆(BH)$ is of type $⟨2,4,2;14⟩$, which saves two workers.
	We end up paying $49+14=63$ workers;
	note that this implies the code is wealthy.
	
	It can be shown that this Charon code is guaranteed to defeat two erasures.
	With empirical frequency $99.7\%$ it eliminates three;
	with empirical frequency $98.4\%$ it beats four.
	(See \cref{fig:4x4win}.)
	Charon for bigger matrices are expected to work as this example does.

\section*{Acknowledgement}

	We thank Chih-Yang Hsia (Google Inc.) for programming advices.

\hbadness9999\makeatletter\g@addto@macro\sloppy{\advance\baselineskip0ptplus1ptminus1pt}
\bibliographystyle{alphaurl}\bibliography{CheckStrassen-1}

\end{document}